\documentclass[11pt]{article}
\usepackage{amsmath,amssymb,amsthm,bbm,fullpage}
\usepackage{graphicx}
\usepackage{hyperref}
\usepackage{xcolor} 
\usepackage{subcaption}
\usepackage{tikz}
\usepackage{booktabs}
\usepackage{setspace}

\newtheorem{definition}{Definition}

\newtheorem{theorem}{Theorem}
\newtheorem{proposition}{Proposition}
\newtheorem{corollary}{Corollary}
\newtheorem{lemma}{Lemma}
\newtheorem{remark}{Remark}

\usepackage{enumitem}

\title{The Economics of Proof-of-Useful-Work}
\author{Rafael Pass\thanks{Pearl Research Labs. Rafael Pass is also a faculty member at 
Cornell Tech and Technion, and a member of the Cornell-Technion 
Jacobs Institute. He is also a faculty member at Tel Aviv University 
(currently on leave). This work was done while consulting for 
Pearl Research Labs.}}
\date{\today}

\begin{document}
\maketitle

\begin{abstract}
Proof-of-work (PoW) blockchains rely on computational expenditure to secure a ledger supporting a native cryptocurrency. In existing systems such as Bitcoin, this expenditure is intentionally useless: the computation secures consensus but produces no external economic output. 
An emerging alternative---proof of useful work (PoUW)---enables the same computation to simultaneously secure the blockchain and generate 
economically valuable output. However, PoUW is often criticized on economic grounds: if the work is useful, attackers might be ``paid to 
attack," potentially weakening security.

We develop a competitive-equilibrium model of a PoUW blockchain in 
which compute can be allocated across pure mining, pure useful work---instantiated as machine-learning inference---or ``duplex" work that 
produces both with computational overheads. We provide a complete 
closed-form characterization of equilibrium allocations and prices as 
a function of the duplex overheads and a single economic parameter---
the token–inference ratio--- measuring token adoption relative to the inference market. This characterization reveals three regimes: 
``Bitconia," in which the economy reduces to classical PoW; ``Fortessia," 
in which duplex replaces mining, increasing security while useful output 
remains unchanged; and ``Duplexia," in which token rewards subsidize 
inference, lowering prices and expanding inference supply.

Contrary to the common strawman argument, PoUW does not make attacks 
economically cheap: once equilibrium prices are taken into account, the 
economic cost of a majority attack remains tied to the block reward. 
Moreover, in Duplexia, block rewards act as rebates on inference prices, 
generating additional socially useful computation that would not arise 
without the blockchain---an expansion monotonically increasing in token 
adoption and technological efficiency.
\end{abstract}


\newpage
\tableofcontents
\thispagestyle{empty}
\clearpage
\setcounter{page}{1}
\section{Introduction}
Proof-of-work (PoW) blockchains rely on 
computational resource expenditure to secure a 
distributed ledger (the \emph{blockchain}) that 
supports a native token with a monetary value. This paradigm was introduced by Nakamoto \cite{Nakamoto2008} through the Bitcoin
protocol, whose security today is sustained by a vast amount of
specialized compute and whose total market capitalization is at a trillion dollar scale.

At a high level, so-called \emph{miners} maintain the blockchain by
solving computational puzzles known as proof of work (PoW).
These puzzles have three key properties: (i) \emph{difficulty}---finding 
a solution requires a prescribed expected amount of computation; 
(ii) \emph{unpredictability}---solutions cannot be found faster than 
by brute-force search; and (iii) \emph{verifiability}---a solution 
can be efficiently checked by anyone. In Bitcoin, the difficulty is 
continuously calibrated so that the expected time to find a solution 
equals 10 minutes given the total computational power participating 
in the protocol. These puzzles regulate block production and secure 
the protocol's ``consensus'' mechanism, ensuring both 
\emph{consistency} and \emph{liveness}.\footnote{Consistency 
informally means that honest participants observe a common history 
of the ledger, up to small delays due to network propagation, while 
liveness guarantees that valid transactions can continue to be 
included over time.} Together, these properties allow the protocol 
to function as a secure infrastructure for transacting in the native 
token. Additionally, formal analyses show that Nakamoto-style 
protocols achieve the above properties provided that honest parties 
control a majority of the computational power devoted to mining 
\cite{GarayKiayiasLeonardos2015,PassSeemanShelat2017}.

\paragraph{Does PoW need to be \emph{useless}?}
In Bitcoin, the work performed by miners is \emph{useless} in an 
economic sense: the computation satisfies the three properties above, 
yet produces no external economic output beyond maintaining consensus 
itself. This feature has attracted criticism as ``wasted computation''.
At the same time, others view it as a deliberate design criterion:
the \emph{waste} is precisely what makes attacks economically painful,
since an attacker must burn real resources to acquire ``security weight".

An emerging alternative paradigm is to rely on \emph{proof of useful work} (PoUW) (see e.g. \cite{King2013Primecoin, Ball2017UsefulPoW, KomargodskiWeinsteinPoUW})
in which the same computation simultaneously secures a blockchain and
produces economically valuable output. 
Indeed, the recent breakthrough construction of Komargodski and 
Weinstein~\cite{KomargodskiWeinsteinPoUW} realizes proofs of useful work for matrix 
multiplication---the workhorse of machine learning inference---and has motivated 
the launch of the Pearl Network~\cite{pearl2026}, a blockchain 
built on this construction.

The hope of PoUW is clear: less waste.
But PoUW is often met with a seemingly decisive strawman objection:
if the work is useful, then an attacker can be ``paid to attack'' by
selling the useful output, so the effective cost of
attacking collapses.
In the extreme, if useful work were perfectly reusable (what will we refer to as $\alpha=1$ below), one might conclude that attacks become essentially free!

\paragraph{Our Results in a Nutshell:}
This paper develops a competitive-equilibrium model that allows us to
formally investigate the above intuitions and, more broadly, study how
the \emph{design of work}---the computational overheads of PoUW relative
to pure useful computation and to pure security-generating computation
(i.e., mining)---shapes equilibrium allocations, pricing, and security.
Notably, we provide a complete closed-form characterization of equilibrium allocations and prices as a function of the computational overheads and a single economic parameter---the
\emph{token--inference ratio (T/I ratio)}---which measures token adoption relative to the size of the useful computation market.

As we show, contrary to the common strawman argument, PoUW does \emph{not} make attacks 
economically cheap: once equilibrium prices are taken into account, the economic cost of a majority attack remains tied to the block reward; under the same T/I ratio, the cost is at least as large as in a classical Bitcoin-style PoW system.
Moreover, when the computational overheads are small and token
adoption is sufficiently high, the presence of the PoUW blockchain
\emph{increases the supply of socially useful computation}. In this
regime, block rewards effectively act as rebates on useful computation,
generating additional useful work that would not arise in the absence
of the blockchain. Furthermore, the magnitude of this effect grows with
token adoption. 

\subsection{Modeling a Proof of Useful Work Economy}
The goal of our model is to develop a tractable 
framework that isolates the \emph{first-order 
economic forces} at play in a PoUW system. To 
this end, we consider an equilibrium model with a single 
time period (the \emph{blocktime}), a constant 
per-block reward (i.e., a fixed $R$ tokens are 
issued each block), and no transaction fees.
We show in Appendix~\ref{app:extensions} that the model readily extends 
to a multi-epoch setting with ``halving schedules" and transaction fees, and we defer the 
details to there.
As our goal is to focus on the economic 
aspects, we abstract away the details of the underlying 
blockchain protocol and simply assume a Nakamoto-style 
consensus protocol, which has been proven to implement a 
secure public ledger supporting a native token 
\cite{GarayKiayiasLeonardos2015, PassSeemanShelat2017} 
for any PoW puzzle satisfying the standard difficulty, 
unpredictability, and verifiability properties---whether 
useless or useful. The PoUW is assumed to simply replace 
the PoW puzzle in this protocol, with difficulty calibrated 
to ensure the appropriate blocktime. Finally, as we discuss 
in more detail below, we abstract away the implementation 
details of the PoUW scheme itself, summarizing its 
computational efficiency through the overhead parameters 
$\alpha$ and $\gamma$ (to be defined below).
\paragraph{The Compute Market}
We focus on a setting in which compute resources 
can produce economically valuable output; for 
concreteness, we take this useful computation to 
be machine-learning (ML) inference, though nothing 
in the analysis depends on this: any useful task 
with exogenous demand can be substituted. We 
assume a continuum of competitive agents (each 
modeled by a single GPU) that can freely enter 
and allocate their compute resources across three 
activities: \emph{pure mining} $M$ (a.k.a.\ 
solo-mine), \emph{pure inference} $L$ 
(a.k.a.\ solo-inference), and \emph{duplex} 
(a.k.a.\ ``hybrid'') \emph{compute} $D$, which 
simultaneously contributes to security and 
produces useful output (i.e., inference).

Each operation corresponds to a computation performed by a single GPU during a single time-period (a.k.a. the blocktime); for concreteness, think of the blocktime as being 10 minutes; we refer to this computation as a ``compute operation'', or ``compute op'' for short. Let $e>0$ denote the cost of a compute op (i.e., the total cost of running 1 GPU for the blocktime); for concreteness, think of $e$ as 1 USD. We note $e$ is the \emph{all-in cost of compute} which includes e.g., rental cost of a GPU-hour, inclusive of electricity and other operating costs. (In the main model in this paper, we assume a fixed cost $e$; however, as we show in Section \ref{sec:heterogeneous_costs}, our main theorems extend to a setting with heterogeneous cost, specified by an upward sloping supply function.)

\paragraph{Duplex Compute and its Overheads}
Each $M$ operation generates one unit of 
``security'' (i.e., the ``hashrate'' of 1 GPU), 
and each $L$ operation generates one unit of 
inference. The duplex (a.k.a.\ ``hybrid'', or ``dual-use") 
operation $D$ performs both simultaneously, 
but at imperfect efficiency.
Formally, each duplex operation yields 
$1/\alpha$ units of inference and $1/\gamma$ 
units of security, where $\alpha \geq 1$ and 
$\gamma \geq 1$ capture the technological 
overheads (equivalently, $\alpha$ duplex 
operations are required to replicate the 
inference output of one $L$ operation, and 
$\gamma$ duplex operations are required to 
replicate the mining output of one $M$ 
operation). 
Think of $\alpha$ and $\gamma$ as 
being on the order of $1.3$: for the price of 
one compute op, duplex delivers $\frac{1}{1.3} 
\approx 0.77$ units of inference and 
$\frac{1}{1.3} \approx 0.77$ units of security, 
for a combined output of $\approx 1.54$ units 
(i.e., we get ``$1.54$-for-$1$").

To ground this in a concrete PoUW construction, 
consider the matrix multiplication scheme 
of~\cite{KomargodskiWeinsteinPoUW}, which requires 
performing a perturbed version of matrix 
multiplication together with additional ``hashes'' 
to enable public verifiability. A pure inference 
provider can skip the perturbation and hashes 
entirely, saving roughly 23\% of computation. 
A pure miner can start from an all-zero matrix, 
also saving roughly 23\%, since the perturbed 
multiplication is slightly cheaper on trivial inputs. 
(The figures are purely illustrative; $\alpha = \gamma = 1.3$
corresponds to a saving of $1 - 1/1.3 \approx 23\%$.)

\paragraph{The Security and Inference Markets}
Total ``effective security" (a.k.a. ``total hashrate") is defined as:
\begin{equation*}\label{eq:security}
S := M + \frac{D}{\gamma},
\end{equation*}
$S$ corresponds to the amount of computational operations needed to control all mining in the blockchain, and thus the number of computational operations needed violate the honest majority assumption in a single time period is $S/2$.
Total ``inference supplied" during the time-period is
\begin{equation*}\label{eq:infsupply}
Q := L + \frac{D}{\alpha}.
\end{equation*}
The above operations effectively give rise to two separate markets, an \emph{inference market} and a \emph{security market}, that are connected through the duplex operation:
\begin{itemize} 
\item {\bf The inference market:} The inference 
market is specified by some exogenous 
\emph{inference demand} given by a 
downward-sloping curve $\mathcal{D}(p)$, where 
$p$ denotes the dollar price of 1 unit of 
inference (i.e., the inference provided by 1 GPU 
in 1 time-period). The inference price $p$ 
adjusts so that the total inference supplied, 
$Q = L + D/\alpha$, equals the quantity 
demanded at that price, i.e., $Q = \mathcal{D}(p)$; 
at this \emph{market-clearing} price, each 
solo-inference compute op yields revenue $p$.

\item {\bf The security market:} The native 
token trades at price $P$ (in USD) and each 
block mints $R > 0$ new tokens. We model 
block production deterministically, with 
exactly one block produced per time 
period.\footnote{In Bitcoin-style protocols 
blocks arrive stochastically as a Poisson 
process, with an average inter-arrival time 
of roughly ten minutes. Modeling block 
production deterministically is without loss 
of generality for risk-neutral agents: 
normalizing the time period to equal the 
expected inter-block time, a miner 
contributing $1/S$ of total hashrate earns 
an expected block reward of $PR/S$ per 
period under stochastic arrival, which 
coincides with the deterministic model.} 
In PoW blockchains, starting with 
Bitcoin~\cite{Nakamoto2008}, block rewards 
are shared proportionally among miners: each 
unit of mining earns a share of the block 
reward proportional to its contribution to 
total hashrate $S$, so one unit of pure 
mining earns PoW revenue $PR/S$.
\end{itemize}

\paragraph{Per-operation Profits and Competitive Equilibrium}
In total, and taking into account the cost $e$ of each compute op, we have that  per-operation profits (for each of $M,L,D)$ are:
\begin{equation}\label{eq:profits}
\pi_M=\frac{P R}{S}-e,\qquad
\pi_L=p-e,\qquad
\pi_D=\frac{P R}{\gamma S}+\frac{p}{\alpha}-e.
\end{equation}

We assume that agents enter competitively (\emph{free entry}). Competitive (general) equilibrium (see e.g. \cite{MasColellWhinstonGreen1995}) therefore requires that profits satisfy
\[
\pi_X \le 0 \quad \text{for all activities } X,
\]
with
\[
\pi_X = 0 \quad \text{for every active activity } X>0 .
\]
This captures the standard competitive-entry logic: if some activity generated strictly positive profit, either additional agents would enter that
activity or prices would adjust until profits are driven to zero.
\footnote{The zero-profit conditions implicitly 
assume a continuum of miners and inference 
providers, so that each individual agent 
controls a negligible fraction of total 
hashrate $S$ and inference supply $Q$ and 
therefore takes prices as given. In the limit 
of a continuum of agents, Cournot-style 
\cite{Cournot1838} strategic interactions 
among miners vanish, and Bertrand competition 
\cite{Bertrand1883} on the inference side 
pins $p = e$ whenever solo inference is active.}

\paragraph{Broader Economic Structure: Conjoint Production Markets}
At an abstract level, our model combines two classical market structures 
through a novel coupling mechanism. The inference market is a standard 
\emph{Bertrand competition}~\cite{Bertrand1883}: with free entry and 
homogeneous compute cost, inference providers compete on price driving it to production cost (in our case $e$). 
The mining market, on the other hand, is a standard \emph{Tullock contest}~\cite{tullock1980}: 
a fixed prize (the block reward $PR$) is shared proportionally among miners 
according to their hashrate contribution, and with a continuum of players 
the contest converges to a competitive outcome with ``full prize dissipation": 
$eS = PR$. 

What is novel is the \emph{duplex operation} $D$, a joint 
action that simultaneously produces outputs in both markets---inference and 
security---thereby coupling two markets that would otherwise be completely 
independent. We refer to such markets as \emph{conjoint production markets}. While we instantiate this framework in 
the context of PoUW blockchains, the underlying structure---two markets connected through a joint action that simultaneously produces outputs 
in both---may have broader economic implications beyond the blockchain setting, which we leave for future work.
\subsection{The Structure of Equilibrium and the Design of Work}
The key parameter governing the structure of equilibria is the
\emph{duplex overhead index}
\[
\Delta \;:=\; \frac1\gamma+\frac1\alpha-1 ,
\]
which captures the aggregate efficiency of duplex computation relative
to performing mining and inference separately.
A unit of duplex compute produces $1/\gamma$ units of security and
$1/\alpha$ units of inference. To replicate this output using solo 
activities would require $1/\gamma$ units of $M$ and $1/\alpha$ units 
of $L$, at a total cost of $(\frac{1}{\gamma} + \frac{1}{\alpha})$ 
compute ops. Since one duplex operation costs exactly one compute op, 
the sign of $\Delta = \frac{1}{\gamma} + \frac{1}{\alpha} - 1$ 
determines whether duplex is efficient relative to mixing:
\begin{itemize}
\item If $\Delta<0$, then $\frac{1}{\gamma} + \frac{1}{\alpha} < 1$, 
so duplex is strictly dominated by splitting compute between $M$ and 
$L$: one can achieve the same inference and security output at strictly 
lower cost. Thus only $M$ and $L$ can be active and the system 
effectively operates as a classic PoW/Bitcoin-style economy.
\item If $\Delta>0$, then $\frac{1}{\gamma} + \frac{1}{\alpha} > 1$, 
so duplex is strictly cheaper than replicating its output via $M$ and 
$L$ separately, and $D$ must always be active in equilibrium.
\end{itemize}

\subsection{The Main Theorem}
\label{sec:maintheorem.intro}
\begin{figure}[t]
\centering
\includegraphics[width=1\linewidth]{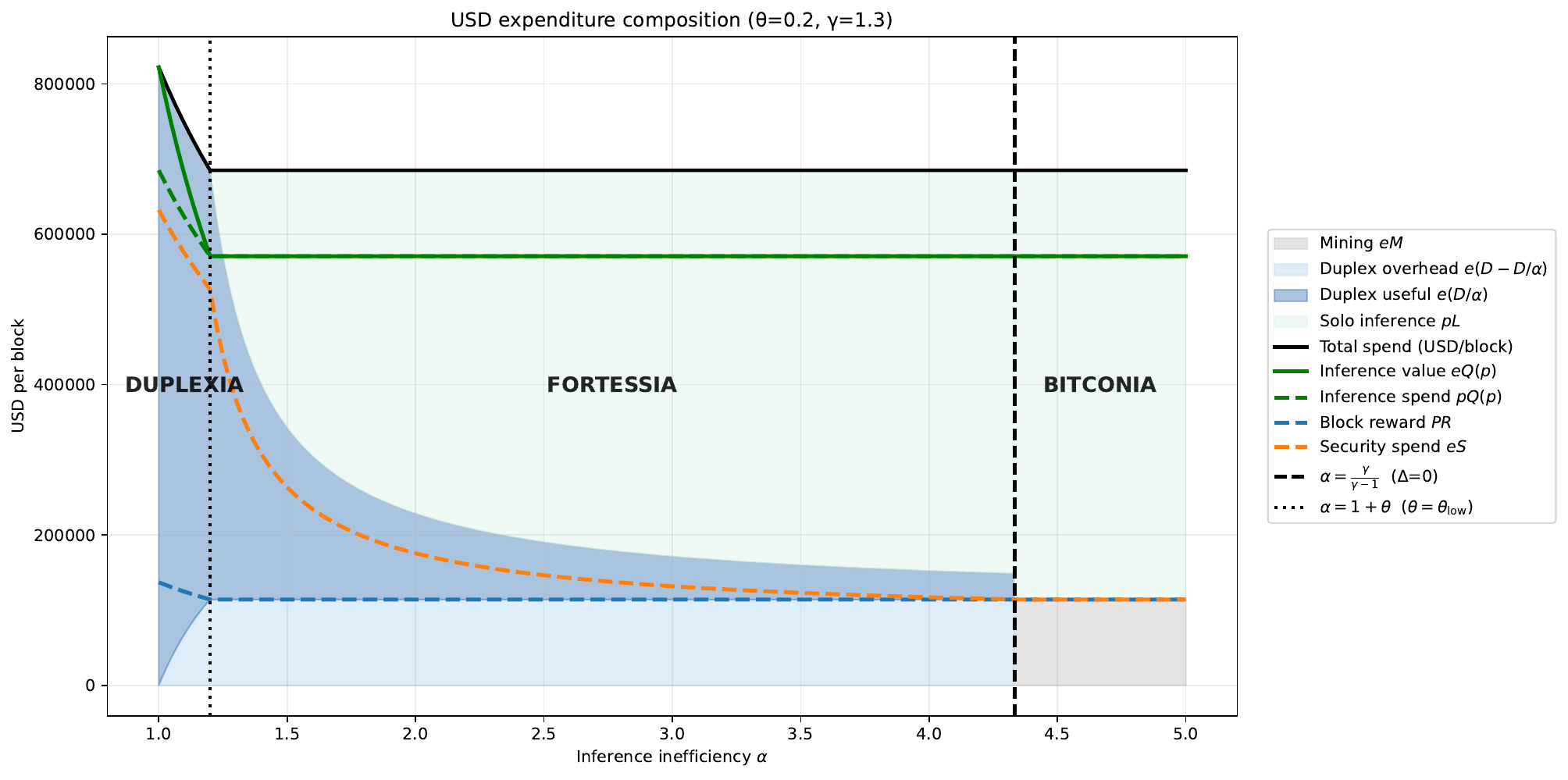}
\caption{Total spend in USD at duplex overhead
($\gamma = 1.3$), and its composition in terms of different activities. The labeled regions correspond to Bitconia, Fortessia, and Duplexia as
described in Theorem~\ref{thm:informal_full_transition}.
Note how in Duplexia, total amount spent on inference is lower than the value of this inference; this is due to the rebate on inference prices illustrated in Figure \ref{fig:rebate_low_gamma}.}
\label{fig:total_spend_low_gamma}
\end{figure}

\begin{figure}[t]
\centering
\includegraphics[width=1\linewidth]{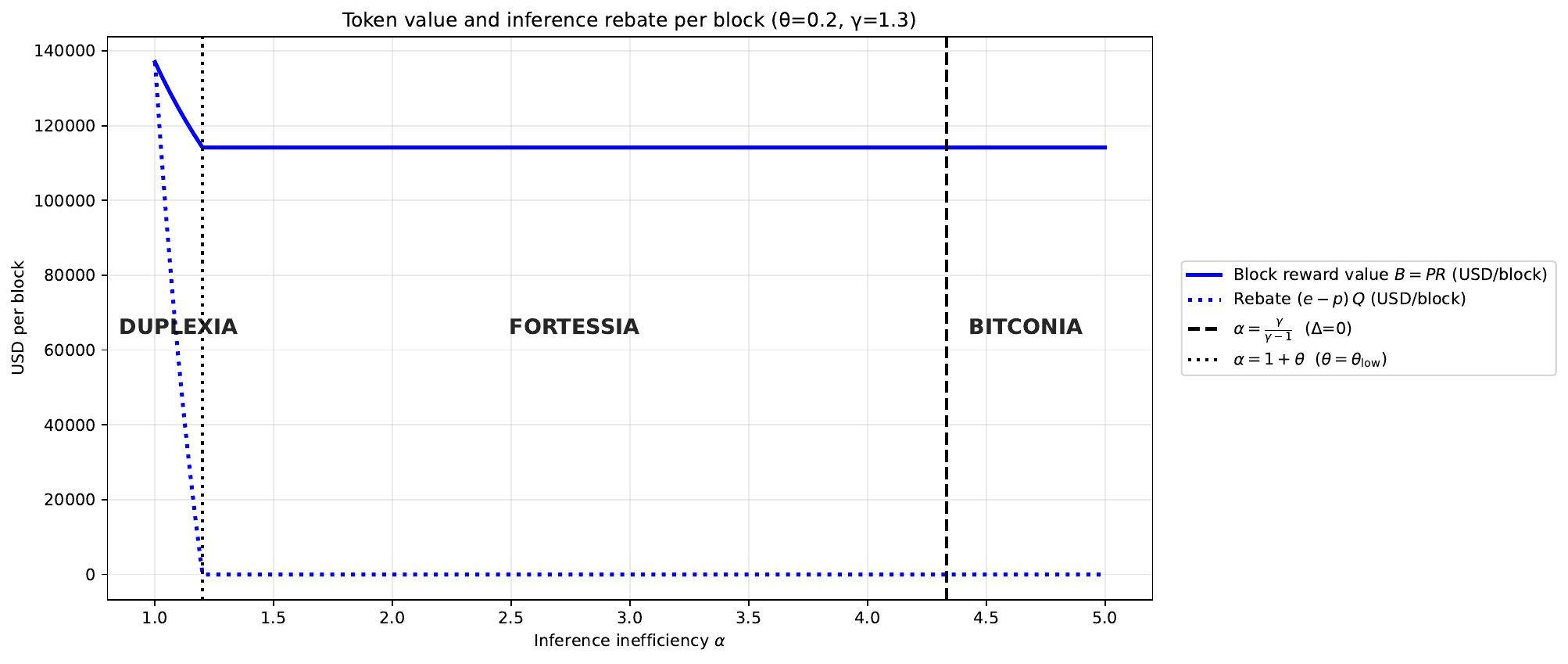}
\caption{Block-reward value and inference subsidy across regimes.
The subsidy emerges in Duplexia, where token rewards effectively reduce
the inference price, creating the feedback loop between usage, token
value, and security.}
\label{fig:rebate_low_gamma}
\end{figure}

In a competitive equilibrium, the token and 
inference markets need not be linked: the 
free-entry conditions alone do not pin down 
the relationship between the token price $P$ 
and the inference price $p$, and different 
equilibria may correspond to different 
relative magnitudes of token and inference 
market activity. Moreover, when $\Delta > 0$, 
multiple activity patterns are feasible---
$(D+L)$, $(D+M)$, and $(D$ only$)$---and the 
free-entry conditions alone do not select 
among them. To organize this multiplicity, 
we introduce a single scalar that captures 
the degree of economic linkage between the 
token economy and the inference market, and 
show that fixing this scalar uniquely pins 
down the equilibrium.

\paragraph{The Token/Inference Ratio (T/I ratio).}
We define the \emph{token--inference ratio 
(T/I ratio)}, denoted $\theta$, as the ratio 
between the dollar value of block rewards and 
the dollar value of inference transacted per 
block period:
\[
\theta = \frac{PR}{pQ}.
\]
In analogy with the price--earnings (P/E) 
ratio in equity valuation---where P/E 
summarizes the relationship between a firm's 
market value and its earnings as a single 
scalar, without requiring a theory of what 
determines either---the T/I ratio summarizes 
the degree of economic coupling between the 
token economy and the inference market. 

\paragraph{Main Characterization}
As we show, for any $\theta>0$, except for a knife-edge of 
technology parameters $\alpha,\gamma$, there is a \emph{unique} 
competitive equilibrium with T/I ratio $\theta$. As such T/I 
ratio is a natural metric that provides a \emph{one-dimensional 
parametrization} of the set of competitive equilibria: every 
equilibrium induces a value of $\theta$, and conversely, fixing 
$\theta$ uniquely pins down the equilibrium allocation and 
prices. Furthermore, as we discuss in Section~\ref{sec:monetary_intro} and 
in more detail in Appendix~\ref{sec:monetary_foundation}, under a simple 
monetary model, $\theta$ depends only on monetary primitives 
(settlement intensity, velocity, and hoarding behavior) and 
is independent of the compute technology parameters $\alpha$ 
and $\gamma$. This suggests that keeping $\theta$ fixed and 
varying the technological parameters is a natural methodology 
for evaluating the economic impact of improvements in duplex 
efficiency.

We assume throughout that inference demand is sufficiently elastic so that total spending $p\,\mathcal D(p)$ increases when
effective inference prices fall (elasticity greater than one), and ignore parameters that occur with measure 0. As we shall show, equilibrium structure is governed by two forces:
the technological efficiency of duplex compute $\Delta$ and the strength of cross-market linkage $\theta$: Technology determines whether duplex is viable at all; conditional on viability, the T/I ratio determines how much duplex activity reorganizes the economy.

\begin{theorem}[Equilibrium characterization via technology and T/I ratio, informal]
\label{thm:informal_full_transition}
Fix $(e,R)$ and a strictly downward-sloping inference demand curve
$Q=\mathcal D(p)$. For each value of the T/I ratio $\theta>0$, except for measure 0 set of values for $\Delta$ and $\theta$, 
the equilibrium is uniquely determined:
\begin{enumerate}
\item \textbf{Bitconia ($M{+}L$).}
If $\Delta<0$, the economy reduces to a Bitcoin-style world: mining and inference are fully decoupled:
\[
\begin{aligned}
M&>0, & D&=0, & L&>0, \\
p&=e, & Q&=\mathcal D(e), &
S&=M .
\end{aligned}
\]
Security is produced entirely by $M$, and inference by $L$. As a benchmark, let $Q^{\text{Bitconia}} = Q, \quad S^{\text{Bitconia}} = S$ denote the inference and security (i.e., ``hashrate") levels in such an economy.

\item \textbf{Fortessia ($D{+}L$).}
If $\Delta>0$ but $\theta<\theta_{\rm low}:=\alpha-1$, duplex replaces
solo mining while solo inference remains active:
\[
\begin{aligned}
M&=0, & D&>0, & L&>0, \\
p&=e, & Q&=Q^{\text{Bitconia}}, &
S&=\frac{D}{\gamma} > S^{\text{Bitconia}} .
\end{aligned}
\]
That is, useful output (i.e., inference) is unchanged relative to Bitconia, but security strictly increases (i.e.\ the security is \emph{fortified}).

\item \textbf{Duplexia ($D$ only, or $D{+}M$).}
If $\Delta>0$ and $\theta>\theta_{\rm low}$, solo inference disappears
and duplex becomes the only source of inference:
\[
\begin{aligned}
M&\ge 0, & D&>0, & L&=0, \\
p&<e, & Q&=\mathcal D(p) > Q^{\text{Bitconia}}, &
S&=M+\frac{D}{\gamma} > S^{\text{Bitconia}} .
\end{aligned}
\]
That is, token rewards effectively \emph{subsidize} useful work, lowering inference
prices and expanding demand. Additionally, as in Fortessia, security is also expanded w.r.t. Bitconia. In essence, we get \emph{duplex} (i.e., two-fold) gains.

Within this regime, the threshold $\theta_{\rm high}:=\frac{1}{\gamma-1}$ determines whether $M$ is active alongside duplex: if
$\theta<\theta_{\rm high}$, we obtain pure duplex ($D$ only), while if $\theta>\theta_{\rm high}$ pure mining also enters ($D{+}M$).
\end{enumerate}
Additionally, in all cases, we provide a simple closed-form expression of the equilibrium variables (as a function of technology parameters $\alpha,\gamma$, the T/I ratio $\theta$ and the exogenous demand function $\mathcal D$)---see Table \ref{tab:MLI_allocations:intro}---and in all cases 
\[
Q=\mathcal D(p),
\qquad
P=\frac{\theta\,p\,\mathcal D(p)}{R},
\]

\begin{table}[h!]
\centering
\caption{Equilibrium Prices and Allocations}
\label{tab:MLI_allocations:intro}
\renewcommand{\arraystretch}{1.6}
\begin{tabular}{|l|c|c|c|c|c|}
\hline
\textbf{Regime} 
& \textbf{Condition} 
& \textbf{Price $p$} 
& \textbf{$M$} 
& \textbf{$D$} 
& \textbf{$L$} \\ \hline

\textbf{M+L} 
& $\Delta < 0$ 
& $e$ 
& $\theta\,\mathcal D(e)$ 
& $0$ 
& $\mathcal D(e)$ 
\\ \hline

\textbf{D+L} 
& $\begin{matrix} \Delta > 0 \\ \theta < \theta_{\mathrm{low}} \end{matrix}$ 
& $e$ 
& $0$ 
& $\dfrac{\theta\alpha}{\alpha-1}\,\mathcal D(e)$ 
& $\Bigl(1-\dfrac{\theta}{\alpha-1}\Bigr)\mathcal D(e)$
\\ \hline

\textbf{D} 
& $\begin{matrix} \Delta > 0 \\ 
\theta \in (\theta_{\mathrm{low}},\theta_{\mathrm{high}}) \end{matrix}$ 
& $\dfrac{\alpha e}{1+\theta}$ 
& $0$ 
& $\alpha\,\mathcal D(p)$ 
& $0$ 
\\ \hline

\textbf{M+D} 
& $\begin{matrix} \Delta > 0 \\ 
\theta > \theta_{\mathrm{high}} \end{matrix}$ 
& $\alpha e\!\left(1-\dfrac{1}{\gamma}\right)$ 
& $\Bigl(\dfrac{\theta p}{e}-\dfrac{\alpha}
{\gamma}\Bigr)\mathcal D(p)$
& $\alpha\,\mathcal D(p)$ 
& $0$ 
\\ \hline
\end{tabular}
\end{table}
\end{theorem}
\paragraph{Three Worlds and the Duplex 
``Feedback Loop''}
The equilibrium regimes above correspond to 
three qualitatively different economic 
``worlds''.

\begin{itemize}

\item {\bf Bitconia.}
When duplex is technologically inefficient 
($\Delta < 0$), the economy reduces to a 
classical proof-of-work blockchain. Security 
is generated entirely by mining, while useful 
computation occurs independently at price 
$p = e$. No interaction exists between the 
inference economy and the token economy.

\item {\bf Fortessia.}
When duplex becomes technologically efficient 
($\Delta > 0$) but the token economy remains 
relatively small ($\theta < \theta_{\rm low}$), 
duplex replaces pure mining as the source of 
security. Inference prices remain pinned at 
$p = e$, so the volume of useful computation 
is unchanged relative to Bitconia. The system 
is, however, strictly more secure: the number 
of computational operations needed to violate 
the honest majority assumption has increased, 
and security strictly exceeds wasted compute, 
so security efficiency is strictly improved 
(thus ``Fortessia'': the security is 
\emph{fortified}).

\item {\bf Duplexia.}
When both the technology and the token economy 
are sufficiently strong ($\Delta > 0$ and 
$\theta > \theta_{\rm low} = \alpha-1$) (i.e., the overhead $\alpha$ is sufficiently small compared to the T/I ratio), solo inference 
disappears ($L = 0$): all inference is now 
produced through duplex compute, and the 
inference price falls strictly below the 
compute cost, $p < e$, in contrast to 
Fortessia where $p = e$. Token rewards 
effectively subsidize inference, expanding 
demand beyond what would arise in the absence 
of the blockchain. Under the assumption that 
larger inference spending leads to higher 
token value---which holds, for example, under 
a fixed T/I ratio $\theta$, an assumption we 
justify in Section~\ref{sec:monetary_intro}---this 
gives rise to a \emph{positive feedback 
loop}: higher inference demand raises token 
value, which in turn amplifies 
the subsidy.\footnote{We use the term 
``feedback loop'' in the equilibrium sense: 
these are not sequential causal steps but 
mutually consistent conditions that hold 
jointly in the unique equilibrium characterized 
in Theorem~\ref{thm:informal_full_transition}.} 
As a result, the magnitude of the subsidy---
and hence the expansion of inference 
supply---grows monotonically with token 
adoption $\theta$ 
(Corollary~\ref{cor:P_monotone_theta}). Duplexia 
thus simultaneously expands useful computation 
and increases security relative to the 
classical PoW baseline---gains on both 
economic and security dimensions, hence 
``Duplexia''. 

The threshold $\theta_{\rm low} 
= \alpha - 1$ has a precise economic meaning: 
$\alpha - 1$ measures the efficiency loss of 
duplex work relative to solo inference, and 
once the T/I ratio $\theta$ exceeds this 
overhead, token rewards are sufficient to 
compensate for the loss in inference 
efficiency, so duplex compute dominates solo 
inference in equilibrium.
\end{itemize}

\paragraph{Graphical illustration.}
To illustrate the transition between the different worlds, we consider an 
isoelastic demand function $\mathcal{D}(p) = Kp^{-\varepsilon}$---where 
$K > 0$ is a scale parameter and $\varepsilon > 0$ is the price 
elasticity of demand---with parameters $(K, \varepsilon) = (571{,}000, 2)$\footnote{The value $571{,}000$ corresponds to an inference market of \$30B 
across $N = 52{,}560$ ten-minute blocks per year. The value $\varepsilon = 2$ 
corresponds to an elastic market in which lower prices unlock new use cases, 
consistent with the Jevons effect~\cite{jevons1865coal,saunders1992khazzoom,sorrell2009rebound}.}; 
see Section~\ref{sec:graph_elasticity} for sensitivity analysis across values 
of~$\varepsilon$.
Figure~\ref{fig:total_spend_low_gamma} shows the equilibrium total spending composition as the duplex inefficiency parameter $\alpha$ varies (for fixed $(e,\gamma,\theta)$).
For large $\alpha$, duplex is dominated and the economy operates in
\emph{Bitconia}, where mining and inference coexist independently.
As $\alpha$ decreases, the system moves into \emph{Fortessia}, where
duplex replaces pure mining while useful output remains unchanged,
leading to higher security for the same economic workload.

Finally, when $\alpha$ becomes sufficiently small, the economy enters
\emph{Duplexia}: solo inference disappears and duplex rewards
effectively subsidize useful work.
As illustrated in Figure~\ref{fig:rebate_low_gamma}, this subsidy lowers
the effective inference price, expands demand, and increases token
value.
The resulting feedback loop---lower $p$ increasing spending $pQ$,
which raises block-reward value $PR$ and hence token price $P$
---explains the joint rise in economic activity and security observed in Duplexia.

\subsection{Economic security and the strawman argument}
In the analysis so far, when discussing security, we have focused on the
\emph{computational} level $S$ securing the blockchain.
This is a natural benchmark: indeed, as mentioned, classical analyses of Nakamoto consensus \cite{GarayKiayiasLeonardos2015,PassSeemanShelat2017,GansHalaburda2024} show that the protocol remains secure provided
an attacker controls strictly less than $50\%$ of total mining power, so $S/2$ computational operations are needed to violate the security of the blockchain.

More recent works have instead considered a notion of
\emph{economic security}
\cite{budish2025,LeshnoPassShiEconomicSecurityRails},
which measures the \emph{economic cost} required to mount a majority attack.
A frequent criticism of proof-of-useful-work (PoUW) systems is that
useful computation may reduce this cost: because the work produces
valuable output, an attacker could in principle be ``paid to attack,''
potentially lowering the effective cost of acquiring security weight.
In the extreme limit $\alpha=1$, where duplex work produces useful output
without \emph{any} overhead, this heuristic suggests that attacks might even become
essentially costless.

Our equilibrium analysis shows that this intuition is 
flawed once considering prices obtained in equilibrium.
We define the economic cost of a $50\%$ attack as the 
minimum expenditure required to generate $S/2$ units of 
security, either via solo mining at cost $eS/2$, or via 
duplex compute at net cost $\frac{\gamma S}{2}(e - 
p/\alpha)$ (accounting for inference revenue earned 
during the attack):
\[
\mathrm{Cost}_{50\%}
\;:=\;
\min\!\left\{
\frac{eS}{2},
\;
\frac{\gamma S}{2}\!\left(e-\frac{p}{\alpha}\right)
\right\}.
\]
This models either a liveness attack, or a consistency 
attack under the assumption that a successful 
double-spend causes the token to crash, rendering block 
rewards unredeemable. We adopt this definition because, 
as observed in 
\cite{budish2025,LeshnoPassShiEconomicSecurityRails,GansHalaburda2024}, 
if the attacker can also redeem block rewards at their 
full current USD price, the net cost of attacking is 
zero even for Bitcoin, since block rewards exactly offset 
the computational expenditure. Moreover, 
\cite{LeshnoPassShiEconomicSecurityRails} show that a variant of the 
Nakamoto protocol can guarantee consistency 
unconditionally, with the honest-majority assumption 
required only for liveness; our cost measure therefore 
captures the economic resources required to violate 
security even in such strengthened protocols.

As we show, in any competitive equilibrium, 
the equilibrium prices $p$ and $P$ are such 
that the cost as defined above is exactly $PR/2$, 
regardless of which regime the economy 
operates in.

\begin{theorem}[Economic security, informal]
\label{thm:economic_security_monotone}
In any competitive equilibrium, the economic cost of a
$50\%$ attack is proportional to the block-reward value:
\[
\mathrm{Cost}_{50\%} =\;\frac{PR}{2}.
\]
\end{theorem}

Theorem~\ref{thm:economic_security_monotone} captures a 
central economic insight: even when work is useful, once 
considering prices obtained in equilibrium, the cost of 
acquiring security weight is tied directly to the 
block-reward budget, just as in classical PoW systems. 
Furthermore, this result is robust to partial reward 
redemption: if the attacker can redeem a fraction $\nu 
\in [0,1]$ of block rewards, the attack cost becomes 
$(1-\nu)PR/2$, which remains identical to the 
corresponding cost for a classical Bitcoin-style PoW 
system (see Section~\ref{sec:attack_cost} for a formal 
derivation).

Moreover, assuming elastic demand and keeping a fixed 
T/I ratio, token rewards are weakly increasing as the 
duplex inefficiency parameter $\alpha$ decreases. 
Consequently, the economic cost of a $50\%$ attack is 
weakly higher in Fortessia and Duplexia than in Bitconia. 
The intuition differs across regimes. In Fortessia, 
duplex overhead replaces mining waste one-for-one, so 
the expenditure required to generate security remains 
unchanged; see Figure~\ref{fig:compute_low_gamma}. In 
Duplexia, duplex rewards subsidize useful computation, 
lowering the inference price and expanding demand. Under 
elastic demand this raises token value and therefore 
increases the economic cost of a $50\%$ attack---in 
essence, the attacker must not only bear the overhead of 
duplex computation, but also pay the inference subsidy, 
which together equal the token price.%

\subsection{Social Value of PoUW}
\label{sec:intro_social_value}

A central question in the economics of PoUW is to what extent PoUW
eliminates the ``useless computation" of Bitcoin-style POW mining. While,
intuitively, PoUW removes ``pure mining waste", it introduces two new sources of
inefficiency. First, duplex compute incurs overhead relative to solo
inference, which can be thought of as a form of burnt work. Second,
and more interestingly, in Duplexia, the block reward subsidy expands
the inference market, leading to
production of units that are assigned to consumers who value them
\emph{below their resource cost}---a so-called \emph{deadweight loss} \cite{harberger1964}. To analyze
these effects, we define the net social value of PoUW.

\begin{definition}[Net social value, informal]
Let
\[
  SV_{\rm inf}(e)
  := \int_0^{\mathcal{D}(e)} v_{\mathcal{D}}(q)\,dq - e\,\mathcal{D}(e)
\]
denote the social value of the standalone inference optimum: a world
in which inference is provided at cost $p = e$ with no
blockchain, where $v_{\mathcal{D}}(q)$ is the value to consumers of
the $q$-th unit of inference. The \emph{net social value} of PoUW is
\[
  \Delta SV
  := \underbrace{\int_0^{Q} v_{\mathcal{D}}(q)\,dq}_{\text{value of inference to consumers}}
  - \underbrace{e(M+D+L)}_{\text{resource cost of all compute}}
  + V_{\rm chain}
  \;-\; SV_{\rm inf}(e),
\]
where $Q$ is total inference supplied, $M+D+L$ is total compute, and
$V_{\rm chain}$ is the value of the blockchain.
\end{definition}

We focus on pure Duplexia, the empirically relevant Duplexia
subregime.\footnote{Mixed Duplexia requires $\theta > \theta_{\rm
high} = \frac{1}{\gamma-1}$, which for $\gamma = 1.3$ demands $\theta
> 3.33$, a token economy more than three times the size of the
inference market.}

\begin{theorem}[Informal; see Proposition~\ref{prop:social_value}
and Corollaries~\ref{cor:sv_duplexia}]
\label{thm:informal_sv}
In pure Duplexia ($\Delta > 0$, $\theta \in (\theta_{\rm
low},\theta_{\rm
high})$), the equilibrium inference price is $p =
\frac{\alpha}{1+\theta}e < e$ and the net social value is:
\[
  \Delta SV
  = V_{\rm chain}
  - \underbrace{
      \int_{\mathcal{D}(e)}^{\mathcal{D}(p)}
      \bigl(e - v_{\mathcal{D}}(q)\bigr)dq
    }_{\text{deadweight loss on expansion}}
  - \underbrace{
      e(\alpha-1)\,\mathcal{D}(p)
    }_{\text{duplex overhead}}.
\]
\end{theorem}

The deadweight loss measures the efficiency cost of the inference
subsidy: block rewards push the price below cost, inducing
consumption of units $q \in (\mathcal{D}(e), \mathcal{D}(p)]$ whose
consumer value $v_{\mathcal{D}}(q)$ falls short of their resource
cost $e$. The loss per unit is $e - v_{\mathcal{D}}(q) \in [0,e-p]$,
strictly less than the full resource cost $e$. 
Since $p = e$ maximizes standalone inference social value, this
deadweight loss is second order in the price discount $e - p$ by
the classical result of Harberger~\cite{harberger1964}, and hence
small when the subsidy is modest.
The duplex overhead $e(\alpha-1)\mathcal{D}(p)$ measures the cost
of the overhead of using duplex operations to provide the compute
needed for the inference demand $\mathcal{D}(p)$.

\begin{proposition}[Informal; see Corollary~\ref{cor:sv_solo}]
In Bitconia and Fortessia ($L > 0$), the equilibrium inference price $p = e$ and the net social value is:
\[
  \Delta SV = V_{\rm chain} - PR.
\]
\end{proposition}

Whether PoUW improves on Bitcoin depends on what is held fixed.
At the same block reward $PR$, Duplexia strictly dominates Bitcoin
(see Proposition~\ref{prop:social_value}): Bitcoin burns $PR$
entirely on mining without producing useful output, while Duplexia
spends the same $PR$ but also produces inference that consumers
value. 
The more appropriate comparison holds fixed the T/I ratio $\theta$,
consistent with how we compare equilibria throughout the paper. At the same $\theta$, Duplexia avoids
the mining waste $PR = \theta\,e\,\mathcal{D}(e)$ that Bitcoin burns
on mining, but introduces two residual costs: the deadweight
loss from the inference subsidy and the duplex overhead on the
expanded inference quantity. Whether these residual costs are smaller
than the avoided mining waste depends on the demand function and the
technology parameters---it is not guaranteed in general.

Under isoelastic demand $\mathcal{D}(p) = Kp^{-\varepsilon}$ with
$\varepsilon > 1$, we obtain a clean characterization: Duplexia
strictly dominates Bitcoin at the same $\theta$ if and only if
$\theta < \frac{1}{\varepsilon-1}$. The threshold
$\frac{1}{\varepsilon-1}$ decreases with elasticity: more elastic
demand means a larger inference expansion, amplifying both the
overhead and the deadweight loss on expanded compute, so Duplexia
dominates over a smaller range of $\theta$. For $\varepsilon = 2$,
the threshold is $\theta < 1$ --- Duplexia dominates whenever
inference market revenue exceeds block reward revenue. 

\subsection{A Simple Monetary Foundation for the T/I Ratio}
\label{sec:monetary_intro}

As we show in Appendix~\ref{sec:monetary_foundation}, 
the T/I ratio $\theta$ admits a natural monetary 
interpretation. In particular, we show that if
\begin{enumerate}
    \item {\em a fixed multiple $\theta_{\rm raw}$ of 
    inference spending is settled in the native token}, and
    \item {\em circulating tokens follow a simple 
    stock-flow process, with newly issued tokens 
    partially entering circulation and a fraction 
    periodically withdrawn into hoarding,}
\end{enumerate}
then any stationary equilibrium must satisfy a fixed 
T/I ratio $\theta$, which depends only on 
$\theta_{\rm raw}$ and the parameters of the monetary 
dynamics model, and is independent of the underlying 
compute technology parameters $\alpha$ and $\gamma$.

In more detail, let $pQ$ denote the dollar value of 
inference spending per block period, and suppose that 
a fraction $\theta_{\rm raw}$ of this spending is 
settled in the native token, so that the total dollar 
value of token-denominated transactions per block 
period equals
\[
\mathcal{E} = \theta_{\rm raw} \, pQ.
\]
The multiple $\theta_{\rm raw}$ may exceed one, 
allowing for token-denominated demand beyond inference 
settlement alone. We relate $\mathcal{E}$ to the 
token price $P$ via Fisher's classical equation of 
exchange~\cite{Fisher1911}: the total value of 
transactions equals the price level times the 
velocity of money times the money stock, i.e., 
$\mathcal{E} = PVX$, where $X$ is the stock of 
circulating tokens and $V$ is the average number 
of times each circulating token is used per block 
period (the ``velocity''). The circulating stock 
$X$ evolves according to a simple stock-flow 
process: each block period, a fraction $\nu \in 
(0,1]$ of newly issued tokens enters circulation 
(capturing the idea that miners sell a portion of 
their rewards), while a fraction $\lambda \in 
(0,1]$ of circulating tokens exits circulation 
into hoarding (capturing the tendency of token 
holders to withdraw tokens from active use). In 
steady state, inflows equal outflows, giving 
$\lambda X = \nu R$ and hence $X = (\nu/\lambda)R$. 
Substituting $X = (\nu/\lambda)R$ into the 
equation of exchange $\mathcal{E} = PVX$ gives
\[
\mathcal{E} = PV \cdot \frac{\nu}{\lambda} R
= \frac{\nu V}{\lambda} \cdot PR.
\]
Substituting $\mathcal{E} = \theta_{\rm raw} \, pQ$ 
and rearranging yields
\[
PR = \frac{\lambda}{\nu V} \cdot \mathcal{E} 
= \frac{\lambda}{\nu V} \cdot \theta_{\rm raw} \, pQ 
= \theta \, pQ,
\]
where 
\[
\theta := \frac{\lambda}{\nu V} \, \theta_{\rm raw}
\]
summarizes settlement intensity ($\theta_{\rm raw}$), 
hoarding behavior ($\lambda, \nu$), and velocity 
($V$) into a single index linking inference spending 
to token value. Crucially, $\theta$ is independent 
of the compute technology parameters $\alpha$ and 
$\gamma$.

Following Fisher's equation-of-exchange 
framework~\cite{Fisher1911},
and Friedman's 
restatement of the quantity 
theory~\cite{friedman1956}, 
we treat $V$, $\lambda$, 
and $\nu$ as stable empirical parameters rather than 
endogenous outcomes. Any model in which these respond 
to prices or adoption necessarily introduces 
behavioral assumptions---about expectation formation, 
hoarding decisions, or responses to inference 
subsidies---that are difficult to discipline 
empirically and inevitably more ad hoc than the 
stability assumption. This perspective 
therefore \emph{justifies} treating $\theta$ as a 
\emph{behavioral constant} when analyzing how 
changes in technology reshape equilibrium structure.

\subsection{The Structure of Competitive Equilibria without a fixed T/I Ratio}
Theorem~\ref{thm:informal_full_transition} characterizes equilibrium when
the T/I ratio $\theta$ is fixed, enabling an ``apples-to-apples''
comparison between PoUW equilibria and the classical PoW benchmark
(Bitconia). The same framework can also be used to understand the
structure of equilibria even \emph{without fixing the T/I ratio},
once we restrict attention to \emph{robust} equilibria---that is, equilibria with \emph{strict} complementarity; assuming $\Delta \neq 0$, this rules out the measure 0 set of parameters on which Theorem \ref{thm:informal_full_transition} did not apply.

\paragraph{Security Efficiency}
Without fixing the T/I ratio, absolute security levels cannot be
directly compared across equilibria: higher token rewards induce more mining activity. Instead, we here compare
the \emph{security efficiency} of equilibria, defined as the amount
of security produced relative to the amount of \emph{``wasted
computation"} $W$, defined as: 
\[
W := M + D\Bigl(1-\frac1\alpha\Bigr),
\]
i.e., the amount of computation that does not correspond to useful
inference. A key observation is that the difference between security and wasted
compute equals
\[
S-W = D\Delta.
\]
Thus whenever $\Delta>0$ and duplex compute is active, security
strictly exceeds wasted computation: part of the system's security is
obtained ``for free,'' without additional wasted compute.

\paragraph{The Structure of General Equilibria}
Combining the above observation (i.e., that robust equilibria exclude the measure 0 set of parameters where Theorem~\ref{thm:informal_full_transition} does not apply), Theorem~\ref{thm:informal_full_transition} yields the following classification of general competitive equilibria:

\begin{theorem}[Equilibrium regimes and security efficiency, informal]
\label{thm:informal_efficiency}
Fix primitives with $\Delta\neq 0$, and consider any non-degenerate robust
competitive equilibrium $(M,D,L,S,Q,p,P)$.
Then exactly one of the following holds:
\begin{enumerate}
\item \textbf{Case $\Delta<0$: Bitconia ($M+L$).} $M>0,\quad L>0,\quad D=0$ and
\[
p=e,\quad Q=\mathcal D(e),\quad S=M =W,\quad P=\frac{eS}{R}.
\]
\item \textbf{Case $\Delta>0$:} exactly one of the following holds.
\begin{enumerate}
\item[(a)] \textbf{Fortessia ($D{+}L$).} $M=0,\quad D>0,\quad L>0$ and
\[
p=e,\quad Q=\mathcal D(e),\quad S=\frac{D}{\gamma}>W,\quad
P=\frac{\gamma S}{R}\,e\Bigl(1-\frac1\alpha\Bigr).
\]

\item[(b)] \textbf{Duplexia ($L=0$).} $D>0,\quad L=0$, and exactly one of the following subcases holds.
\begin{enumerate}
\item[(b1)] \textbf{Pure Duplexia ($D$ only).} $M=0$ and
\[
p\in\Bigl(\alpha e\Bigl(1-\frac1\gamma\Bigr),\,e\Bigr),\qquad
P=\frac{\gamma S}{R}\Bigl(e-\frac{p}{\alpha}\Bigr)\;<\;\frac{eS}{R}.
\]

\item[(b2)] \textbf{Mixed Duplexia ($M{+}D$).} $M>0$ and
\[
p=\alpha e\Bigl(1-\frac1\gamma\Bigr),\qquad
P=\frac{eS}{R}.
\]
In both subcases,
\[
p<e,\qquad S > W,\qquad Q=\mathcal D(p) \geq D(e).
\]
\end{enumerate}
\end{enumerate}
\end{enumerate}
\end{theorem}
In summary, when duplex technology is efficient ($\Delta>0$), PoUW
improves security efficiency relative to classical PoW. In the
Fortessia regime, this manifests purely as more efficient security w.r.t. Bitconia, while in Duplexia the system additionally delivers lower inference prices and greater useful computation. However, without fixing the T/I ratio, we cannot determine in which of those regimes we will end up.

\subsection{Conclusions and Discussion}
\label{sec:discussion}

We have developed a competitive equilibrium model of a 
PoUW blockchain in which compute can be allocated 
across pure mining ($M$), pure inference ($L$), 
and duplex work ($D$) that produces both 
simultaneously. Our main contributions are as 
follows.

\paragraph{Three Regimes.} We provide a complete 
closed-form characterization of the competitive 
equilibrium, which exhibits three qualitatively 
distinct regimes determined by the duplex 
technology efficiency $\Delta$ and the T/I ratio 
$\theta$: \emph{Bitconia} ($\Delta \leq 0$; 
duplex is inefficient, the economy reduces to 
classical PoW with no duplex activity), 
\emph{Fortessia} ($\Delta > 0$, $\theta < 
\theta_{\rm low}$; duplex replaces pure mining 
$M = 0$, security compute increases while inference 
output is unchanged), and \emph{Duplexia} 
($\Delta > 0$, $\theta > \theta_{\rm low}$; 
duplex replaces pure inference $L = 0$, inference 
prices fall below marginal cost and, under strictly elastic demand, the inference 
market expands). The transition between regimes 
is governed by interplay between the duplex overheads and the T/I ratio $\theta$, which 
measures token adoption relative to the size of 
the inference market.

\paragraph{Economic Security.} Contrary to the common 
strawman argument that PoUW makes attacks 
economically cheap (since attackers are 
``paid'' to attack via useful work), the economic 
cost of a majority attack equals $PR/2$ in all 
regimes, identical to Bitcoin in Fortessia and 
at least as high in Duplexia. 

\paragraph{Expansion of Socially Useful Computation.}
In the Duplexia regime, block rewards subsidize inference prices, which under elastic demand,
expands the inference market beyond what would arise without the
blockchain: more inference is supplied at a lower price, which under elastic demand leads to an increase in the total market size. This
expansion comes with two social value costs --- deadweight loss from
the inference subsidy and duplex overhead --- that must be weighed
against the avoided mining waste. At the same block reward $PR$,
Duplexia strictly dominates Bitcoin in terms of net social value.
At the same T/I ratio $\theta$, Duplexia dominates Bitcoin if and
only if the avoided mining waste exceeds these two costs---which
under isoelastic demand with elasticity $\varepsilon > 1$ holds if
and only if $\theta < \frac{1}{\varepsilon-1}$.

\medskip
The above summarizes our formal results. We now 
expand on some of the economic intuitions behind 
them, and discuss qualitative implications and 
limitations that go beyond the model.

\paragraph{Are We Getting a ``Free Lunch''?}
At first glance, the inference expansion may appear paradoxical:
the PoUW blockchain produces additional useful computation
apparently ``for free.'' The funding comes from inflationary token
issuance (seigniorage)---the same mechanism that finances
proof-of-work security in Bitcoin, and that governments use when
they finance public expenditure through money creation. In a PoUW
system, the same issuance not only secures the blockchain but also
subsidizes inference, pushing the price below marginal cost.

But this is not a free lunch. Subsidies distort prices and generate
deadweight loss~\cite{harberger1964}: the inference price reduction
attracts consumption of units whose value to consumers falls short
of their resource cost. PoUW therefore replaces the pure mining
waste of Bitcoin with two smaller costs: deadweight loss from the
inference subsidy and duplex overhead. Whether this is a net
improvement depends on the parameters, as characterized formally
in Section~\ref{sec:sv_graphs}.

That said, our analysis is conservative in one important respect.
We benchmark against the competitive inference price $p = e$, but
in practice market frictions may cause the standalone inference
price to exceed the true resource cost $e_{\rm phys} \leq e$; for example, under monopoly pricing the classic Lerner formula~\cite{Lerner1934} implies a wedge between price $e$ and production cost. 
In
that case the inference market already underproduces relative to
the social optimum, and the PoUW subsidy is partially corrective
rather than purely distortionary. Under such conditions, the
blockchain is effectively removing a pre-existing market
inefficiency, and may generate strictly positive net social value
even accounting for the duplex overhead---and even if the
blockchain provides no transaction settlement value whatsoever
($V_{\rm chain} = 0$); see Remark~\ref{rem:ephys} for a more detailed discussion and analysis.

\paragraph{Monotonicity and Incentive Alignment.}
As we formally show in 
Corollaries~\ref{cor:P_monotone_alpha} 
and~\ref{cor:P_monotone_theta}, under elastic 
demand, equilibrium token value and cost of 
attack grow monotonically with token adoption 
$\theta$ and duplex efficiency $\Delta$. 
This creates natural incentive alignment: 
inference providers benefit from greater token 
adoption (lower inference prices, larger market), 
giving them a natural incentive to promote the 
PoUW token; and token holders benefit from 
improvements in duplex technology (higher token 
value), giving them a direct financial stake in 
funding such research. 

\paragraph{On the Role of Demand Elasticity.}
Our results above all rely on inference demand being elastic (i.e., total 
inference spending increases as prices fall). This should be expected, 
consistent with the Jevons 
effect~\cite{jevons1865coal,saunders1992khazzoom,sorrell2009rebound}, for AI 
inference and more broadly for useful-computation applications where lower 
prices unlock new use cases rather than merely substituting for existing ones. 
We note that the strength of the above effects grows with the degree of 
elasticity: higher elasticity amplifies both the expansion of inference supply 
and the appreciation of token value in Duplexia.

Notably, if demand becomes inelastic, we no longer obtain expansion in token 
price and inference spending---both shrink relative to the Bitconia 
baseline (see Section~\ref{sec:graph_elasticity} for more details). However, as we note (see Remark \ref{rem:inference_value_monotone}) even under inelastic demand, the blockchain still delivers an 
expansion of socially useful computation: \emph{more inference is supplied at a lower 
price}, even if the token and security amplification effects are absent.

\paragraph{ASICs versus GPUs and the Honest Majority Assumption.}
Bitcoin mining has given rise to ASICs---specialized hardware that
is $100$--$1000\times$ more energy efficient than GPUs for
Bitcoin's PoW, but has few productive uses outside mining. PoUW,
by contrast, is designed so that the same GPUs used for AI
inference are also used for consensus. Our formal security
analysis (Theorem~\ref{thm:attack_cost}) is a \emph{flow-cost}
analysis: it measures the one-period economic cost of acquiring
$S/2$ units of security, abstracting from stock effects such as
hardware resale values and token holdings. In practice these
stock effects matter. We discuss three dimensions along which
the two paradigms differ.

\begin{enumerate}

\item \emph{Hardware and token value after an attack.}
At first sight, a successful attack causes more value destruction
in ASIC-based PoW than in GPU-based PoUW: ASICs are
single-purpose and lose most of their value if the attacked chain
crashes, whereas GPUs retain outside value in inference,
training, and rendering. This suggests ASICs provide a stronger
hardware-destruction deterrent.

However, hardware value may not be the dominant stock loss.
The implied value of the global Bitcoin ASIC fleet is on the
order of tens of billions of dollars: Using a network hashrate
of roughly 1,000~EH/s~\cite{CoinWarz2026} and ASIC prices of
approximately \$20 per TH/s~\cite{SimpleMining2026} yields an
estimated fleet value of roughly \$20 billion. Miner-attributed
Bitcoin holdings are estimated at approximately \$50
billion~\cite{VanEck2026}---roughly two to three times larger.
A crash that destroys the token therefore destroys token
holdings that may exceed the hardware capital stock, and this
channel exists in both paradigms.

Even focusing only on hardware destruction, the comparison
between ASIC PoW and GPU PoUW is non-trivial. In Fortessia,
solo-inference GPUs outside the consensus pool retain their
full value after a crash, so ASIC PoW  dominates on
hardware destruction. In Duplexia, however, the blockchain
subsidy supports total GPU compute demand of
$\alpha\mathcal{D}(p)$ per period; after the token crashes,
the subsidy disappears and GPU demand falls to $\mathcal{D}(e)$.
If GPU capital values are proportional to the income they earn,
the fraction of PoUW-supported GPU capital value destroyed by
the attack is:
\[
\frac{\alpha\mathcal{D}(p)-\mathcal{D}(e)}{\alpha\mathcal{D}(p)}
= 1 - \frac{\rho^\varepsilon}{\alpha},
\]
where $\rho := p/e < 1$ and the second equality uses isoelastic
demand $\mathcal{D}(x) = Kx^{-\varepsilon}$. This fraction
depends on the duplex overhead $\alpha$ and the demand
elasticity $\varepsilon$: a larger subsidy (smaller $\rho$) or
more elastic demand (larger $\varepsilon$) implies greater GPU
value destruction. Since the global AI inference market is
potentially much larger than the Bitcoin ASIC market, the total
GPU capital loss can exceed the ASIC capital loss even when the
percentage loss is smaller.

\item \emph{Redirectability and market coverage.}
In Bitcoin, since ASICs have no use outside mining, essentially all existing
ASICs are deployed in equilibrium: a $50\%$ attack thus requires
controlling essentially $50\%$ of all ASICs in the world. In Fortessia, only
a fraction of GPUs participate in consensus, so an adversary
could redirect existing solo-inference GPUs without any new
acquisition. In Duplexia, all modeled inference demand is served
through duplex compute, so the consensus pool already covers the
subsidy-supported inference market---though GPUs used for
training, private inference, or other tasks remain outside.

\item \emph{Observability.}
Idle ASICs generate little observable economic activity, so a
manufacturer or state actor can produce and warehouse a large
stockpile invisibly, deploying it only at the moment of attack.
The honest majority assumption for Bitcoin is therefore
essentially unverifiable empirically. GPUs, by contrast, are
actively used for valuable economic activity: large-scale
diversion or hoarding carries a significant opportunity cost
and would disrupt ongoing commercial activity in ways that are
difficult to conceal, though not impossible.

\end{enumerate}

We leave a formal dynamic model of these stock effects for
future work.

\paragraph{Availability of Useful Work.}
Our model assumes that inference demand flows freely to duplex miners, just 
as it would to standard inference providers, at the market-clearing price 
$p$. In practice, however, attracting inference demand requires more than 
compute: a duplex miner must also acquire users with inference queries, 
which may require building infrastructure, establishing reputation, or 
developing distribution channels---not unlike any service provider seeking 
to attract customers. This friction is not captured by the overhead 
parameters $(\alpha, \gamma)$, and may in practice lead to aggregation of 
duplex miners around large inference providers, analogous to the emergence 
of mining pools in Bitcoin. We leave a formal treatment of such demand 
acquisition frictions for future work.

\subsection{Related Work}
\label{sec:related}
\paragraph{Economic security of proof-of-work.}
The economic security of proof-of-work blockchains has been 
studied by Budish \cite{Buddish2018}, Leshno, Pass, and Shi 
\cite{LeshnoPassShiEconomicSecurityRails}, and Gans and Halaburda 
\cite{GansHalaburda2024}. These works show that if attackers 
can redeem block rewards, the economic cost of attacking 
Bitcoin is zero, motivating our definition of attack cost. 

\paragraph{Token pricing and adoption.}
Athey et al.\ \cite{Athey2016BitcoinPricing} consider a transaction-demand 
perspective in which token value is derived from its use as a 
settlement medium---the same logic underlying our monetary 
foundation in Appendix~\ref{sec:monetary_foundation}. They observe experimentally that ``\emph{aggregate data is qualitatively consistent with the theory that prices are determined by current transaction volume and a fairly stable velocity}". Cong, Li, and 
Wang \cite{CongLiWang2021} develop a dynamic asset-pricing 
model in which token value derives from aggregating 
heterogeneous users' transactional demand on a platform. 

\paragraph{Cross-chain compute allocation.}
Bissias et al.\ \cite{bissias2019greedy,bissias2020pricing} 
study competitive equilibria in which miners allocate 
compute across multiple proof-of-work blockchains. This is the ``cross-chain analogue" of our free-entry 
competitive equilibrium within a single chain: where they 
ask how compute allocates across chains, we ask how it 
allocates across activities ($M$, $D$, $L$) within a 
chain. However, our setting is richer in a fundamental 
way: duplex compute $D$ simultaneously participates in 
both the security market and the inference market, 
creating a coupling between the two markets that has no 
analogue in settings where each unit of compute is 
allocated to exactly one activity. It is precisely this 
leakage through $D$ that generates the three-regime 
structure and the inference subsidy feedback loop in 
Duplexia.

\paragraph{Miner Equilibria and External-utility PoW.}
Fiat, Karlin, Koutsoupias, and Papadimitriou 
\cite{FiatKarlinKoutsoupiasP2019} study energy equilibria 
in proof-of-work mining, characterizing Nash equilibria 
under heterogeneous miner costs. 
Bar-On, Komargodski, and Weinstein 
\cite{BarOnKomargodskiWeinstein2025} extend the framework 
of \cite{FiatKarlinKoutsoupiasP2019} to incorporate 
external rewards (``coupons''), modeling PoUW as a setting 
where miners face subsidized costs due to useful work 
revenue. They characterize the resulting Nash equilibrium 
and study its implications for decentralization, showing 
that broadly distributed access to useful tasks increases 
rather than decreases network decentralization. 

Our model 
differs in two key respects: we adopt a competitive 
equilibrium framework with free entry rather than a 
finite-player Nash equilibrium, and we explicitly model 
the inference market and its price, allowing us to study 
how token rewards and inference prices co-determine 
equilibrium allocations. In particular, our 
Theorem~\ref{thm:attack_cost} shows that endogenous price 
adjustment---absent from fixed-price coupon 
models---is precisely what neutralizes the apparent 
attack-cost reduction from external rewards.

\paragraph{Concrete PoUW constructions.}
The idea of replacing useless PoW puzzles with useful 
computation has a long history. Primecoin \cite{King2013Primecoin} 
was among the first deployed cryptocurrencies to use a 
useful PoW puzzle, requiring miners to find Cunningham 
chains of prime numbers. Ball et al.\ \cite{Ball2017UsefulPoW} 
provided the first theoretical foundations for PoUW, 
constructing proofs of useful work from average-case 
hardness assumptions. Cao et al.\ \cite{cao2024optimization} 
propose an optimization-based PoUW scheme with accompanying 
security analysis. Oleksak et al.\ \cite{oleksak2025zksnark} 
propose a zk-SNARK marketplace built on PoUW, in which 
miners generate zero-knowledge proofs as useful work. 
Chatterjee et al.\ \cite{chatterjee2025scalowork} propose 
a PoUW scheme supporting distributed pool 
mining. 
These constructions face overheads from 
proof generation and domain-structure 
constraints on the useful computation. 
The most practically relevant instantiation 
of PoUW to date is the construction of Komargodski and 
Weinstein \cite{KomargodskiWeinsteinPoUW}, who show how 
to realize proofs of useful work for matrix multiplication, 
a fundamental primitive in modern machine-learning 
applications.
\section{Model and Equilibrium Definition}
\label{sec:model}

\subsection{Activities and Technology}
\label{activities.sec}
\paragraph{Activities and Compute Operations}
A continuum of competitive agents can freely enter and allocate compute across three activities:
\[
M\ge 0 \quad\text{(solo mining)},\qquad
D\ge 0 \quad\text{(duplex)},\qquad
L\ge 0 \quad\text{(solo inference)}.
\]
Each \emph{compute-op} has constant cost $e>0$ dollars (e.g., GPU rental + electricity cost). 1 compute op corresponds to the amount of computation by a single GPU in the blocktime (a.k.a. time period) for the underlying PoW blockchain. For concreteness, think of 1 compute op as the number of operations performed by a single GPU during 10 minutes (the Bitcoin blocktime), and corresponding to $e = \$1$ and LLM inference for roughly 1M tokens.

\paragraph{Outputs and Overheads}
A solo inference compute op produces one unit of inference and no security, while a solo mining compute op produces one unit of
security and no inference. A duplex compute op produce both inference and security but at reduced rates. We capture the reduced rate through two ``overhead'' parameters:
\begin{itemize}
\item \emph{Inference inefficiency} $\alpha \ge 1$: one duplex operation
produces $1/\alpha$ units of useful computation (inference).
\item \emph{Security inefficiency} $\gamma \ge 1$: one duplex operation
produces $1/\gamma$ units of PoW security.
\end{itemize}
\noindent Total ``effective security" (a.k.a. ``hashrate") is defined as:
\begin{equation}\label{eq:security}
S := M + \frac{D}{\gamma},
\end{equation}
and total inference supplied is
\begin{equation}\label{eq:infsupply}
Q := L + \frac{D}{\alpha}.
\end{equation}

\paragraph{The Design of Work and the Duplex Overhead Index}
It is useful to interpret $(\alpha,\gamma)$ geometrically.
A duplex compute op produces the output bundle
\[
\left(\frac{1}{\alpha},\,\frac{1}{\gamma}\right)
\]
in the two-dimensional space of inference and security. Any proof-of-work system can trivially emulate a PoUW technology by
simply mixing solo mining and solo inference: Suppose a fraction
$\lambda$ of compute is allocated to solo inference and the remaining
$1-\lambda$ to solo mining. One unit of compute then produces
\[
(\lambda,\,1-\lambda)
\]
units of useful computation and security respectively. This is
equivalent to a PoUW technology with
\[
\alpha = \frac{1}{\lambda}, 
\qquad
\gamma = \frac{1}{1-\lambda}.
\]
More generally, such simple mixing generates all $(\alpha,\gamma)$
pairs satisfying
\[
\frac{1}{\alpha}+\frac{1}{\gamma}=1.
\]
To summarize the joint efficiency of duplex computation relative to
simple mixing, we define the \emph{duplex overhead index}
\[
\Delta := \frac{1}{\gamma}+\frac{1}{\alpha}-1.
\]
The sign of $\Delta$ measures whether duplex computation produces
more or less combined output than can be obtained by simply splitting
compute between mining and useful computation, as illustrated in Figure \ref{fig:delta_frontier}.

\begin{figure}[t]
\centering
\begin{tikzpicture}[scale=5]

\fill[gray!15] (0,0) -- (0,1) -- (1,0) -- cycle;

\draw[->] (0,0) -- (1.06,0);
\draw[->] (0,0) -- (0,1.06);

\node at (0.52,-0.12) {useful computation $(1/\alpha)$};
\node[rotate=90] at (-0.12,0.52) {security $(1/\gamma)$};

\draw[thick] (0,1) -- (1,0);

\fill (0,1) circle (0.012);
\node[above left] at (0,1) {pure mining};

\fill (1,0) circle (0.012);
\node[below right] at (1,0) {pure inference};

\node[rotate=-45] at (0.52,0.56) {$\Delta = 0$};

\node at (0.80,0.78) {$\Delta > 0$};
\node at (0.80,0.71) {\small non-trivial PoUW};

\node at (0.40,0.28) {$\Delta < 0$};
\node at (0.40,0.21) {\small dominated by mixing};

\node[right] at (0.60,0.11) {\small trivial mixing frontier};

\end{tikzpicture}

\caption{
Duplex output space. A duplex operation producing
$(1/\alpha,1/\gamma)$ units of useful computation and security
corresponds to a point in the plane. The line $\Delta=0$,
equivalently $\frac{1}{\alpha}+\frac{1}{\gamma}=1$, is the
\emph{trivial mixing frontier}: every point on this frontier can be
obtained by splitting compute between pure inference and pure mining.
The shaded region satisfies $\Delta<0$ and is dominated by simple
mixing, while points above the frontier satisfy $\Delta>0$ and
represent non-trivial PoUW.
}
\label{fig:delta_frontier}
\end{figure}
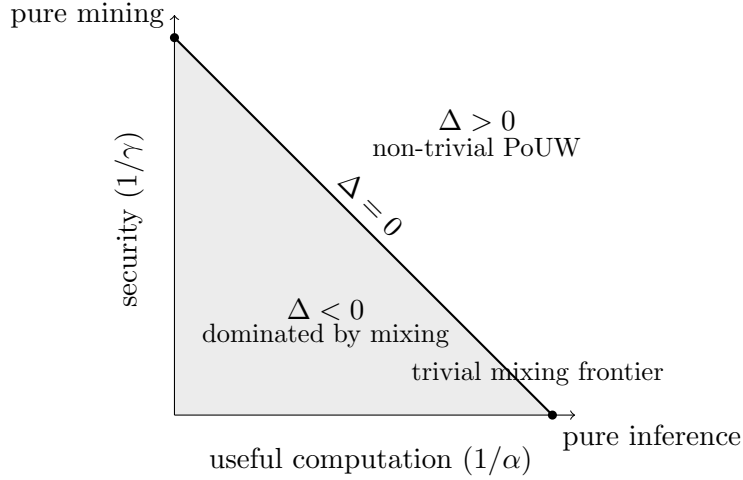
\paragraph{A Note on ``Primitive" PoUW technologies.}
At first sight one might think that new PoUW technologies can always be
constructed by mixing existing ones. Indeed, if $D_1$ and $D_2$ are two
PoUW activities producing inference and security outputs
$(x_1,y_1)$ and $(x_2,y_2)$, then any convex combination
\[
D^{\delta}=\delta D_1+(1-\delta)D_2
\]
produces the bundle
\[
(\delta x_1+(1-\delta)x_2,\;\delta y_1+(1-\delta)y_2),
\]
which could in principle be represented by new parameters
$(\alpha',\gamma')$ in the model. For example, taking $D_1=D$ and
$D_2=L$ yields a ``diluted'' duplex technology.

While such a construction is technically valid, it does not represent a
distinct \emph{primitive} technology in the economic environment modeled here.
In practice, miners are free to allocate their computational resources
across tasks, so any convex combination of activities can already be
implemented by splitting effort across the underlying tasks.
Consequently, introducing a diluted activity such as $D^{\delta}$
without also including the underlying activities $D_1$ and $D_2$ would
misrepresent the feasible choice set: rational miners could deviate and
run the underlying tasks directly.

For the activities $M$, $L$, and $D$ to meaningfully describe the mining
environment, they should therefore be interpreted as \emph{primitive}
or \emph{atomic} PoUW tasks. That is, the output bundle of the duplex
activity $D$ should not itself arise from a convex combination of other
feasible primitive tasks. 

\subsection{Inference market}
Inference sells at price $p\ge 0$ dollars per inference unit, where an inference unit is the amount of inference generated by 1 compute op (e.g., think of this as 1M LLM tokens). Demand is specified by a standard demand function:
\begin{equation}\label{eq:demand}
Q = \mathcal D(p)
\end{equation}
and market clearing requires
\begin{equation}\label{eq:inf_clearing}
L+\frac{D}{\alpha}=\mathcal D(p).
\end{equation}
where $Q$ corresponds to the inference demand for a single time-period (i.e., the blocktime, $\approx 10$ min) in inference units ($\approx 10M$ tokens).

\begin{definition}[Demand function]\label{def:demand}
A demand function is a map $\mathcal{D}:[0,\infty)\to[0,\infty)$ that is 
continuous and weakly decreasing, with $\mathcal{D}(p)$ finite for all 
$p\ge 0$.
\end{definition}
We thus treat the inference demand function as \emph{exogenous}: while the actual inference quantify supplied in equilibrium will depend on e.g., token adoption and prices endogenously, the \emph{mapping} from prices to quantities is assumed not to shift with token adoption or token price.

\subsection{Per-op profits}

In each time-period, each block mints a fixed quantity $R>0$ of tokens (in expectation), which trade at 
dollar price $P$. Each inference unit trades at dollar price $p$ (with total quantity $Q$ per time-period). Let $B:=PR$ denote the expected dollar value of the minted tokens per period. Under proportional PoW rewards, one unit of solo mining earns expected PoW revenue $B/S$. A duplex op contributes $1/\gamma$ security units and so earns $B/(\gamma S)$. Solo inference yields revenue $p$ per op; duplex yields $p/\alpha$ inference revenue per op.

Expected per-op profits are
\begin{equation}\label{eq:profits}
\pi_M=\frac{B}{S}-e,\qquad
\pi_L=p-e,\qquad
\pi_D=\frac{B}{\gamma S}+\frac{p}{\alpha}-e.
\end{equation}


\subsection{Equilibrium definition}
We are now ready to state our equilibrium definition. We use a standard general equilibrium definition (see e.g. \cite{MasColellWhinstonGreen1995}).

\begin{definition}[Outcome]
\label{def:outcome}
An \emph{outcome} is a tuple 
$(M, D, L, S, Q, p, P) \in \mathbb{R}_+^7$.
\end{definition}

\begin{definition}[Competitive equilibrium with free entry]\label{def:CE}
Fix primitives $(e,\alpha,\gamma,R)$ with
$e>0$, $\alpha\ge1$, $\gamma \ge 1$, $R>0$.
A competitive equilibrium is an outcome 
$(M,D,L,S,Q,p,P)$ such that:

\begin{enumerate}

\item \textbf{Accounting and market clearing:}
\[
S \;=\; M+\frac{D}{\gamma},
\qquad
Q \;=\; L+\frac{D}{\alpha},
\qquad
Q \;=\; \mathcal D(p).
\]

\item \textbf{Free-entry complementarity:}
with per-op profits
\[
\pi_M=\frac{B}{S}-e,\qquad
\pi_D=\frac{B}{\gamma S}+\frac{p}{\alpha}-e,\qquad
\pi_L=p-e,
\]
where $B:= P R$, we require for each $X\in\{M,D,L\}$,
\[
\begin{aligned}
X &\ge 0 && (\text{feasibility}),\\
\pi_X &\le 0 && (\text{free entry}),\\
X>0 &\Rightarrow \pi_X=0 && (\text{zero profit}) .
\end{aligned}
\]
\end{enumerate}
\end{definition}
\begin{remark}[Mixed vs.\ pure strategies]
Allowing players to randomize over activities does not change the set of
equilibrium outcomes. Since profits are deterministic, expected profit is
linear in the mixing probabilities, so any optimal mixed strategy assigns
positive probability only to profit-maximizing activities. Aggregate
outcomes depend only on the total mass of players performing each activity,
and any mixed allocation can therefore be represented by an equivalent
pure-strategy allocation.
\end{remark}
\paragraph{Non-degenerate equilibria:} Throughout the paper, we restrict attention to \emph{non-degenerate equilibria} where $Q > 0$ and $S>0$.

\subsection{Robust equilibria/Strict complementarity}
The complementarity condition in Definition \ref{def:CE} allows knife-edge equilibria in which an
inactive activity earns exactly zero profit. Such equilibria are not
robust to arbitrarily small perturbations. We therefore introduce a
refinement, analogous to Selten’s notion of
\emph{trembling-hand perfection}~\cite{Selten1975}.

\begin{definition}[Robust competitive equilibrium]
\label{def:robust_CE}
A \emph{robust} competitive equilibrium is a competitive equilibrium
such that for each activity $X\in\{M,D,L\}$,
\[
X=0 \;\Rightarrow\; \pi_X<0,
\qquad
X>0 \;\Rightarrow\; \pi_X=0 .
\]
\end{definition}
A simple but insightful observation is that unless solo inference is active, the inference price must be \emph{subsidized} in robust equilibria.
\begin{proposition}
[Inference subsidy without solo inference]
\label{lem:robust_subsidy}
In any robust equilibrium, if $L=0$ then $p<e$.
If $\mathcal D$ is strictly decreasing, this implies
\[
Q=\mathcal D(p)>\mathcal D(e).
\]
\end{proposition}

\begin{proof}
If $L=0$, robustness requires $\pi_L<0$.  
Since $\pi_L=p-e$, we obtain $p<e$.  
Strict monotonicity of $\mathcal D$ yields the strict inequality for $Q$.
\end{proof}


\subsection{Price pinning by activity and the Impossibility of Triple Activity}
We observe some simple implications of competitive entry. 
\begin{proposition}[Price pinning by active activities]
\label{prop:price_pinning}
In any competitive equilibrium:
\begin{enumerate}
\item \textbf{Inference pins the inference price.}
If $L>0$, then
\[
p=e .
\]

\item \textbf{Mining pins the token price.}
If $M>0$, then
\[
P=\frac{eS}{R}.
\]

\item \textbf{Duplex pins a price trade-off.}
If $D>0$, then token and inference prices satisfy
\[
P=\frac{\gamma S}{R}\Bigl(e-\frac{p}{\alpha}\Bigr).
\]
\end{enumerate}
\end{proposition}

\begin{proof}
Each statement follows from the complementarity conditions
$X \ge 0, \pi_X \le 0, X\pi_X = 0$ applied to the per-op profit functions:
\[
\pi_M=\frac{PR}{S}-e,\qquad
\pi_L=p-e,\qquad
\pi_D=\frac{PR}{\gamma S}+\frac{p}{\alpha}-e .
\]
If $L>0$, then the condition $\pi_L=0$ implies $p=e$, pinning the inference price to the marginal cost of compute. If $M>0$, then $\pi_M=0$ implies $PR/S=e$. Solving for $P$ yields $P=eS/R$, pinning the token price to the security-adjusted cost of mining.
If $D>0$, then $\pi_D=0$ implies:
\[ \frac{PR}{\gamma S} = e - \frac{p}{\alpha} \]
Multiplying both sides by the factor $\frac{\gamma S}{R}$ isolates the token price:
\[ P = \frac{\gamma S}{R}\Bigl(e-\frac{p}{\alpha}\Bigr). \]
\end{proof}

\noindent We finally observe that ``triple-activity" can only happen with $\Delta = 0$:
\begin{corollary}[Knife-edge triple activity]
\label{cor:no_interior}
There is no competitive equilibrium with $M>0$, $D>0$, and $L>0$
simultaneously unless $\Delta = 0$.
\end{corollary}
\begin{proof}
Suppose an equilibrium exists where all three activities are active ($M, D, L > 0$). By Proposition~\ref{prop:price_pinning}:
\begin{enumerate}
    \item $L > 0$ implies $p = e$.
    \item $M > 0$ implies $\frac{PR}{S} = e$.
    \item $D > 0$ implies $\frac{PR}{\gamma S} = e - \frac{p}{\alpha}$.
\end{enumerate}
Substituting the first two conditions ($p=e$ and $PR/S = e$) into the third condition yields:
\[
\frac{e}{\gamma} = e - \frac{e}{\alpha}.
\]
Dividing both sides by $e$ (since $e > 0$) gives:
\[
\frac{1}{\gamma} = 1 - \frac{1}{\alpha} \quad \implies \quad \frac{1}{\gamma} + \frac{1}{\alpha} - 1 = 0.
\]
This is precisely the condition $\Delta = 0$. 
\end{proof}

\subsection{The Security Efficiency of Duplex}
\label{secefficiency.sec}
In general, we cannot say anything about the level of security $S$ in equilibria in absolute terms. We observe, however, that in any equilibrium where Duplex is inactive (i.e., in an $M+L$ equilibrium), $S=M$, so all security comes from ``wasted computation.'' In contrast, when $D$ is active, we obtain additional security without any additional wasted computation.
Formally, let
\[
W:=M+D\Bigl(1-\frac1\alpha\Bigr)
\]
denote the amount of ``wasted'' (i.e., useless) computation.

\begin{proposition}[Duplex improves security efficiency]
\label{prop:SI_duplex}
In any non-degenerate competitive equilibrium,
$S-W = D\Delta$.
\end{proposition}

\begin{proof}
By definition of security,
\[
S = M + \frac{D}{\gamma},
\]
and by definition of wasted computation,
\[
W = M + D\Bigl(1-\frac1\alpha\Bigr).
\]
Subtracting gives
\[
S-W
=
D\Bigl(\frac1\gamma+\frac1\alpha-1\Bigr)
=
D\Delta.
\]
\end{proof}
In particular, as we show later in Theorem~\ref{thm:MLI_characterization}, 
away from the knife-edge $\Delta=0$, duplex is active only when $\Delta>0$. 
Combining this with the identity above yields
\[
D>0 \;\Longrightarrow\; S>W.
\]
\subsection{On the Economic Cost of $50\%$ Attacks.}
\label{sec:attack_cost}

A common criticism of blockchains based on proof-of-useful-work is that allowing
useful computation during mining allegedly reduces the economic cost of
attacking the network: an adversary could ``get paid'' for useful work
while simultaneously attempting to compromise security.  In the extreme limit $\alpha=1$, where duplex work produces useful output
without \emph{any} overhead, this heuristic suggests that attacks might even become
costless.

 We now formalize this and 
show that the criticism is unfounded \emph{once considering prices obtained in equilibrium}. We start by formalize the economic cost of a majority attack. An attacker wishing to acquire $S/2$ units of security can do so via two methods.

\medskip\noindent\textbf{Cost via solo mining.} Each solo 
mining operation produces one unit of security at cost 
$e$. To produce $S/2$ units of security, the attacker 
therefore requires $S/2$ mining operations, for a total 
expenditure of
\[
\frac{eS}{2}.
\]

\medskip\noindent\textbf{Cost via duplex compute.} Each 
duplex operation produces $1/\gamma$ units of security 
at cost $e$, but also generates $1/\alpha$ units of 
inference, earning inference revenue $p/\alpha$. The net 
cost per duplex operation is therefore $e - p/\alpha$. 
To produce $S/2$ units of security via duplex, the 
attacker requires $\gamma S/2$ duplex operations (since 
each contributes $1/\gamma$ units of security), for a 
total net expenditure of
\[
\frac{\gamma S}{2}\left(e - \frac{p}{\alpha}\right).
\]

\medskip\noindent The attacker optimally chooses the 
cheaper of the two methods. This motivates the following 
definition.

\begin{definition}[Cost of a $50\%$ attack]
\label{def:cost50}
Let $\Pi = (M,D,L,S,Q,p,P)$ be a competitive equilibrium.
The \emph{economic cost of a $50\%$ attack} is defined as the minimal
dollar expenditure required to generate $\tfrac12 S$ units of security
using the available security technologies (solo mining or duplex):
\[
\mathrm{Cost}_{50\%}(\Pi)
\;:=\;
\min\!\left\{
\frac{eS}{2},
\;
\frac{\gamma S}{2}\!\left(e-\frac{p}{\alpha}\right)
\right\}.
\]
The first term corresponds to producing $\tfrac12 S$ units of security
via solo mining, while the second term corresponds to producing the
same security via duplex compute, accounting for inference revenue
$p/\alpha$ earned per duplex unit.
\end{definition}
\noindent This definition does not account for any mining rewards obtained 
by the attacker, modeling either (i) a liveness attack, or 
(ii) a consistency attack under the assumption that a successful 
double-spend causes the token to crash, rendering block rewards 
unredeemable. 

We adopt this definition because, as observed in 
\cite{Buddish2018,LeshnoPassShiEconomicSecurityRails,GansHalaburda2024}, 
if attackers also receive block rewards that can be redeemed 
at their full current USD price, then the net cost of 
attacking is zero \emph{even for Bitcoin}, since block rewards 
exactly offset the computational expenditure. Additionally, 
\cite{LeshnoPassShiEconomicSecurityRails} show that a variant of the Nakamoto 
protocol can guarantee consistency \emph{unconditionally}, with the 
honest-majority assumption required only for liveness. Our 
cost measure therefore captures the economic resources 
required to violate security even in such strengthened protocols.

\begin{theorem}[Attack cost equals half the block-reward budget]
\label{thm:attack_cost}
In any non-degenerate competitive equilibrium,
the economic cost of a $50\%$ attack equals
\[
\mathrm{Cost}_{50\%}=\frac{PR}{2}.
\]
\end{theorem}

\begin{proof}
By Definition~\ref{def:cost50},
\[
\mathrm{Cost}_{50\%}
=
\min\!\left\{
e\cdot\frac{S}{2},
\;
\gamma\Bigl(e-\frac{p}{\alpha}\Bigr)\cdot\frac{S}{2}
\right\}.
\]
It therefore suffices to show that, in equilibrium,
\[
\min\!\left\{
e,
\;
\gamma\Bigl(e-\frac{p}{\alpha}\Bigr)
\right\}
=
\frac{PR}{S}.
\]

\medskip
If $M>0$, then mining is active and Proposition~\ref{prop:price_pinning}
implies  $\frac{PR}{S}=e$. Since $\pi_D\le 0$, we also have
$\frac{PR}{\gamma S}+\frac{p}{\alpha}-e\le 0$, or equivalently
$\gamma(e-\frac{p}{\alpha})\ge \frac{PR}{S}=e$. Hence
\[
\min\!\left\{e,\gamma\Bigl(e-\frac{p}{\alpha}\Bigr)\right\}=e=\frac{PR}{S}.
\]

\medskip
If $D>0$, then duplex is active and Proposition~\ref{prop:price_pinning}
implies $\frac{PR}{\gamma S}+\frac{p}{\alpha}=e$, hence
$\frac{PR}{S}=\gamma(e-\frac{p}{\alpha})$. Since $\pi_M\le 0$ we have
$\frac{PR}{S}-e\le 0$, i.e.\ $e\ge \frac{PR}{S}=\gamma(e-\frac{p}{\alpha})$.
Thus
\[
\min\!\left\{e,\gamma\Bigl(e-\frac{p}{\alpha}\Bigr)\right\}
=
\gamma\Bigl(e-\frac{p}{\alpha}\Bigr)
=
\frac{PR}{S}.
\]

\medskip
Since the equilibrium is non-degenerate, either $M>0$, or $D>0$ (or else $S = 0$), so the above cases are exhaustive.
\end{proof}

\begin{remark}[Robustness to partial reward redemption]
Theorem~\ref{thm:attack_cost} considers the cost of 
either a liveness attack or a consistency attack that 
causes the token value to crash, rendering block rewards 
unredeemable. We emphasize that the result is also 
robust to relaxing this assumption: if the attacker can 
redeem a fraction $\nu \in [0,1]$ of the block reward 
during the attack, then the net cost of acquiring $S/2$ 
units of security becomes
\[
\mathrm{Cost}_{50\%} = \min\left(\frac{eS}{2} - 
\nu\frac{PR}{2},\ \frac{\gamma S}{2}\left(e - 
\frac{p}{\alpha}\right) - \nu\frac{PR}{2}\right).
\]
Since $\nu PR/2$ is subtracted from both terms, we have
\[
\mathrm{Cost}_{50\%} = \min\left(\frac{eS}{2},\ 
\frac{\gamma S}{2}\left(e - 
\frac{p}{\alpha}\right)\right) - \nu\frac{PR}{2}
= \frac{PR}{2} - \nu\frac{PR}{2} = \frac{(1-\nu)PR}{2},
\]
where the second equality uses 
Theorem~\ref{thm:attack_cost} (applied with $\nu=0$), 
which establishes that the minimum of the two terms 
equals $PR/2$. This is exactly equal to the 
corresponding attack cost against a classical 
Bitcoin-style PoW system under the same assumption, 
since in that case $\mathrm{Cost}_{50\%}^{\mathrm{Bitcoin}} 
= eS/2 - \nu PR/2 = (1-\nu)PR/2$ by the same argument. 
Thus, for any fixed $\nu \in [0,1]$, PoUW neither 
increases nor decreases the economic cost of a majority 
attack relative to Bitcoin.
\end{remark}

\section{Equilibria Parameterized by the Token-Inference Ratio}
\label{sec:linkage}
\label{sec:equilibrium}
Our model features two interacting markets: an inference market,
with dollar price $p$ and quantity $Q$, and a token market,
with token price $P$ and per-block issuance $R$.
In general, competitive equilibria do not determine
the relationship between these two markets; different equilibria
may correspond to different relative magnitudes of inference spending and token valuation.

To organize the equilibrium set, we introduce a scalar measure
of ``cross-market linkage"---think of this as measure of adoption of the token.

\paragraph{The Token-inference Ratio (T/I ratio)}
For any non-degenerate competitive equilibrium with $pQ>0$ and $PR>0$,
define the \emph{token-inference ratio (T/I ratio)} by
\begin{equation}
\label{eq:MLI}
\theta \;=\; \frac{PR}{pQ}.
\end{equation}
This ratio measures the per-block dollar value
captured by token issuance relative to contemporaneous inference-market
expenditure. It summarizes the degree of linkage between the
inference market and the token market.

\begin{definition}[$\theta$-T/I ratio  Competitive Equilibrium]\label{def:CE:TI}
Fix primitives $(e,\alpha,\gamma,R,\theta)$ with
$e>0$, $\alpha\ge1$, $\gamma>1$, $R>0$, and $\theta>0$.
A $\theta$-T/I ratio competitive equilibrium is a competitive equilibrium $(M,D,L,S,Q,p,P)$ satisfying  Equation \ref{eq:MLI}. 
\end{definition}

\paragraph{T/I ratio uniquely pins equilibria}
A central technical result of this paper is that the T/I ratio
provides a complete one-dimensional parameterization of competitive
equilibria. Away from knife edges, for each admissible value of
$\theta$, there exists a \emph{unique} non-degenerate equilibrium, and
every such equilibrium corresponds to exactly one value of $\theta$.
Varying $\theta$ therefore traces the entire set of equilibria
and reveals how the economy transitions across distinct
regimes.

\paragraph{T/I Ratio as a Measure of ``Effective Market Adoption".}
As we show in Appendix~\ref{sec:monetary_foundation}, the T/I ratio $\theta$
admits a simple steady-state monetary foundation. In particular, we show that
if (i) a fixed multiple $\theta_{\rm raw}$ of inference spending is settled in
the native token, and (ii) circulating tokens follow a simple monetary
dynamics model, then any stationary
equilibrium must satisfy a fixed T/I ratio $\theta$.
Moreover, $\theta$ depends only on $\theta_{\rm raw}$ and the parameters of the
monetary dynamics model, and is independent of the underlying compute technology. If we think of $\theta_{\rm raw}$ as a measure of market adoption of the token, then $\theta$ becomes a measure of \emph{``effective" market adoption} (normalized by exogenous monetary parameters). 

\subsection{Equilibrium Pivots}
As we shall see, the equilibrium landscape is first and foremost determined by the  \emph{duplex overhead index} $\Delta$ (recall, $\Delta = \frac{1}{\gamma}+\frac{1}{\alpha}-1$).
Next, when $\Delta>0$, two threshold values of the T/I ratio determine the equilibrium regime:
\begin{equation}\label{eq:theta_thresholds_eqregime}
\theta_{\mathrm{low}} \;:=\; \alpha-1,
\qquad
\theta_{\mathrm{high}} \;:=\; \frac{1}{\gamma-1}.
\end{equation}

\begin{lemma}[Ordering of T/I ratio thresholds]
\label{lem:theta_order}
\[
\Delta>0
\quad\Longleftrightarrow\quad
\theta_{\mathrm{low}}<\theta_{\mathrm{high}}.
\]
\end{lemma}

\begin{proof}
\[
\Delta>0
\;\Longleftrightarrow\;
\frac{1}{\gamma}+\frac{1}{\alpha}-1>0
\;\Longleftrightarrow\;
\frac{1}{\alpha}>1-\frac{1}{\gamma}
=\frac{\gamma-1}{\gamma}
\;\Longleftrightarrow\;
\alpha<\frac{\gamma}{\gamma-1}.
\]
Subtracting $1$ from both sides yields
\[
\theta_{\mathrm{low}}=\alpha-1
\;<\;
\frac{\gamma}{\gamma-1}-1
=\frac{1}{\gamma-1}
=\theta_{\mathrm{high}} .
\]
\end{proof}

\subsection{Equilibrium Characterization via the T/I ratio}
\label{sec:equilibria_support}
We now characterize equilibria indexed by the T/I ratio. 

\begin{theorem}[Equilibrium Characterization via the T/I ratio]
\label{thm:MLI_characterization}
Fix primitives $(e,\alpha,\gamma,R)$ with
$e>0$, $\alpha\ge1$, $\gamma>1$, and $\mathcal D(e)>0$. 
For every $\theta>0$, away from the cases that $\Delta = 0$, or $\Delta > 0$ but $\theta\in \{\theta_{\mathrm{low}},\theta_{\mathrm{high}}\}$, there exists a unique non-degenerate $\theta$-T/I ratio competitive
equilibrium. The inference price $p$ and allocation are determined as follows:
\begin{enumerate}
\item \textbf{Low-efficiency technology ($\Delta<0$): Bitconia ($M{+}L$).}

The unique equilibrium has $M,L>0$ and $D=0$, with
\[
p=e.
\]

\item \textbf{High-efficiency technology ($\Delta>0$): Duplex always active ($D>0$).}

The regime is determined by $\theta$:
\begin{enumerate}
\item[\textup{(a)}] \textbf{Fortessia ($D{+}L$)}  
if $\theta < \theta_{\mathrm{low}}$: $p=e$.

\item[\textup{(b)}] \textbf{Pure Duplexia ($D$)}  
if $\theta_{\mathrm{low}} < \theta < \theta_{\mathrm{high}}$:
\[
p=\frac{\alpha}{1+\theta}e.
\]

\item[\textup{(c)}] \textbf{Mixed Duplexia ($M{+}D$)}  
if $\theta > \theta_{\mathrm{high}}$:
\[
p=\alpha e\Bigl(1-\frac{1}{\gamma}\Bigr).
\]
\end{enumerate}
\end{enumerate}
In all cases,
\[
Q=\mathcal D(p),
\qquad
P=\frac{\theta\,p\,\mathcal D(p)}{R},
\]
The full allocation $(M,D,L)$ and inference price $p$ in each regime is summarized in
Table~\ref{tab:MLI_allocations}. 
Furthermore, in all cases the equilibrium is robust.

Moreover:
\begin{itemize}
\item If $\Delta > 0$ and $\theta \in \{\theta_{\mathrm{low}},
\theta_{\mathrm{high}}\}$: there exists a unique competitive
equilibrium (the $(D)$ pattern), but it is non-robust.
\item If $\Delta = 0$: there exists a continuum of competitive
equilibria parametrized by $D \in [0, \min(\gamma\theta Q,
\alpha Q)]$.
\end{itemize}
\end{theorem}

\begin{table}[h!]
\centering
\caption{Equilibrium Prices and Allocations}
\label{tab:MLI_allocations}
\renewcommand{\arraystretch}{1.6}
\begin{tabular}{|l|c|c|c|c|c|}
\hline
\textbf{Regime} 
& \textbf{Condition} 
& \textbf{Price $p$} 
& \textbf{$M$} 
& \textbf{$D$} 
& \textbf{$L$} \\ \hline

\textbf{M+L} 
& $\Delta < 0$ 
& $e$ 
& $\theta\,\mathcal D(e)$ 
& $0$ 
& $\mathcal D(e)$ 
\\ \hline

\textbf{D+L} 
& $\begin{matrix} \Delta > 0 \\ \theta < \theta_{\mathrm{low}} \end{matrix}$ 
& $e$ 
& $0$ 
& $\dfrac{\theta\alpha}{\alpha-1}\,\mathcal D(e)$ 
& $\Bigl(1-\dfrac{\theta}{\alpha-1}\Bigr)\mathcal D(e)$
\\ \hline

\textbf{D} 
& $\begin{matrix} \Delta > 0 \\ 
\theta \in (\theta_{\mathrm{low}},\theta_{\mathrm{high}}) \end{matrix}$ 
& $\dfrac{\alpha e}{1+\theta}$ 
& $0$ 
& $\alpha\,\mathcal D(p)$ 
& $0$ 
\\ \hline

\textbf{M+D} 
& $\begin{matrix} \Delta > 0 \\ 
\theta > \theta_{\mathrm{high}} \end{matrix}$ 
& $\alpha e\!\left(1-\dfrac{1}{\gamma}\right)$ 
& $\Bigl(\dfrac{\theta p}{e}-\dfrac{\alpha}{\gamma}\Bigr)\mathcal D(p)$
& $\alpha\,\mathcal D(p)$ 
& $0$ 
\\ \hline

\end{tabular}
\end{table}
\begin{remark}[On dilution and the role of $\alpha$]
The transition into Duplexia occurs when $\theta>\theta_{\rm low}=\alpha-1$,
which notably does not depend on $\gamma$. At first glance, this may seem
surprising: one might think that any PoUW technology can be ``improved'' by
diluting it with solo inference, i.e., by forming a convex combination
$D^\delta=(1-\delta)D+\delta L$, which lowers $\alpha$ (while increasing
$\gamma$). Since the Duplexia threshold depends only on $\alpha$, this
suggests that one could induce an earlier transition into Duplexia by
considering such diluted technologies.

This intuition is misleading---it is only valid in a scenario when the original activity $D$ is not available. In reality, the original $D$ is available, and this activity strictly dominates 
$D^\delta$ whenever $D^\delta$ strictly dominates $L$: Since payoffs are linear in activity bundles, for any
prices we have
\[
\pi(D^\delta)=(1-\delta)\pi(D)+\delta\pi(L).
\]
In particular, if $D^\delta$ strictly dominates $L$ (i.e., $\pi(D^\delta)>\pi(L)$),
then necessarily $\pi(D)>\pi(L)$, which implies $\pi(D)>\pi(D^\delta)$.
Thus, whenever $D^\delta$ would be profitable relative to $L$ (i.e., we are in Duplexia), the original
duplex activity $D$ is strictly more profitable than $D^\delta$.

Therefore, in any environment where $D$ remains available, the diluted
activity $D^\delta$ cannot arise in a robust equilibrium. More generally, as
discussed in Section~\ref{activities.sec}, such non-primitive ``diluted'' technologies do
not affect the equilibrium analysis, as they can be replicated by agents
through mixtures of the original activities $D$ and $L$. The parameter
$\alpha$ should thus be interpreted as describing the overhead of an \emph{atomic}
technology, and the threshold
$\theta_{\rm low}=\alpha-1$ reflects a genuine economic condition rather
than an artifact of the modeling.
\end{remark}

\paragraph{Proof overview.}
Recall that our model combines two classical market structures through
the duplex action $D$: the security market is a Tullock contest
\cite{tullock1980} (a fixed prize $PR$ shared among miners
proportionally to their mining effort; free entry drives $PR/S = e$)
and the inference market is a Bertrand market \cite{Bertrand1883}
(free entry with homogeneous cost $e$; free entry drives $p = e$).
In the absence of the duplex operation, these two markets are
completely independent and can be analyzed in isolation. With the
duplex operation present, they become coupled: a duplex operator
simultaneously participates in both, earning revenue on each side.
In such a scenario, we can no longer analyze the two markets in
isolation (indeed, this is what led to the flawed strawman
``zero-cost attack argument''; see Section~\ref{sec:attack_cost}).

We note, however, that when a standalone activity (i.e., $M$ or $L$)
is active, it pins the revenue on its respective side of the market,
and the coupling can be broken:
\begin{itemize}
    \item $L > 0$ pins $p = e$, so the net cost of one duplex operation
    becomes $e - e/\alpha$ (since producing $1/\alpha$ units of inference
    earns revenue $e/\alpha$). Thus, $D$ can be thought of as a pure
    mining operation with cost $\gamma(e - e/\alpha) = \gamma e(1-1/\alpha)$
    per unit of security, and we can solve the standard Tullock
    competition with this reduced cost.

    \item $M > 0$ pins $PR/S = e$, so the net cost of one duplex
    operation becomes $e - e/\gamma$ (since producing $1/\gamma$ units
    of security earns revenue $e/\gamma$). Thus, $D$ can be thought of
    as a pure inference operation with cost $\alpha(e - e/\gamma) =
    \alpha e(1-1/\gamma)$ per unit of inference, and we can solve the
    standard Bertrand competition with this reduced cost.
\end{itemize}
When neither standalone activity is active (the (D) case), there is a
single action by definition, and the zero-profit condition for $D$
combined with the T/I condition gives a single equation in the single
unknown $p$.

\begin{proof}[Proof of Theorem~\ref{thm:MLI_characterization}]

In any non-degenerate competitive equilibrium, $S, Q > 0$, so at least
one of $M, D$ is active and at least one of $L, D$ is active. Moreover,
$M, D, L$ cannot all be simultaneously active: if $M, D, L > 0$ then
$\pi_M = 0$ gives $PR/S = e$, $\pi_L = 0$ gives $p = e$, and $\pi_D = 0$
gives $e/\gamma + e/\alpha - e = e\Delta = 0$, contradicting $\Delta \neq
0$. Hence exactly one of the following four activity patterns holds:
$$\text{(M+L)}, \quad \text{(D+L)}, \quad \text{(D)}, \quad \text{(M+D)}.$$

For each pattern, we proceed in three steps. In the \emph{zero-profit}
step, the conditions $\pi_Y = 0$ for each active $Y > 0$ uniquely
determine all equilibrium variables. In the \emph{free-entry} step,
for each inactive $Y = 0$, we derive a necessary and sufficient
condition on $(\Delta, \theta)$ for $\pi_Y \leq 0$, with strict
inequality away from knife-edges. In the \emph{feasibility} step,
for each activity $Y$ claimed to be strictly positive, we derive a
necessary and sufficient condition on $(\Delta, \theta)$ for $Y > 0$.
The pattern constitutes a competitive equilibrium if and only if all
free-entry and feasibility conditions hold; it is robust if and only
if all free-entry conditions hold with strict inequality.

\paragraph{Case (M+L).}

\emph{Zero-profit.} $\pi_M = 0 \iff PR/S = e$ and $\pi_L = 0 \iff
p = e$, uniquely determining: $p = e$, $Q = \mathcal{D}(e)$,
$P = \theta eQ/R$, $S = M = \theta Q$, $D = 0$.

\emph{Free-entry.} Substituting $PR/S = e$ and $p = e$:
$$\pi_D = \frac{PR}{\gamma S} + \frac{p}{\alpha} - e
= \frac{e}{\gamma} + \frac{e}{\alpha} - e = e\Delta \leq 0
\iff \Delta \leq 0,$$
with strict inequality iff $\Delta < 0$.

\emph{Feasibility.} $M = \theta Q > 0$ and $L = \mathcal{D}(e) > 0$
since $\theta, Q > 0$. No additional conditions.

To sum up, the (M+L) pattern constitutes a competitive equilibrium if and only if
$\Delta \leq 0$. Moreover, it is robust if and only if $\Delta < 0$.

\paragraph{Case (D+L).}

\emph{Zero-profit.} $\pi_L = 0 \iff p = e$. Substituting $p = e$
into $\pi_D = 0$:
$$\frac{PR}{\gamma S} + \frac{e}{\alpha} = e
\iff \frac{PR}{\gamma S} = e\!\left(1-\frac{1}{\alpha}\right)
\iff \frac{PR}{S} = \gamma e\!\left(1-\frac{1}{\alpha}\right).$$

With $M = 0$ and $S = D/\gamma$, substituting $PR = \theta eQ$:
$$\frac{\theta eQ}{D} = e\!\left(1-\frac{1}{\alpha}\right)$$

If \(\alpha=1\), the right-hand side is zero while the left-hand side
is strictly positive, since \(\theta>0\), \(e>0\), \(Q>0\), and
\(D>0\) in the \((D+L)\) case. Hence no \((D+L)\) equilibrium exists
when \(\alpha=1\).
Thus, in any \((D+L)\) equilibrium we must have \(\alpha>1\), and we
may solve:
$$ D = \frac{\theta\alpha Q}{\alpha-1}$$
Inference accounting $Q = L + D/\alpha$ gives:
$$L = Q\!\left(1 - \frac{\theta}{\alpha-1}\right).$$
This uniquely determines all variables: $p = e$, $Q = \mathcal{D}(e)$,
$D = \theta\alpha Q/(\alpha-1)$, $S = D/\gamma = \theta\alpha Q/(\gamma(\alpha-1))$,
$P = \theta eQ/R$, $M = 0$, and $L = Q(1-\theta/(\alpha-1))$.

\emph{Free-entry.} From the zero-profit step, $PR/S =
\gamma e(1-1/\alpha)$, so:
$$\pi_M = \frac{PR}{S} - e
= \gamma e\!\left(1-\frac{1}{\alpha}\right) - e
= e\!\left(\gamma - \frac{\gamma}{\alpha} - 1\right)
= -e\gamma\Delta \leq 0 \iff \Delta \geq 0,$$
with strict inequality iff $\Delta > 0$.

\emph{Feasibility.} $D = \theta\alpha Q/(\alpha-1) > 0$ since
$\theta, Q > 0$. $L = Q(1-\theta/(\alpha-1)) > 0$ iff:
$$\theta < \alpha - 1 =: \theta_{\mathrm{low}}.$$

To sum up, the (D+L) pattern constitutes a competitive equilibrium if and only if
$\Delta \geq 0$ and $\theta < \theta_{\mathrm{low}}$. Moreover, it is
robust if and only if $\Delta > 0$.

\paragraph{Case (M+D).}

\emph{Zero-profit.} $\pi_M = 0 \iff PR/S = e$. Substituting $PR/S =
e$ into $\pi_D = 0$:
$$\frac{e}{\gamma} + \frac{p}{\alpha} = e
\iff p = \alpha e\!\left(1-\frac{1}{\gamma}\right).$$
This uniquely determines $p$, and hence $Q = \mathcal{D}(p)$,
$D = \alpha Q$ (from $L = 0$ and $Q = D/\alpha$),
$P = \theta pQ/R$, $S = PR/e = \theta pQ/e$, and
$M = S - D/\gamma = \theta pQ/e - \alpha Q/\gamma$.

\emph{Free-entry.} From the zero-profit step, $p =
\alpha e(1-1/\gamma)$, so:
$$\pi_L = p - e
= \alpha e\!\left(1-\frac{1}{\gamma}\right) - e
= e\!\left(\alpha - \frac{\alpha}{\gamma} - 1\right)
= -e\alpha\Delta \leq 0 \iff \Delta \geq 0,$$
with strict inequality iff $\Delta > 0$.

\emph{Feasibility.} $D = \alpha Q > 0$ since $Q = \mathcal{D}(p) > 0$.
$M = \theta pQ/e - \alpha Q/\gamma > 0$ iff:
$$\frac{\theta p}{e} > \frac{\alpha}{\gamma}
\iff \theta\alpha\!\left(1-\frac{1}{\gamma}\right) > \frac{\alpha}{\gamma}
\iff \theta > \frac{1}{\gamma-1} =: \theta_{\mathrm{high}}.$$

To sum up, the (M+D) pattern constitutes a competitive equilibrium if and only if
$\Delta \geq 0$ and $\theta > \theta_{\mathrm{high}}$. Moreover, it is
robust if and only if $\Delta > 0$.

\paragraph{Case (D).}

\emph{Zero-profit.} With $M = L = 0$, inference accounting gives
$Q = D/\alpha$ and security accounting gives $S = D/\gamma$. The
condition $\pi_D = 0$ gives:
$$\frac{PR}{\gamma S} + \frac{p}{\alpha} - e = 0
\iff \frac{PR}{\gamma \cdot \alpha Q/\gamma} + \frac{p}{\alpha} = e
\iff \frac{PR}{\alpha Q} + \frac{p}{\alpha} = e
\iff \frac{PR + pQ}{\alpha Q} = e
\iff PR + pQ = e\alpha Q.$$
Substituting the T/I condition $PR = \theta pQ$:
$$p(1+\theta)Q = e\alpha Q \iff p = \frac{\alpha e}{1+\theta},$$
uniquely determining $p$, and hence $Q = \mathcal{D}(p)$,
$D = \alpha Q$, $S = \alpha Q/\gamma$, $P = \theta pQ/R$.

\emph{Free-entry.} For $\pi_M \leq 0$:
$$\pi_M = \frac{PR}{S} - e = \frac{\theta pQ}{\alpha Q/\gamma} - e
= \frac{\theta\gamma e}{1+\theta} - e
= e\!\left(\frac{\theta\gamma}{1+\theta} - 1\right) \leq 0
\iff \theta \leq \frac{1}{\gamma-1} = \theta_{\mathrm{high}},$$
with strict inequality iff $\theta < \theta_{\mathrm{high}}$.
For $\pi_L \leq 0$:
$$\pi_L = p - e = \frac{\alpha e}{1+\theta} - e
= e\!\left(\frac{\alpha}{1+\theta} - 1\right) \leq 0
\iff \theta \geq \alpha - 1 = \theta_{\mathrm{low}},$$
with strict inequality iff $\theta > \theta_{\mathrm{low}}$.

\emph{Feasibility.} $D = \alpha Q > 0$ since $Q > 0$. No further
conditions.

To sum up, the (D) pattern constitutes a competitive equilibrium if and only if
$\Delta \geq 0$ and $\theta_{\mathrm{low}} \leq \theta \leq
\theta_{\mathrm{high}}$. Moreover, it is robust if and only if
$\Delta > 0$ and $\theta \in (\theta_{\mathrm{low}},
\theta_{\mathrm{high}})$.

\paragraph{The conditions partition the parameter space.}
Summarizing, the four activity patterns constitute competitive equilibria
under the following conditions:
\[
\begin{array}{ll}
(M+L) & \Delta \le 0,\\
(D+L) & \Delta \ge 0 \text{ and } \theta < \theta_{\mathrm{low}},\\
(D)   & \Delta \ge 0 \text{ and }
        \theta_{\mathrm{low}}\le \theta \le \theta_{\mathrm{high}},\\
(M+D) & \Delta \ge 0 \text{ and } \theta > \theta_{\mathrm{high}}.
\end{array}
\]
By Lemma~\ref{lem:theta_order}, $\Delta > 0$ implies
$\theta_{\mathrm{low}} < \theta_{\mathrm{high}}$, so the four conditions
are mutually exclusive and exhaustive by inspection. Hence:
\begin{itemize}
\item \emph{If $\Delta < 0$, or $\Delta > 0$ and $\theta \notin
\{\theta_{\mathrm{low}}, \theta_{\mathrm{high}}\}$}: exactly one
pattern constitutes a competitive equilibrium, and it is robust.

\item \emph{If $\Delta > 0$ and $\theta = \theta_{\mathrm{low}}$}:
the unique competitive equilibrium is the $(D)$ pattern, but it is
non-robust: substituting $\theta = \theta_{\mathrm{low}} = \alpha - 1$
into the $(D)$ price formula gives $p = \alpha e/(1+\theta_{\mathrm{low}})
= e$, so $\pi_L = p - e = 0$ and robustness fails.

\item \emph{If $\Delta > 0$ and $\theta = \theta_{\mathrm{high}}$}:
the unique competitive equilibrium is the $(D)$ pattern, but it is
non-robust: substituting $\theta = \theta_{\mathrm{high}} =
1/(\gamma-1)$ into the $(D)$ mining profit formula gives
$\pi_M = \theta_{\mathrm{high}}\gamma e/(1+\theta_{\mathrm{high}}) - e
= 0$ and robustness fails.

\item \emph{If $\Delta = 0$}: for any $D \in [0,
\min(\gamma\theta Q, \alpha Q)]$ with $Q = \mathcal{D}(e)$, setting
$M = \theta Q - D/\gamma$ and $L = Q - D/\alpha$ gives a competitive
equilibrium with $p = e$, $PR/S = e$, and $\pi_M = \pi_L = \pi_D =
0$. Thus there is a continuum of competitive equilibria parametrized
by $D$.
\end{itemize}
\end{proof}

\section{Economic Security and Social Value}
\label{sec:security_social_value}
The equilibrium characterization of 
Section~\ref{sec:equilibrium} reveals how PoUW 
allocates compute across mining, inference, and 
duplex work. We now turn to the economic 
implications. We first introduce a standard 
elasticity condition and derive a key lemma 
controlling how inference spending responds to 
price changes (Section~\ref{sec:elastic}). We 
then use this to establish monotonicity of 
equilibrium prices, token value, and cost of 
attack in the technology parameter $\alpha$ and 
the T/I ratio $\theta$ 
(Section~\ref{sec:monotonicity}). Finally, we 
formally characterize the social value of PoUW 
and show that it strictly dominates Bitcoin 
(Section~\ref{sec:social_value}).

\subsection{Elasticity}
\label{sec:elastic}
A key tool in our analysis is a standard 
elasticity condition. Recall that the price elasticity 
of demand at a price point $p$ captures the 
relative change in quantity demanded per 
relative change in price, for an infinitesimal 
price change: an $x\%$ decrease in price leads to 
an $\varepsilon(p) \cdot x\%$ increase in quantity, 
for sufficiently small $x$.

\begin{definition}[Elastic demand]
\label{def:elastic_demand}
Let the demand function $\mathcal{D}:\mathbb{R}_{>0}
\to\mathbb{R}_{\ge 0}$ be differentiable and weakly
decreasing. Its \emph{(point) price elasticity} at $p>0$ 
is defined as:
\[
\varepsilon(p)\;:=\;-\frac{p\,\mathcal{D}'(p)}
{\mathcal{D}(p)}.
\]
We say that demand $\mathcal{D}$ is \emph{elastic} 
on an interval $I$ if $\varepsilon(p)\geq 1$ for 
all $p\in I$, and strictly elastic if $\varepsilon(p)>1$.
\end{definition}

To see why this definition captures the informal  description above, note that the relative change in quantity per relative change in price equals
\[
\frac{\Delta Q / Q}{\Delta p / p} 
= \frac{\Delta Q}{\Delta p} \cdot \frac{p}{Q}
\;\xrightarrow{\Delta p \to 0}\;
\mathcal{D}'(p) \cdot \frac{p}{\mathcal{D}(p)},
\]
which is negative since $\mathcal{D}'(p) \leq 0$. The minus sign in the definition flips the sign 
to give a non-negative elasticity 
$\varepsilon(p) \geq 0$. To build further 
intuition, consider two examples. If 
$\varepsilon(p) = 1$ (unit elastic), an $x\%$ 
decrease in price leads to an $x\%$ increase in 
quantity, for sufficiently small $x$, so total 
spending $p\,\mathcal{D}(p)$ remains unchanged. 

If $\varepsilon(p) > 1$ (strictly elastic), an $x\%$ 
decrease in price leads to a more than $x\%$ 
increase in quantity, for sufficiently small $x$, 
so total spending $p\,\mathcal{D}(p)$ rises as 
price falls. The following lemma formalizes this.

\begin{lemma}[Elasticity and spending]
\label{lem:elastic_spending}
If $\mathcal{D}$ is differentiable and strictly 
decreasing, then the total spending function 
$G(p) := p\,\mathcal{D}(p)$ satisfies
\[
G'(p) = \mathcal{D}(p)\!\left(1 - \varepsilon(p)
\right).
\]
In particular, if demand is elastic on $I$ then 
$G$ is weakly decreasing on $I$; and if demand 
is strictly elastic on $I$ then $G$ is strictly 
decreasing on $I$.
\end{lemma}

\begin{proof}
Differentiating $G(p) = p\,\mathcal{D}(p)$ and 
substituting the definition of $\varepsilon(p)$:
\[
G'(p) = \mathcal{D}(p) + p\,\mathcal{D}'(p) 
= \mathcal{D}(p)\!\left(1 - \varepsilon(p)\right).
\]
If demand is elastic, $\varepsilon(p) \geq 1$, so 
$G'(p) \leq 0$ and $G$ is weakly decreasing. If 
demand is strictly elastic, $\varepsilon(p) > 1$, so $G'(p) < 0$ and $G$ is strictly decreasing.
\end{proof}

For AI inference---and more broadly for 
useful-computation applications where lower prices 
unlock new use cases, demand is expected to be strictly elastic, 
consistent with the Jevons effect: improvements in 
efficiency tend to increase rather than decrease 
aggregate resource use, as lower prices make 
previously uneconomical applications 
viable~\cite{jevons1865coal,saunders1992khazzoom,sorrell2009rebound}. 
See Section~\ref{sec:graphs} for further discussion of elasticity.

\subsection{Monotonicity and Cost of Attack}
\label{sec:monotonicity}

We now establish that equilibrium prices, token 
value, and cost of attack are all monotone in 
the duplex technology parameters $(\alpha, \gamma)$ 
and the T/I ratio $\theta$. Intuitively, better 
duplex technology (lower $\alpha$ or $\gamma$) 
and higher token adoption (higher $\theta$) both 
drive inference prices lower, expanding the 
inference market. Monotonicity of equilibrium 
prices follows directly from the equilibrium 
characterization of 
Theorem~\ref{thm:MLI_characterization}. 
Monotonicity of token value and cost of attack 
additionally requires the elasticity condition 
of Section~\ref{sec:elastic}: it ensures that 
lower inference prices translate into higher 
total inference spending, which in turn drives 
higher token value and greater economic security.

\paragraph{Monotonicity in improved technology}
We now show that token and inferences prices are monotone in either $\alpha$ or $\theta$.
Towards this, we first show that inference prices do not jump once we pass regime changes.

\begin{lemma}[Agreement of regime formulas on the boundaries]
\label{lem:boundary_agreement}
Consider the inference-price formulas from
Theorem~\ref{thm:MLI_characterization}:
\[
p_{M+L}=e,\qquad
p_{D+L}=e,\qquad
p_D=\frac{\alpha}{1+\theta}e,\qquad
p_{M+D}=\alpha e\left(1-\frac1\gamma\right).
\]
These formulas agree on the relevant regime boundaries:
\[
\theta=\theta_{\rm low}=\alpha-1
\quad\Rightarrow\quad
p_{D+L}=p_D=e,
\]
\[
\theta=\theta_{\rm high}=\frac1{\gamma-1}
\quad\Rightarrow\quad
p_D=p_{M+D}.
\]
Moreover, on the technology boundary \(\Delta=0\),
\[
p_{M+L}=p_{D+L}=p_{M+D}=e.
\]
At the common point where \(\Delta=0\) and
\(\theta=\theta_{\rm low}=\theta_{\rm high}\), all four formulas agree.
\end{lemma}

\begin{proof}
If \(\theta=\theta_{\rm low}=\alpha-1\), then
\[
p_D=\frac{\alpha}{1+\theta}e
=\frac{\alpha}{1+\alpha-1}e=e=p_{D+L}.
\]
If \(\theta=\theta_{\rm high}=1/(\gamma-1)\), then
\[
p_D
=
\frac{\alpha}{1+1/(\gamma-1)}e
=
\alpha e\frac{\gamma-1}{\gamma}
=
\alpha e\left(1-\frac1\gamma\right)
=
p_{M+D}.
\]
Finally, if \(\Delta=0\), then
\[
\frac1\alpha+\frac1\gamma=1,
\]
so
\[
\alpha\left(1-\frac1\gamma\right)=1,
\]
and hence
\[
p_{M+D}=\alpha e\left(1-\frac1\gamma\right)=e
=p_{M+L}=p_{D+L}.
\]
At the common point \(\theta=\theta_{\rm low}=\theta_{\rm high}\), the
first boundary calculation also gives \(p_D=e\).
\end{proof}

We now show that, keeping the T/I ratio 
$\theta$ fixed, any improvement in duplex 
efficiency does not make things worse: 
decreases in either the inference overhead 
$\alpha$ or the security overhead $\gamma$ 
(weakly) lower the inference price and 
(weakly) raise the token price, and therefore 
(weakly) raise the economic cost of a $50\%$ 
attack (Theorem~\ref{thm:attack_cost}).

\begin{corollary}
\label{cor:P_monotone_alpha}
Fix $(e, R, \theta)$ and let $(\alpha_1, \gamma_1)$ 
and $(\alpha_2, \gamma_2)$ satisfy $\alpha_1 \leq 
\alpha_2$ and $\gamma_1 \leq \gamma_2$, with both 
parameter tuples away from the knife-edge cases of 
Theorem~\ref{thm:MLI_characterization}. For 
$i \in \{1,2\}$, let 
$\Pi(\alpha_i, \gamma_i) = (M_i, D_i, L_i, S_i, 
Q_i, p_i, P_i)$ denote the unique non-degenerate 
$\theta$-T/I ratio competitive equilibrium. Then 
$p_1 \leq p_2$. If furthermore demand is elastic, 
then
\[
P_1 \geq P_2 \qquad \text{and} \qquad 
\mathrm{cost}(\Pi(\alpha_1, \gamma_1)) \geq 
\mathrm{cost}(\Pi(\alpha_2, \gamma_2)).
\]
If $L_1 = 0$ (i.e., $\Pi_1$ is in Duplexia) and 
$\alpha_1 < \alpha_2$, then the above inequality 
for $p$ is strict, and if, additionally, demand 
is strictly elastic then the inequalities for $P$ 
and $\mathrm{cost}(\Pi(\alpha, \gamma))$ are 
strict as well.\footnote{As is clear from the 
proof, a decrease in $\gamma$ alone (with 
$\alpha_1 = \alpha_2$) does not affect the 
inference price in the $M{+}L$, $D{+}L$, or 
$D$ only regimes. In the $M{+}D$ regime, 
$p = \alpha e(1-1/\gamma)$ is strictly increasing 
in $\gamma$, so a strict decrease in $\gamma$ 
alone also yields strict monotonicity when 
$\Pi_1$ is in the $M{+}D$ regime.}
\end{corollary}

\begin{proof}
By Theorem~\ref{thm:MLI_characterization}, away 
from knife edges the equilibrium inference price 
is given by
\[
p =
\begin{cases}
e, & \text{in the } M{+}L \text{ and } D{+}L 
    \text{ regimes},\\[4pt]
\dfrac{\alpha}{1+\theta}e, & \text{in the ``} D 
    \text{ only'' regime},\\[6pt]
\alpha e\!\left(1-\dfrac{1}{\gamma}\right), 
    & \text{in the } M{+}D \text{ regime}.
\end{cases}
\]
Each of these expressions is weakly increasing 
in $\alpha$ and weakly increasing in $\gamma$, 
and the latter two are strictly increasing in 
$\alpha$. If the equilibrium changes regime as 
$(\alpha, \gamma)$ varies, the adjacent regime 
formulas agree on the boundary by 
Lemma~\ref{lem:boundary_agreement}. It follows 
that $p_1 \leq p_2$.

If $L_1 = 0$, then $\Pi(\alpha_1, \gamma_1)$ 
lies in the $D$ only or $M{+}D$ regime (by 
non-degeneracy), where the active price formula 
is strictly increasing in $\alpha$. Hence, since 
$\alpha_1 < \alpha_2$, we have $p_1 < p_2$.

Let $G(p) = p\,\mathcal{D}(p)$. By the 
$\theta$-T/I ratio,
\[
P_i = \frac{\theta}{R} G(p_i), \qquad 
i \in \{1,2\}.
\]
If demand is elastic, then by 
Lemma~\ref{lem:elastic_spending}, $G$ is weakly 
decreasing in $p$. Since $p_1 \leq p_2$, this 
implies $P_1 \geq P_2$. If demand is strictly 
elastic and $p_1 < p_2$, then $G$ is strictly 
decreasing and hence $P_1 > P_2$. Finally, by 
Theorem~\ref{thm:attack_cost}, the cost of a 
$50\%$ attack is proportional to the token 
price, so 
$\mathrm{cost}(\Pi(\alpha_1, \gamma_1)) \geq 
\mathrm{cost}(\Pi(\alpha_2, \gamma_2))$, with 
strict inequality when $P_1 > P_2$.
\end{proof}
\begin{remark}[Incentives for technological improvement.]
Corollary~\ref{cor:P_monotone_alpha} shows that 
in the Duplexia regime, under strictly elastic 
inference demand and a fixed T/I ratio, any 
improvement in duplex efficiency (lower $\alpha$) 
strictly raises the equilibrium token price. This 
suggests that token holders have an incentive to 
advance duplex technology, creating a natural 
incentive to fund and promote such research within 
the PoUW ecosystem. More broadly, even if the 
economy is currently in Fortessia, the incentive 
to improve technology remains: a sufficiently 
large improvement in $\alpha$ will push the 
economy into Duplexia, at which point the full 
benefits are realized.
\end{remark}

\paragraph{Monotonicity in token adoption}
Symmetrically, keeping the technology fixed, inference price decreases and 
token price grows with greater token adoption (i.e., when $\theta$ increases):

\begin{corollary}\label{cor:P_monotone_theta}
Fix $(e, \alpha, \gamma, R)$ with $\Delta > 0$ and let $\theta_1 <
\theta_2$, with both parameter tuples away from the knife-edge cases
of Theorem~\ref{thm:MLI_characterization}. For $i \in \{1,2\}$, let
$\Pi(\theta_i) = (M_i, D_i, L_i, S_i, Q_i, p_i, P_i)$ denote the
unique non-degenerate $\theta_i$-T/I ratio competitive equilibrium.
Then $p_1 \geq p_2$. If furthermore demand is elastic, then
\[
P_1 < P_2 \qquad\text{and}\qquad
\mathrm{cost}(\Pi(\theta_1)) < \mathrm{cost}(\Pi(\theta_2)).
\]

Moreover, \(p_1>p_2\) iff the interval \((\theta_1,\theta_2)\)
intersects the pure-Duplexia region \((\theta_{\rm low},\theta_{\rm high})\).
If both equilibria are in Fortessia or both are in mixed Duplexia,
then \(p_1=p_2\).
\end{corollary}

\begin{proof}
We consider the regime formulas for the equilibrium inference price given by
Theorem~\ref{thm:MLI_characterization}:
\[
p =
\begin{cases}
e, & \text{in the } M{+}L \text{ and } D{+}L \text{ regimes},\\[4pt]
\dfrac{\alpha}{1+\theta}e, & \text{in the ``} D \text{ only'' regime},\\[6pt]
\alpha e\!\left(1-\dfrac{1}{\gamma}\right), & \text{in the } M{+}D \text{ regime}.
\end{cases}
\]
Each expression is weakly decreasing in \(\theta\), and the expression
in the \(D\)-only regime is strictly decreasing. If the equilibrium changes
regime as \(\theta\) varies, the adjacent regime formulas agree on the
boundary by Lemma~\ref{lem:boundary_agreement}. Hence
\[
        p_1\ge p_2.
\]

For the strict price inequality, the interval
\((\theta_1,\theta_2)\) intersects the pure-Duplexia region
\((\theta_{\mathrm{low}},\theta_{\mathrm{high}})\) if and only if
\[
        \theta_1<\theta_{\mathrm{high}}
        \quad\text{and}\quad
        \theta_2>\theta_{\mathrm{low}}.
\]
Since \(p=\alpha e/(1+\theta)\) is strictly decreasing throughout
pure Duplexia, and \(p\) is continuous across regime boundaries by
Lemma~\ref{lem:boundary_agreement}, this condition is equivalent to
\[
        p_1>p_2.
\]
Conversely, if both equilibria are in Fortessia
\((\theta_1,\theta_2<\theta_{\mathrm{low}})\), then
\[
        p_1=p_2=e,
\]
and if both equilibria are in mixed Duplexia
\((\theta_1,\theta_2>\theta_{\mathrm{high}})\), then
\[
        p_1=p_2=\alpha e\left(1-\frac1\gamma\right).
\]

Let
\[
        G(p):=p\,\mathcal D(p).
\]
By the \(\theta\)-T/I ratio,
\[
        P_i=\frac{\theta_i}{R}G(p_i),
        \qquad i\in\{1,2\}.
\]
If demand is elastic, then by Lemma~\ref{lem:elastic_spending},
\(G\) is weakly decreasing in \(p\). Since \(p_1\ge p_2\), this gives
\[
        G(p_1)\le G(p_2).
\]
Moreover, \(G(p_1)>0\), since the equilibrium is non-degenerate and
\(p_1>0\). Therefore, using \(\theta_1<\theta_2\),
\[
        P_1
        =
        \frac{\theta_1}{R}G(p_1)
        <
        \frac{\theta_2}{R}G(p_1)
        \le
        \frac{\theta_2}{R}G(p_2)
        =
        P_2.
\]
Thus \(P_1<P_2\) under elastic demand.

Finally, by Theorem~\ref{thm:attack_cost},
\[
        \mathrm{Cost}_{50\%}(\Pi(\theta_i))
        =
        \frac{P_iR}{2}.
\]
Since \(R\) is fixed, \(P_1<P_2\) implies
\[
        \mathrm{Cost}_{50\%}(\Pi(\theta_1))
        <
        \mathrm{Cost}_{50\%}(\Pi(\theta_2)).
\]
\end{proof}

\begin{remark}[An adoption incentive]
Corollary~\ref{cor:P_monotone_theta} formalizes 
the adoption incentive informally described in 
Section~\ref{sec:discussion}. In the pure Duplexia 
regime, under strictly elastic inference demand, 
greater token adoption (higher $\theta$) 
\emph{strictly} lowers inference prices, 
expanding inference supply. This incentivizes 
inference providers to increase token adoption 
as a means to expand the inference market. Even 
in Fortessia, the incentive remains: a 
sufficiently large increase in $\theta$ will push 
the economy into Duplexia, at which point the 
full benefits are realized. In Bitconia, by 
contrast, the token and inference markets are 
fully decoupled and there is no adoption 
incentive.
\end{remark}

\begin{remark}[Inference value monotonicity without elasticity]
\label{rem:inference_value_monotone}
The monotonicity of inference price established in 
Corollaries~\ref{cor:P_monotone_alpha} 
and~\ref{cor:P_monotone_theta} holds without any 
elasticity assumption. As a consequence, the total 
\emph{value} of inference supplied---measured by the 
resource cost needed to provide it, 
$e\,\mathcal{D}(p)$---is weakly monotone in 
$\alpha$, $\gamma$, and $\theta$: as $\alpha$ or 
$\gamma$ decreases (better duplex technology), $p$ 
falls and hence $\mathcal{D}(p)$ weakly increases; 
and as $\theta$ increases (greater token adoption), 
$p$ falls and $\mathcal{D}(p)$ weakly increases 
similarly. In the Duplexia regime, where the 
inference price strictly falls, the expansion is 
strict whenever demand is strictly decreasing at 
the relevant price point. The elasticity condition 
is needed only for the stronger result that 
inference spending $p\,\mathcal{D}(p)$ also 
increases, which in turn drives monotonicity of 
token value $P$ and cost of attack $PR/2$.
\end{remark}


\subsection{Social Value}
\label{sec:social_value}

We define and compute the social value of a PoUW blockchain in three
steps. First, we define the social value of the standalone inference
market, which serves as our welfare benchmark. Second, we define the
social value of the joint PoUW market. Third, we compare the two to
obtain the \emph{net} social value of PoUW.

Beyond inference, the blockchain itself may generate value from
transaction settlement, smart contracts, censorship resistance, and
other applications; we denote this by $V_{\rm chain} \geq 0$.
Let $\mathcal{D}$ be the inference demand function. Write $v_{\mathcal{D}}$
for the associated inverse demand function:
\[
  v_{\mathcal{D}}(q) := \sup\{p \ge 0 : \mathcal{D}(p) \ge q\}.
\]
Thus $v_{\mathcal{D}}(q)$ is the marginal consumer value of the $q$-th
unit of inference. When $\mathcal{D}$ is continuous and strictly
decreasing, $v_{\mathcal{D}} = \mathcal{D}^{-1}$. When the demand
function is clear from context we write $v(q)$.

\begin{definition}[Frame]
\label{def:frame}
A \emph{frame} is a tuple
\[
  \Gamma = (\alpha, \gamma, R, \mathcal{D}, e, V_{\rm chain})
\]
consisting of technology parameters $(\alpha, \gamma)$, a block-reward
schedule $R$, an inference demand function $\mathcal{D}$, a compute cost
$e$, and a blockchain value $V_{\rm chain} \ge 0$.
\end{definition}

The social value of the standalone inference market is the total
value to consumers of the inference units produced, minus the
resource cost of producing them, in a world without a blockchain. Since the price $p$ is a transfer between
consumers and producers it does not affect social value directly;
only the quantity produced and its resource cost matter.

\begin{definition}[Social value of the standalone inference market]
\label{def:sv_inference}
Fix a frame $\Gamma = (\alpha, \gamma, R, \mathcal{D}, e, V_{\rm chain})$.
The \emph{social value of the standalone inference market} at price $p$ is
\[
  SV_{\rm inf}^{\Gamma}(p) :=
  \int_0^{\mathcal{D}(p)} v_{\mathcal{D}}(q)\,dq \;-\; e\,\mathcal{D}(p).
\]
(That is, the total value to consumers of the inference units produced,
minus the total resource cost of producing them.)
\end{definition}

The socially optimal standalone outcome sets price equal to marginal
cost, $p = e$, yielding
\[
  SV_{\rm inf}^{\Gamma}(e) =
  \int_0^{\mathcal{D}(e)} v_{\mathcal{D}}(q)\,dq \;-\; e\,\mathcal{D}(e).
\]
The social value of the joint PoUW market is the total value to
consumers of the inference units produced, plus the blockchain value
$V_{\rm chain}$, minus the full resource cost of all compute---$M+L+D$.
\begin{definition}[Social value of the joint market]
\label{def:social_value}
Fix a frame $\Gamma = (\alpha, \gamma, R, \mathcal{D}, e, V_{\rm chain})$
and an outcome $\pi = (M, D, L, S, Q, p, P)$. The \emph{social value} of
the joint PoUW market is
\[
  SV^{\Gamma}(\pi) :=
  \int_0^Q v_{\mathcal{D}}(q)\,dq + V_{\rm chain} - e(M + D + L).
\]
The \emph{net social value} of outcome $\pi$ relative to the standalone
inference optimum is
\[
  \Delta SV^{\Gamma}(\pi) := SV^{\Gamma}(\pi) - SV_{\rm inf}^{\Gamma}(e).
\]
We drop $\Gamma$ when it is clear from context.
\end{definition}

\bigskip

The first step in evaluating social value is to link total resource cost
to total revenue.

\begin{lemma}[Zero-profit identity]
\label{lem:zero_profit_identity}
In any non-degenerate competitive equilibrium
$\pi = (M, D, L, S, Q, p, P)$,
\[
  e(M + D + L) = pQ + PR.
\]
\end{lemma}

\begin{proof}
For each active activity, zero profit implies revenue equals cost:
\[
eM=\frac{PR}{S}M,
\qquad
eD=\frac{PR}{\gamma S}D+\frac{p}{\alpha}D,
\qquad
eL=pL,
\]
with inactive activities contributing zero to both sides.
Summing yields
\[
e(M+D+L)
=
PR\frac{M+D/\gamma}{S}
+
p\left(L+\frac{D}{\alpha}\right).
\]
Using
\[
S=M+\frac{D}{\gamma}
\qquad\text{and}\qquad
Q=L+\frac{D}{\alpha},
\]
we obtain
\[
e(M+D+L)=PR+pQ.
\]
\end{proof}

\begin{proposition}[Net social value of a PoUW equilibrium]
\label{prop:social_value}
Let $\pi = (M,D,L,S,Q,p,P)$ be a non-degenerate competitive
equilibrium in frame $\Gamma$. Then
\[
  \Delta SV(\pi)
  = V_{\rm chain}
  - \underbrace{
      \int_{\mathcal{D}(e)}^{\mathcal{D}(p)}
      \bigl(e - v_{\mathcal{D}}(q)\bigr)dq
    }_{\displaystyle\text{deadweight loss}}
  + \underbrace{
      (e-p)\,\mathcal{D}(p)
    }_{\displaystyle\text{inference subsidy}}
  - \underbrace{
      PR
    }_{\displaystyle\text{block reward}}.
\]
\end{proposition}

\begin{proof}
By Definition~\ref{def:social_value},
\[
SV(\pi)
=
\int_0^Q v_D(q)\,dq
+
V_{\rm chain}
-
e(M+D+L).
\]
By Lemma~\ref{lem:zero_profit_identity},
\[
e(M+D+L)=pQ+PR.
\]
Therefore,
\[
SV(\pi)
=
\int_0^Q v_D(q)\,dq
+
V_{\rm chain}
-
pQ
-
PR.
\]
Using market clearing \(Q=\mathcal D(p)\), and subtracting the
standalone inference optimum
\[
SV_{\rm inf}(e)
=
\int_0^{\mathcal D(e)} v_D(q)\,dq
-
e\mathcal D(e),
\]
we get
\[
\Delta SV(\pi)
=
V_{\rm chain}
+
\int_{\mathcal D(e)}^{\mathcal D(p)} v_D(q)\,dq
-
p\mathcal D(p)
+
e\mathcal D(e)
-
PR.
\]
Now add and subtract \(e\mathcal D(p)\):
\[
\Delta SV(\pi)
=
V_{\rm chain}
+
\int_{\mathcal D(e)}^{\mathcal D(p)} v_D(q)\,dq
-
e(\mathcal D(p)-\mathcal D(e))
+
(e-p)\mathcal D(p)
-
PR.
\]
Finally,
\[
\int_{\mathcal D(e)}^{\mathcal D(p)} v_D(q)\,dq
-
e(\mathcal D(p)-\mathcal D(e))
=
-
\int_{\mathcal D(e)}^{\mathcal D(p)}
(e-v_D(q))\,dq,
\]
which gives the claimed expression.
\end{proof}

Proposition~\ref{prop:social_value} decomposes $\Delta SV$ into
three components. The blockchain value $V_{\rm chain}$ is the value of the direct
benefit of the blockchain. The inference subsidy
$(e-p)\mathcal{D}(p)$ captures the gain from the price reduction:
inference consumers receive $\mathcal{D}(p)$ units at price $p < e$,
whereas in the standalone market they would pay $e$ per unit; the
saving of $e-p$ per unit, summed over all units produced, is the
magnitude of the subsidy. The block reward $PR$ is the resource cost
of funding consensus.

The deadweight loss $\int_{\mathcal{D}(e)}^{\mathcal{D}(p)}
(e-v_{\mathcal{D}}(q))dq$ is the classic efficiency loss from a
subsidy~\cite{harberger1964}: block rewards push the inference price
below cost, inducing consumption of units $q \in (\mathcal{D}(e),
\mathcal{D}(p)]$ whose consumer value $v_{\mathcal{D}}(q)$ falls
short of their resource cost $e$. The loss per unit is $e -
v_{\mathcal{D}}(q) \in [0, e-p]$: strictly less than the full
resource cost $e$. This contrasts with Bitcoin-style mining, where
wasted compute produces nothing of value and the entire resource cost
$e$ is lost per unit. Hence at the same block reward $PR$, Duplexia
strictly dominates Bitcoin: the same expenditure that Bitcoin burns
entirely on mining also produces inference that consumers value,
at no extra cost.

However, fixing $PR$ across two equilibria with different technology
parameters is not the appropriate comparison, since $PR = \theta pQ$
depends on both the T/I ratio $\theta$ and the equilibrium inference
price $p$, which itself depends on technology. The natural comparison
instead holds fixed the T/I ratio $\theta$, which is determined by
monetary primitives independently of technology
(Appendix~\ref{sec:monetary_foundation}). At the same $\theta$, the block reward
differs across regimes --- in Bitconia $PR = \theta\,e\,\mathcal{D}(e)$,
while in Duplexia $PR = \theta p\mathcal{D}(p)$ with $p < e$ --- and
the comparison is more nuanced. We characterize it in the following
corollaries. We focus on pure Duplexia ($\theta \in (\alpha-1,
\theta_{\rm high})$), the empirically relevant Duplexia
subregime.\footnote{Mixed Duplexia requires $\theta > \theta_{\rm
high} = \frac{1}{\gamma-1}$, which for $\gamma = 1.3$ demands $\theta
> 3.33$ --- a token economy more than three times the size of the
inference market, an extreme and unlikely configuration.}

\begin{corollary}[Social value in pure Duplexia]
\label{cor:sv_duplexia}
Fix a frame $\Gamma = (\alpha,\gamma,R,\mathcal{D},e,V_{\rm chain})$
with $\Delta > 0$, and let $\pi = (M,D,L,S,Q,p,P)$ be a non-degenerate
$\theta$-T/I ratio competitive equilibrium in $\Gamma$ with $L = 0$
and $\theta \in (\theta_{\rm low}, \theta_{\rm high})$. Then $p =
\frac{\alpha}{1+\theta}\,e$, $Q = \mathcal{D}(p)$, and
\[
  \Delta SV(\pi)
  = V_{\rm chain}
  - \underbrace{
      \int_{\mathcal{D}(e)}^{\mathcal{D}(p)}
      \bigl(e - v_{\mathcal{D}}(q)\bigr)dq
    }_{\displaystyle\text{deadweight loss}}
  - \underbrace{
      e(\alpha-1)\,\mathcal{D}(p)
    }_{\displaystyle\text{duplex overhead}}.
\]
\end{corollary}

\begin{proof}
By Theorem~\ref{thm:MLI_characterization}, $\Delta > 0$ and $\theta
\in (\theta_{\rm high},\theta_{\rm high})$ imply $p =
\frac{\alpha}{1+\theta}\,e$, so $p(1+\theta) = \alpha e$. Using
$PR = \theta p\mathcal{D}(p)$ and $p(1+\theta) = \alpha e$:
\[
  (e-p)\mathcal{D}(p) - PR
  = e\,\mathcal{D}(p) - p(1+\theta)\mathcal{D}(p)
  = e\,\mathcal{D}(p) - \alpha e\,\mathcal{D}(p)
  = -e(\alpha-1)\mathcal{D}(p).
\]
Substituting into Proposition~\ref{prop:social_value} gives the
claimed expression. 
\end{proof}

\begin{corollary}[Social value when solo inference is active]
\label{cor:sv_solo}
Fix a frame $\Gamma = (\alpha,\gamma,R,\mathcal{D},e,V_{\rm chain})$
and let $\pi = (M,D,L,S,Q,p,P)$ be a non-degenerate $\theta$-T/I ratio
competitive equilibrium in $\Gamma$ with $L > 0$. Then $p = e$,
$Q = \mathcal{D}(e)$, and
\[
  \Delta SV(\pi) = V_{\rm chain} - PR.
\]
\end{corollary}

\begin{proof}
Since $L > 0$, zero profit for solo inference gives $p = e$, so the
subsidy term $(e-p)\mathcal{D}(p)$ in
Proposition~\ref{prop:social_value} vanishes. The result follows
immediately.
\end{proof}

We now compare social value in Duplexia to Bitconia:

\begin{corollary}[Welfare comparison with Bitconia]
\label{cor:sv_comparison}
Fix a frame $\Gamma$ with $\Delta > 0$ and $\theta \in
(\theta_{\mathrm{low}}, \theta_{\mathrm{high}})$, and let $\pi$ be the
corresponding pure Duplexia equilibrium in $\Gamma$. Let
$\pi^{\mathrm{Bit}}$ be any competitive equilibrium at T/I ratio
$\theta$ in any frame sharing the same $(R, \mathcal{D}, e,
V_{\mathrm{chain}}, \theta)$ as $\Gamma$ but with $\Delta < 0$.
Then
\[
  \Delta SV^{\Gamma}(\pi) - \Delta SV(\pi^{\mathrm{Bit}})
  =
  \underbrace{
    \bigl[\theta - \theta_{\mathrm{low}}\bigr]e\,\mathcal{D}(e)
  }_{\text{net saving on baseline compute}}
  -
  \underbrace{
    \int_{\mathcal{D}(e)}^{\mathcal{D}(p)}
    \bigl(e - v_{\mathcal{D}}(q)\bigr)\,dq
  }_{\text{deadweight loss on expansion}}
  -
  \underbrace{
    e(\alpha-1)\bigl(\mathcal{D}(p)-\mathcal{D}(e)\bigr)
  }_{\text{overhead on expansion}},
\]
where $p = \frac{\alpha}{1+\theta}e$.
\end{corollary}
\begin{proof}
Since $\Delta < 0$, by Theorem~\ref{thm:MLI_characterization} the unique
robust competitive equilibrium $\pi^{\mathrm{Bit}}$ is (M+L) with $L > 0$
and $p = e$. By Corollary~\ref{cor:sv_solo},
\[
  \Delta SV(\pi^{\mathrm{Bit}}) = V_{\mathrm{chain}} - P^{\mathrm{Bit}} R.
\]
Since $L > 0$ pins $p = e$ and hence $Q = \mathcal{D}(e)$, the T/I ratio
gives
\[
  P^{\mathrm{Bit}} R
  = \theta\,p\,\mathcal{D}(p)\big|_{p=e}
  = \theta\,e\,\mathcal{D}(e).
\]
This quantity depends only on the shared primitives
$(R, \mathcal{D}, e, \theta)$ and is independent of the specific
$(\alpha,\gamma)$ chosen for the Bitconia frame, as long as $\Delta < 0$.

For the pure Duplexia equilibrium $\pi$, note that
$\theta_{\mathrm{low}} = \alpha - 1$, so the hypothesis
$\theta \in (\theta_{\mathrm{low}}, \theta_{\mathrm{high}})$ is
precisely $\theta > \alpha - 1$, which by
Theorem~\ref{thm:MLI_characterization} ensures $L = 0$ and
$p = \frac{\alpha}{1+\theta}e < e$ in equilibrium.
By Corollary~\ref{cor:sv_duplexia},
\[
  \Delta SV^{\Gamma}(\pi)
  = V_{\mathrm{chain}}
  - \int_{\mathcal{D}(e)}^{\mathcal{D}(p)}
    \bigl(e - v_{\mathcal{D}}(q)\bigr)\,dq
  - e(\alpha-1)\mathcal{D}(p).
\]

Subtracting $\Delta SV(\pi^{\mathrm{Bit}})$ from
$\Delta SV^{\Gamma}(\pi)$ and using
$P^{\mathrm{Bit}} R = \theta\,e\,\mathcal{D}(e)$:
\begin{align*}
  \Delta SV^{\Gamma}(\pi) - \Delta SV(\pi^{\mathrm{Bit}})
  &= \theta\,e\,\mathcal{D}(e)
     - \int_{\mathcal{D}(e)}^{\mathcal{D}(p)}
       \bigl(e - v_{\mathcal{D}}(q)\bigr)\,dq
     - e(\alpha-1)\mathcal{D}(p).
\end{align*}
Writing $e(\alpha-1)\mathcal{D}(p)
= e(\alpha-1)\mathcal{D}(e)
+ e(\alpha-1)\bigl(\mathcal{D}(p)-\mathcal{D}(e)\bigr)$
to separate the overhead on baseline inference from the overhead on
the expansion:
\begin{align*}
  \Delta SV^{\Gamma}(\pi) - \Delta SV(\pi^{\mathrm{Bit}})
  &= \theta\,e\,\mathcal{D}(e) - e(\alpha-1)\mathcal{D}(e) \\
  &\quad
     - \int_{\mathcal{D}(e)}^{\mathcal{D}(p)}
       \bigl(e - v_{\mathcal{D}}(q)\bigr)\,dq \\
  &\quad
     - e(\alpha-1)\bigl(\mathcal{D}(p)-\mathcal{D}(e)\bigr) \\[4pt]
  &= \underbrace{
       \bigl[\theta-(\alpha-1)\bigr]\,e\,\mathcal{D}(e)
     }_{\text{net saving on baseline compute}}
     -
     \underbrace{
       \int_{\mathcal{D}(e)}^{\mathcal{D}(p)}
       \bigl(e - v_{\mathcal{D}}(q)\bigr)\,dq
     }_{\text{deadweight loss on expansion}}
     -
     \underbrace{
       e(\alpha-1)\bigl(\mathcal{D}(p)-\mathcal{D}(e)\bigr)
     }_{\text{overhead on expansion}}.
\end{align*}
Since $\theta_{\mathrm{low}} = \alpha - 1$, the net saving on baseline
compute equals
$[\theta - \theta_{\mathrm{low}}]\,e\,\mathcal{D}(e) > 0$,
which is strictly positive because $\theta > \theta_{\mathrm{low}}$
by hypothesis.
\end{proof}

The decomposition reveals a tradeoff. The
net saving on baseline compute $[\theta-(\alpha-1)]e\mathcal{D}(e)$
is always positive in pure Duplexia since $\theta > \alpha-1$, and
grows linearly in $\theta$. The deadweight loss on expansion is second order in the price
discount $e-p$ by Harberger~\cite{harberger1964}: since $p=e$
maximizes standalone inference social value, the first-order effect
of the subsidy vanishes, and the loss is proportional to the square
of the price discount. The overhead on expansion
$e(\alpha-1)(\mathcal{D}(p)-\mathcal{D}(e))$, however, is first order
in the inference expansion. As $\theta$ grows, the inference price
falls further below cost, expanding demand and amplifying both the
overhead and the deadweight loss on the additional inference units.
For sufficiently large $\theta$, these costs can dominate the
baseline saving. We characterize the sign of the welfare difference
explicitly under isoelastic demand in Section~\ref{sec:graphs}, where
we show that for any $\varepsilon > 1$, Duplexia strictly dominates
whenever $\theta < \frac{1}{\varepsilon-1}$.

\begin{remark}[Social value with market friction]
\label{rem:ephys}
Our welfare analysis benchmarks against the competitive inference price
$p = e$, where $e$ is the all-in cost of compute including financing
costs, returns to capital, and customer acquisition. In practice, the
true physical resource cost $e_{\rm phys} \leq e$ is strictly lower due to ``market friction".\footnote{We are grateful to Daniel Raskin for suggesting to consider market frictions.} 

One possible source of such a wedge is market power in the supply of
compute. To illustrate, suppose that in the absence of the blockchain,
compute is supplied by a monopolist with physical marginal cost
\(e_{\rm phys}\), that inference is the only use of compute, and that
the monopolist sells compute directly into the inference market. The
monopolist therefore faces the inference demand curve
\(\mathcal D\), and chooses a compute price \(e\) to maximize
$(e-e_{\rm phys})\mathcal D(e)$.
Differentiating with respect to \(e\) and setting the derivative equal
to zero gives $\mathcal D(e)+(e-e_{\rm phys})\mathcal D'(e)=0$.
Rearranging,
\[
\frac{e-e_{\rm phys}}{e}
=
-\frac{\mathcal D(e)}{e\mathcal D'(e)}
=
\frac{1}{\varepsilon(e)},
\]
where $\varepsilon(e)$
is the elasticity of demand at price \(e\). This is the classic Lerner
formula~\cite{Lerner1934}. 
(For example, if \(\varepsilon(e)=2\), the Lerner formula gives
\(e=2e_{\rm phys}\): the monopoly price is twice the monopolist's
physical marginal cost. This should not be interpreted as saying that
the all-in compute price is twice the true social cost. The markup
applies only to the component of compute supplied with market power.
Other components, such as electricity, may be competitively priced and
therefore already reflect real resource costs.)

Relative to the social optimum at
$p = e_{\rm phys}$, the standalone competitive inference market at
$p = e$ already underproduces: the wedge $\delta := e - e_{\rm phys} \geq 0$ acts as an
implicit ``tax" on inference: it raises the competitive price from
$e_{\rm phys}$ to $e$, causing the market to underproduce relative
to the social optimum $\mathcal{D}(e_{\rm phys}) \geq \mathcal{D}(e)$.
The PoUW subsidy may then act as a corrective transfer, offsetting this
tax and pushing the inference price back toward $e_{\rm phys}$.

In particular, provided $\theta^*\in(\theta_{\rm low},\theta_{\rm high})$, so that
the equilibrium is in pure Duplexia,
if the T/I ratio satisfies
\[
  \theta^* = \alpha\,\frac{e}{e_{\rm phys}} - 1
           = \alpha\left(1 + \frac{\delta}{e_{\rm phys}}\right) - 1,
\]
then the equilibrium inference price $p^* = \frac{\alpha}{1+\theta^*}e
= e_{\rm phys}$ exactly corrects the market friction and
implements the socially optimal inference quantity. 
\end{remark}

\section{Graphical Illustrations: Duplexia, Fortessia and Bitconia}
\label{sec:graphs}
We now illustrate how competitive equilibria evolve as either the technological parameters $(\alpha,\gamma)$ vary, or the T/I ratio $\theta$ varies.
We organize the
discussion around the distinct equilibrium ``worlds’’ that emerge.
An interactive computational companion is available, allowing the reader
to explore equilibrium behavior across parameter choices and demand
elasticities~\cite{implementation}.

\paragraph{Isoelastic demand}
To illustrate the behavior of equilibria, for concreteness, we assume the demand function to be isoelastic; that is, takes the constant-elasticity form:\footnote{Strictly speaking, isoelastic demand is not a valid
demand function since
$\mathcal{D}(0) = \infty$. One can always consider a quantity-capped
version $\mathcal{D}(p)=\min\{\bar Q,Kp^{-\varepsilon}\}$ for a large $\bar{Q§}$;
for simplicity, we ignore this inconsequential difference.}
$\mathcal D(p)=K\,p^{-\varepsilon}, \varepsilon>0,$
where $\varepsilon$ is referred to as the elasticity parameter.\footnote{Indeed, note that an isoelastic demand function with elasticity parameter $\varepsilon$ has constant point elasticity $\varepsilon$ everywhere, since $-\frac{\mathcal D'(p)\,p}{\mathcal D(p)}=\varepsilon$.}
For concreteness, we normalize the compute cost to $e = 1$; 
under this normalization, $\mathcal{D}(e) = K e^{-\varepsilon} = K$, 
so $K$ directly equals the quantity of inference demanded 
at the compute cost price during a single 10-minute block.
Assuming annual inference spending of $\$30$B (this number is for illustrative purposes only) and 
$N = 52{,}560$ ten-minute blocks per year yields
$K = \$30\text{B}/N \approx \$571{,}000$, so we use $K = 571{,}000$.

In our illustrations, we further set elasticity $\varepsilon = 2$, but also consider other choices. We note that the qualitative impact of PoUW becomes progressively stronger as demand elasticity increases. 

\subsection{Economic worlds: Duplexia, Fortessia and Bitconia}
\label{sec:worlds_isoelastic}
\begin{figure}[t]
\centering
\includegraphics[width=1\linewidth]{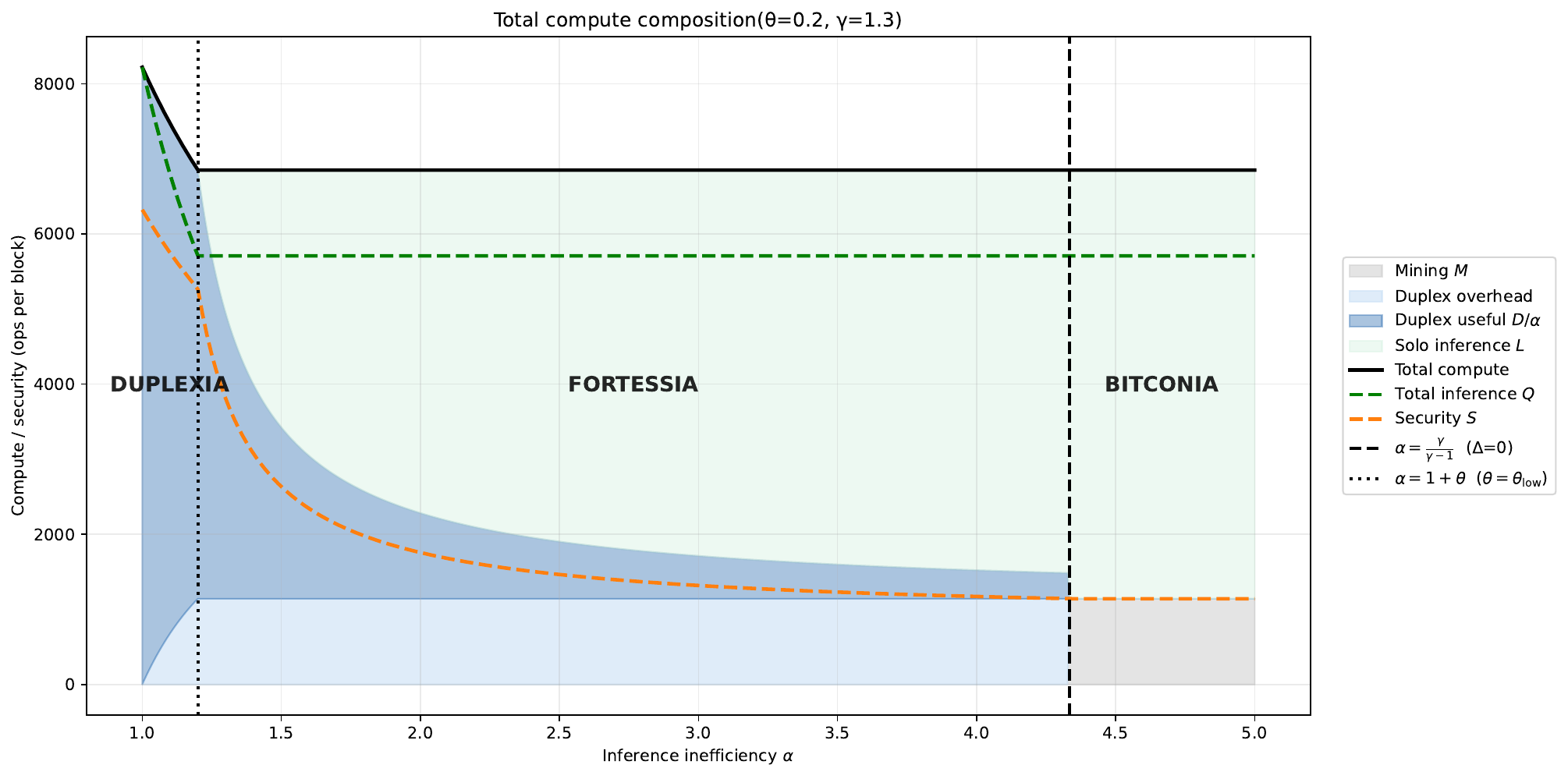}
\caption{Compute composition and security at duplex overhead
($\gamma = 1.3$).
The labeled regions correspond to Bitconia, Fortessia, and Duplexia as
described in Theorem~\ref{thm:informal_full_transition}.}
\label{fig:compute_low_gamma}
\end{figure}
To illustrate the different economic worlds, for concreteness, we focus on the case when $\gamma = 1.3$ 
(i.e., duplex compute requires $30\%$ more 
operations than solo-mining to produce the 
same security output).

Figure~\ref{fig:compute_low_gamma} demonstrates the composition of compute in the equilibria that emerge for varying $\alpha$, but keeping T/I ratio fixed at $\theta = 0.2$. As can be seen, and consistent with Theorem \ref{thm:MLI_characterization}, the choice of $\alpha$ partitions the equilibrium into three
regions:
\emph{Duplexia}, \emph{Fortessia}, and \emph{Bitconia}.

\paragraph{Bitconia:}
For sufficiently large $\alpha$, duplex compute is dominated by solo-inference and solo-mining, so mining and inference decouple, and the system reduces to a classical
proof-of-work equilibrium with a separate inference markets. 

\paragraph{Fortessia:}
As $\alpha$ decreases, we enter \emph{Fortessia}. Here $D$ dominates solo-mining, but a fraction of compute is still spent on solo-inference. Since $L$ is active, the price of inference is pinned to $e$, and thus total inference supplied is no different than in Bitconia. Additionally, as see in Figure~\ref{fig:compute_low_gamma}, the amount of ``wasted" compute (due to the overhead of Duplex compute) is the same as in Bitconia (even if nobody is doing any ``useless PoW" by performing solo-mining, the overhead in performing Duplex compute leads to ``wasted compute"). Thus, from an economic point of view, this world is not very different from Bitconia. But as shown in Figure~\ref{fig:compute_low_gamma}, there is a huge difference in terms of security compared to Bitconia: in Fortessia, the security budget (i.e., the amount of computational resources spend on securing the network) is much higher than in Bitconia for the same amount of ``wasted compute"---that is, security is \emph{fortified} (amplified).

\paragraph{Duplexia:}
Once $\alpha$ reaches $1+\theta$, we enter Duplexia. Here $D$ dominates both $L$ and $M$. Security continues to increase, and more interestingly, the token price starts to subsidize the inference price, which in turn leads to higher inference demand and consequently higher token prices (see~Figure~\ref{fig:total_spend_low_gamma}, and Figure~\ref{fig:rebate_low_gamma} discussed in 
Section~\ref{sec:maintheorem.intro})---that is, we get a \emph{self-reinforcing feedback loop}. 

To summarize, in Duplexia, we have \emph{duplex} (i.e., two-fold) gains with respect to Bitconia, both in terms of security (higher security) and in terms of economics (more inference served, lower inference prices, smaller wasted computation, and higher token prices).

\subsection{The Economic Effect of the T/I Ratio}
\begin{figure}
\centering

\includegraphics[height=0.27\textheight]{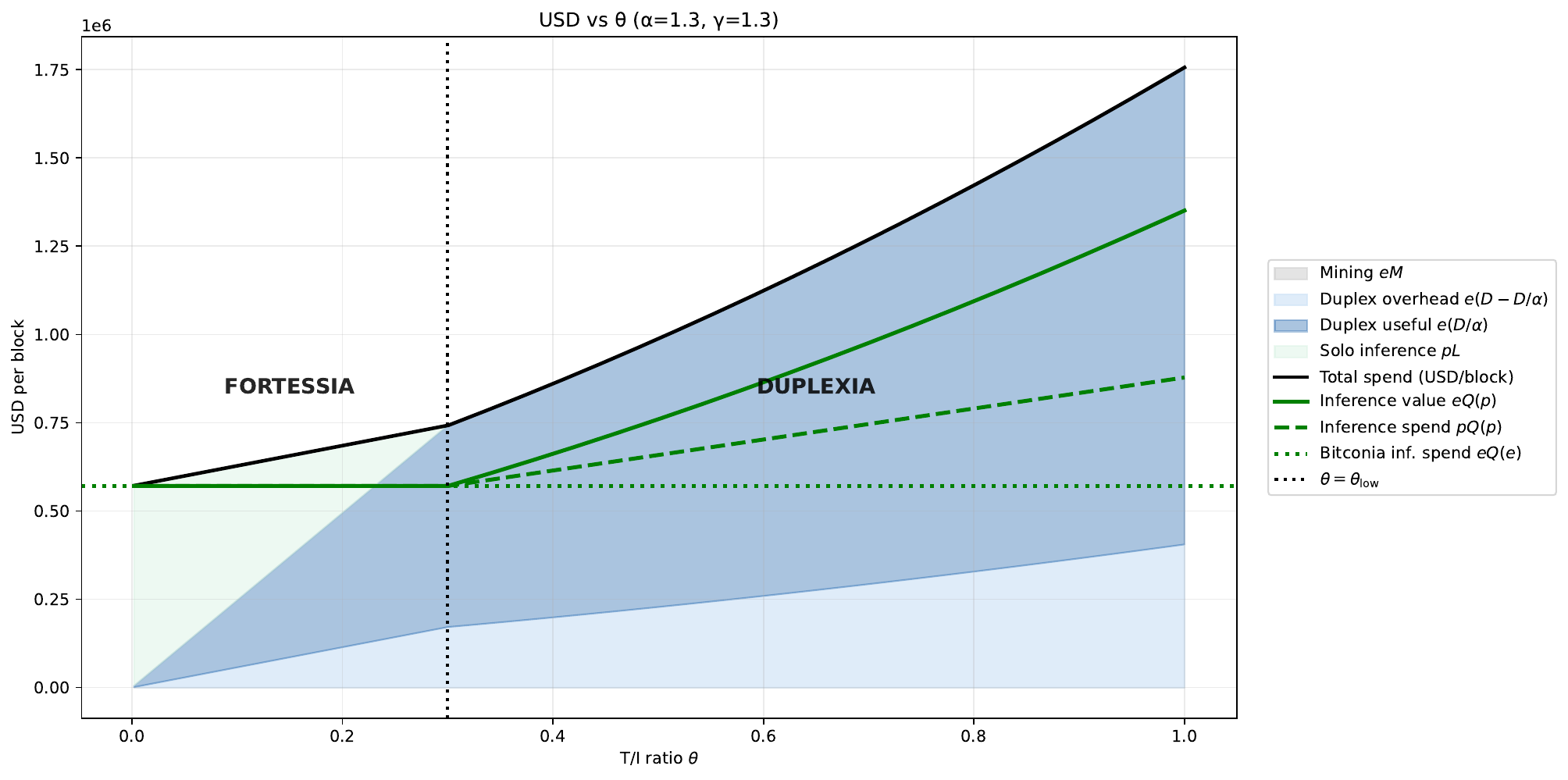}
\caption*{(a) Inference spend $pQ$ vs.\ Bitconia baseline.}

\vspace{1.0em}

\includegraphics[height=0.27\textheight]{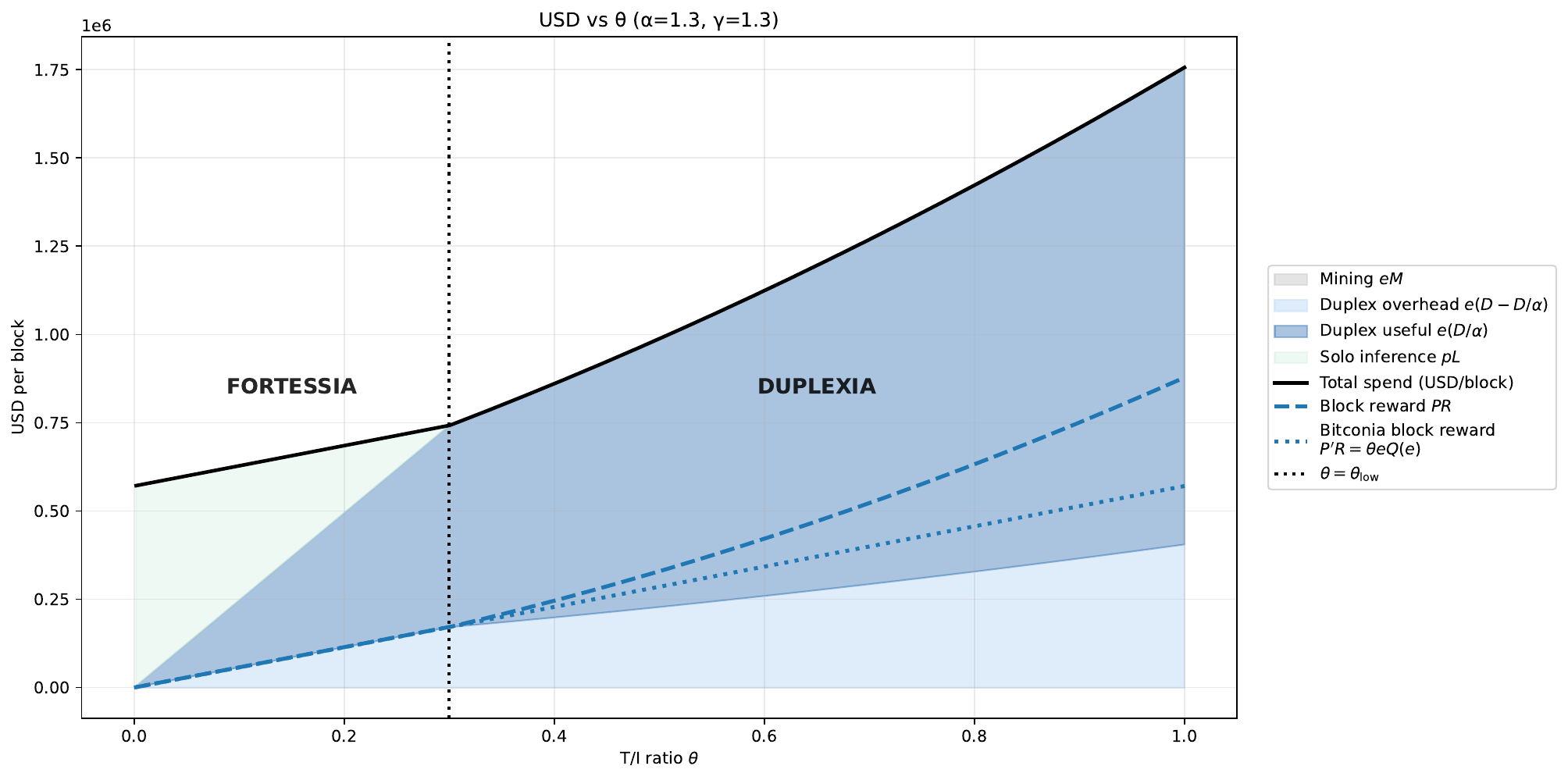}
\caption*{(b) Block reward $PR$ vs.\ Bitconia baseline.}

\vspace{1.0em}

\includegraphics[height=0.27\textheight]{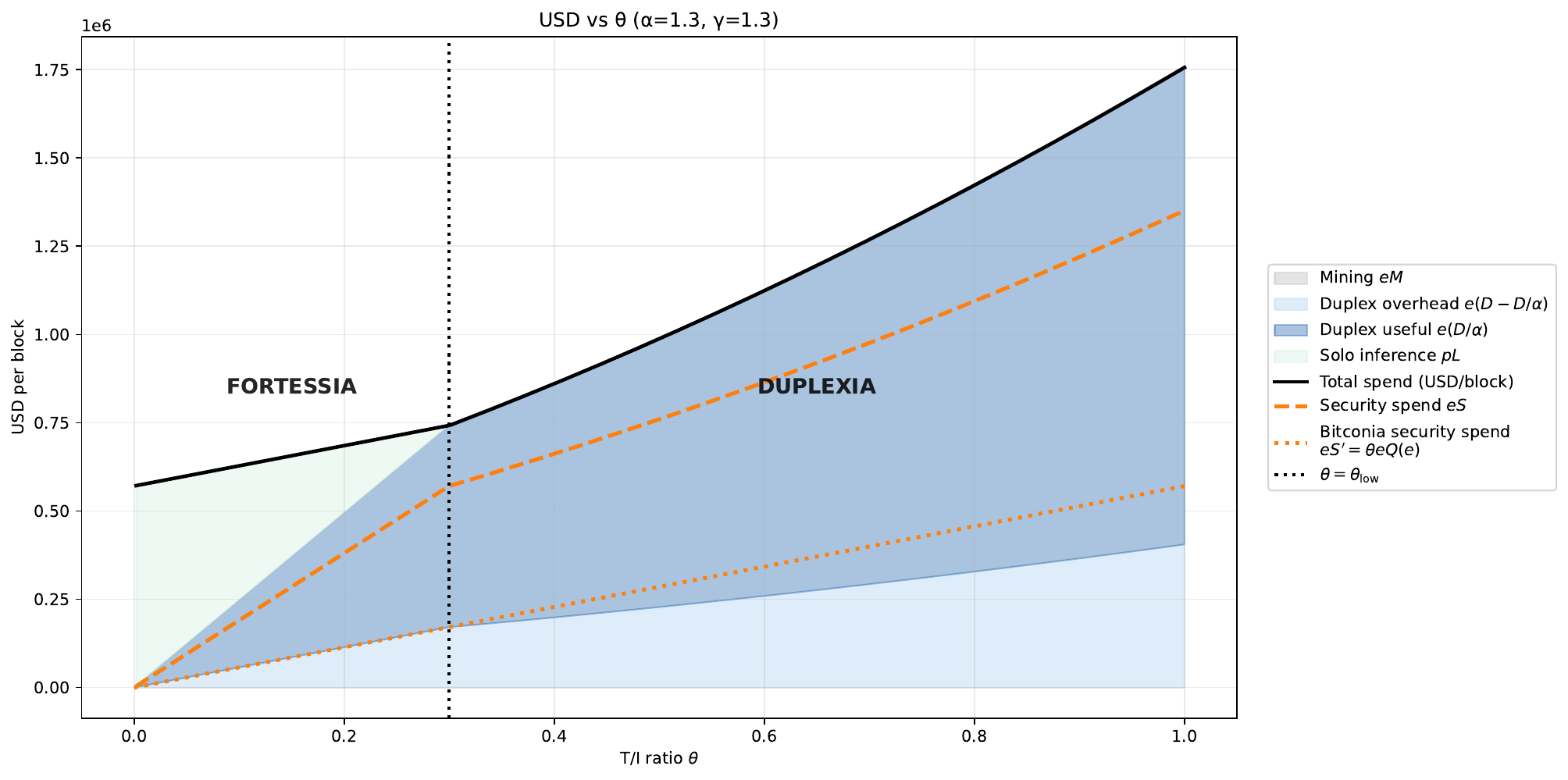}
\caption*{(c) Security spend $eS$ vs.\ Bitconia baseline.}

\caption{Varying T/I rate $\theta$ at fixed $\alpha=1.3$: each panel isolates one economic quantity and compares it a Bitcoin-style baseline.}
\label{fig:theta_triptych}

\end{figure}

We next examine how the T/I ratio $\theta$ affects the equilibrium, as illustrated in Figure \ref{fig:theta_triptych}.
In a Bitcoin-style system (Bitconia), increasing $\theta$ has no economic effect beyond raising the block reward. Higher rewards attract more mining and thus more security, but this additional expenditure does not translate into useful outcomes: the extra security comes entirely from additional wasted work.

\paragraph{Fortessia:} When duplex efficiency is moderately high and the system enters Fortessia, the effect of increasing $\theta$ is already more meaningful: As the T/I ratio increases, the cumulative hashrate (i.e., security) increases at a significantly higher rate than in Bitconia, strengthening the system's resistance to attack \emph{without increasing the amount of wasteful work} from Bitconia. That is, the security efficiency increases.

\paragraph{Duplexia:}
The most significant impact occurs once the system enters \emph{Duplexia}. 
In this regime, security continues to diverge from Bitconia, as in Fortessia, but now the divergence also appears in inference spending and block rewards.  In particular, while total inference spending increases by roughly a factor 
of $1.5$, the total \emph{value} of inference supplied—i.e., the cost of 
generating this inference at the compute cost $e$—increases by 
roughly a factor of $2$; this gap arises from the previously discussed 
subsidy on inference; see Figure~\ref{fig:rebate_low_gamma}.

To understand how elasticity affects the block 
reward $B = PR$, note that by the T/I ratio 
condition, $B = \theta \, p \, Q(p)$. In 
Bitconia, $p = e$ is independent of $\theta$, 
so $B^{\rm Bitconia} = \theta \, e \, Q(e)$ 
grows linearly in $\theta$. In Duplexia, by 
contrast, $p(\theta) = \frac{\alpha e}{1+\theta}$ 
depends on $\theta$, so both $p$ and $Q(p)$ 
vary with $\theta$, leading to a nonlinear 
relationship. To see this concretely, consider 
isoelastic demand $Q(p) = Kp^{-\varepsilon}$ 
for constants $K > 0$ and $\varepsilon > 1$. 
Under this demand function, $B = \theta \, K \, 
p^{1-\varepsilon}$, and substituting 
$p(\theta) = \frac{\alpha e}{1+\theta}$ gives
\[
B(\theta)
= \theta \cdot K \cdot \left(\frac{\alpha e}{1+\theta}
\right)^{1-\varepsilon}
= K(\alpha e)^{1-\varepsilon}\,\theta(1+\theta)^{\varepsilon-1}.
\]
When $\varepsilon = 2$, this becomes 
$B(\theta) = K(\alpha e)^{-1}\,\theta(1+\theta)$, 
which is quadratic in $\theta$, explaining the 
visible widening of the gap relative to the 
Bitconia baseline in 
Figures~\ref{fig:theta_triptych}b--\ref{fig:theta_triptych}c.

\subsection{Sensitivity to Demand Elasticity}
\label{sec:graph_elasticity}
\begin{figure}
\centering

\includegraphics[height=0.27\textheight]{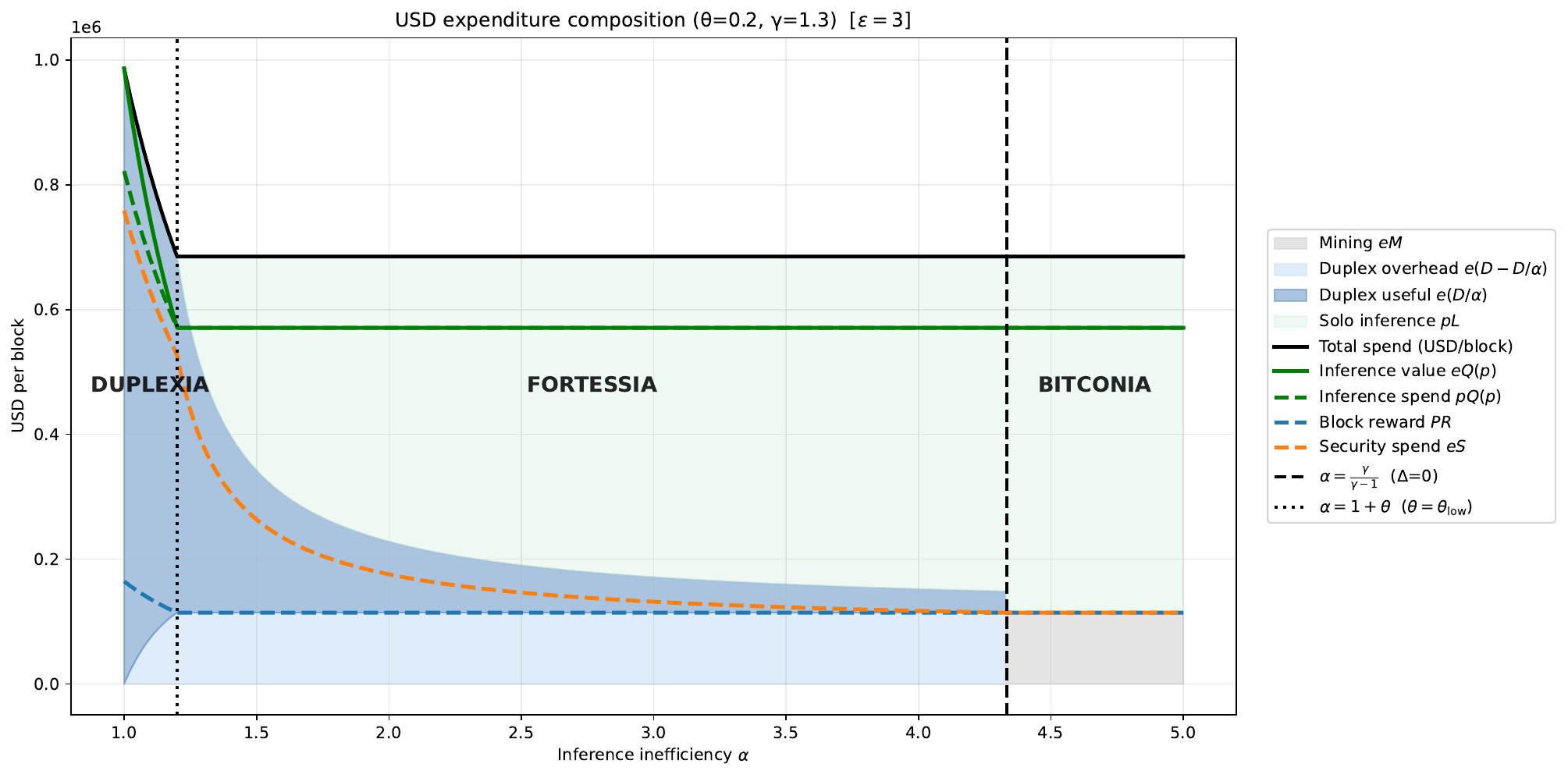}
\caption*{(a) USD spend composition with elasticity $\varepsilon=3$.}

\vspace{0.6em}

\includegraphics[height=0.27\textheight]{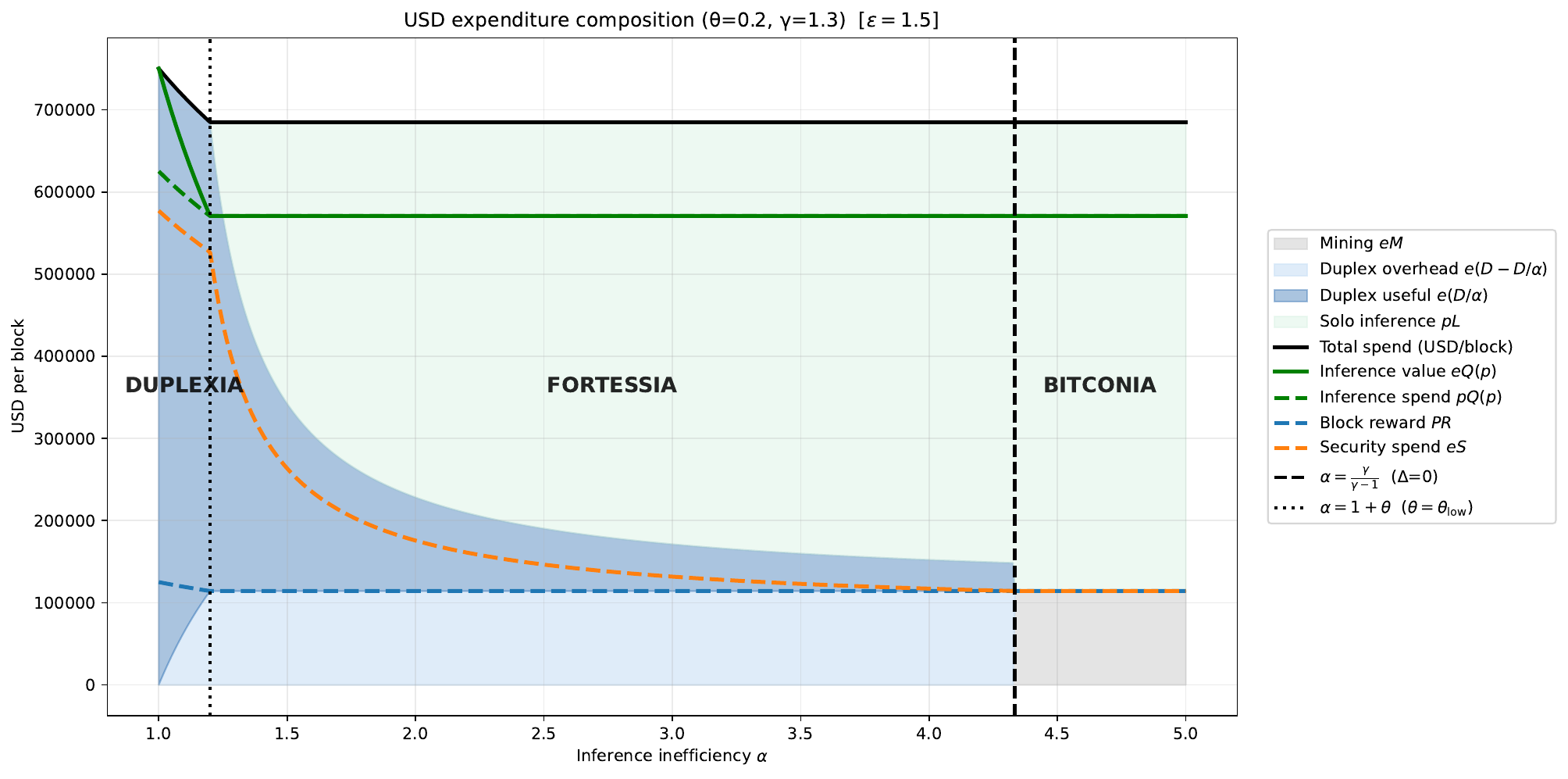}
\caption*{(b) USD spend composition with elasticity $\varepsilon=1.5$.}

\vspace{0.6em}

\includegraphics[height=0.27\textheight]{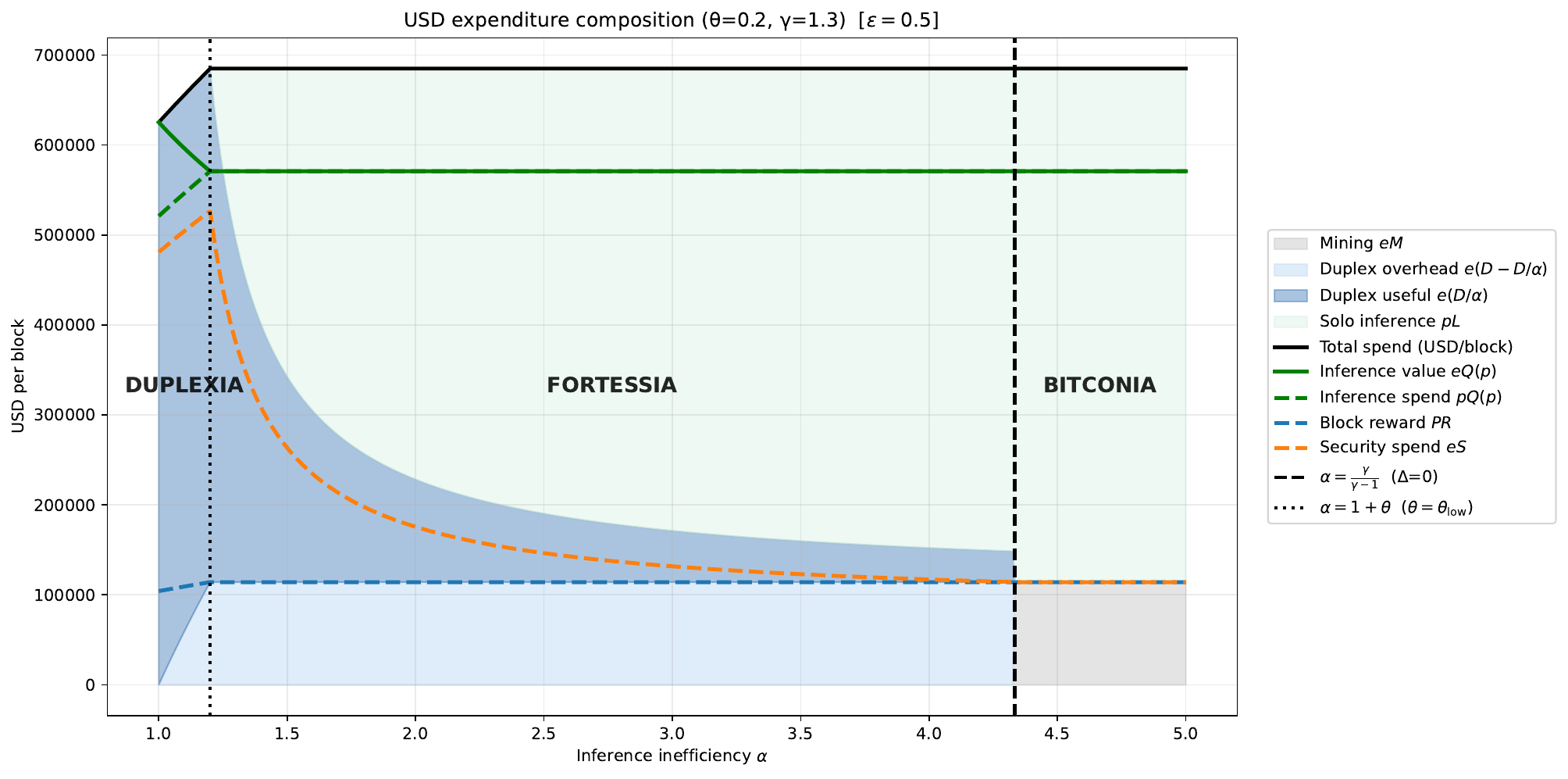}
\caption*{(c) USD spend composition with elasticity $\varepsilon=0.5$.}

\caption{Sensitivity of the USD spend composition to demand elasticity. Higher elasticity strengthens the amplification effect of PoUW. Inelastic demand leads to slightly lower inference spend, security and token values, but still increases inference value supplied.}
\label{fig:eps_sensitivity}
\end{figure}

We here examine how the effect of PoUW varies with the 
elasticity parameter $\varepsilon$, keeping all other parameters fixed at  $(\alpha, \gamma, \theta) = (1.3, 1.3, 0.2)$. Figure~\ref{fig:eps_sensitivity} plots the 
USD expenditure composition for $\varepsilon \in \{3, 1.5, 0.5\}$.

The role of elasticity is visible only Duplexia, where the block-reward subsidy lowers 
inference prices and expands inference supply. The key quantity is total 
inference spending $pQ(p)$: by Lemma~\ref{lem:elastic_spending}, this is 
decreasing in $p$ when $\varepsilon \geq 1$ and (by the proof of Lemma~\ref{lem:elastic_spending}) increasing in $p$ when 
$\varepsilon < 1$. Consequently, when demand is strictly elastic ($\varepsilon > 1$, 
panels (a) and (b)), lower inference prices in Duplexia expand total inference 
spending, which raises token value and security spend; and the amplification 
effect grows stronger with $\varepsilon$, as seen by the larger gap between 
inference value $eQ(p)$ and inference spend $pQ(p)$ in panel (a) relative to 
panel (b).

By contrast, when demand is \emph{inelastic} ($\varepsilon = 0.5$, panel (c)), lower 
inference prices reduce total inference spending $pQ(p)$, which in turn lowers 
token value and security spend relative to the elastic case. Nonetheless, even 
in this case, the inference value delivered $eQ(p)$---that is, the cost of 
generating the supplied inference at compute cost $e$---still 
increases in Duplexia relative to Bitconia, since more inference is supplied 
at a lower price. Thus the blockchain still expands the amount of ``socially 
useful computation" even under inelastic demand, though the token and security amplification effects are weaker.
Indeed, as noted in Remark~\ref{rem:inference_value_monotone}, this holds 
for any downward-sloping demand function, regardless of elasticity.

Taken together, these figures confirm that $\varepsilon > 1$ is the key condition driving the amplification of token value and security in Duplexia, 
while the expansion of useful computation is a more robust phenomenon that 
persists even under inelastic demand.
As mentioned, for AI inference---and more broadly for useful-computation applications where 
lower prices unlock new use cases, we expect demand to be elastic, consistent with the Jevons 
effect~\cite{jevons1865coal,saunders1992khazzoom,sorrell2009rebound}.
Thus, $\varepsilon >  1$ is likely the empirically relevant regime for inference 
markets.

\subsection{Social Value}
\label{sec:sv_graphs}

We now specialize Corollary~\ref{cor:sv_comparison} to isoelastic
demand $\mathcal{D}(p) = Kp^{-\varepsilon}$, $\varepsilon > 1$. The
inverse demand function is $v_{\mathcal{D}}(q) = (K/q)^{1/\varepsilon}$.
Set $Q_e := \mathcal{D}(e) = Ke^{-\varepsilon}$, so $K = Q_e
e^\varepsilon$ and
\[
  v_{\mathcal{D}}(q)
  = \left(\frac{Q_e e^\varepsilon}{q}\right)^{1/\varepsilon}
  = e\left(\frac{Q_e}{q}\right)^{1/\varepsilon}.
\]

\begin{corollary}[Welfare dominance under isoelastic demand]
\label{cor:sv_isoelastic}
Fix a frame $\Gamma$ with $\Delta > 0$, isoelastic demand
$\mathcal{D}(p) = Kp^{-\varepsilon}$ with $\varepsilon > 1$, and
$\theta \in (\theta_{\mathrm{low}}, \theta_{\mathrm{high}})$. Let $\pi$ be the
non-degenerate $\theta$-T/I ratio competitive equilibrium in $\Gamma$
with $L = 0$ (pure Duplexia), and let $\pi^{\mathrm{Bit}}$ be the
unique robust competitive equilibrium at T/I ratio $\theta$ in any
frame sharing the same $(R, \mathcal{D}, e, V_{\mathrm{chain}}, \theta)$
as $\Gamma$ but with $\Delta < 0$. Set $Q_e := \mathcal{D}(e)$ and
$\rho := (1+\theta)/\alpha > 1$. Then
\[
  \Delta SV^{\Gamma}(\pi) - \Delta SV(\pi^{\mathrm{Bit}})
  = \left(\rho^{\varepsilon-1}-1\right)
    \left(\frac{1}{\varepsilon-1} - \theta\right)
    eQ_e.
\]
Consequently,
\[
  \Delta SV^{\Gamma}(\pi) - \Delta SV(\pi^{\mathrm{Bit}})
  \begin{cases}
    > 0 & \text{if } \theta < \dfrac{1}{\varepsilon-1}, \\[4pt]
    = 0 & \text{if } \theta = \dfrac{1}{\varepsilon-1}, \\[4pt]
    < 0 & \text{if } \theta > \dfrac{1}{\varepsilon-1}.
  \end{cases}
\]
\end{corollary}
\begin{proof}
By Corollary~\ref{cor:sv_comparison}:
\[
  \Delta SV^{\Gamma}(\pi) - \Delta SV(\pi^{\rm Bit})
  =
  \underbrace{
    \bigl[\theta - (\alpha-1)\bigr]eQ_e
  }_{\text{net saving on baseline}}
  -
  \underbrace{
    \int_{Q_e}^{Q_e\rho^\varepsilon}
    \bigl(e - v_{\mathcal{D}}(q)\bigr)dq
  }_{\text{deadweight loss on expansion}}
  -
  \underbrace{
    e(\alpha-1)Q_e(\rho^\varepsilon-1)
  }_{\text{overhead on expansion}}.
\]
Using $v_{\mathcal{D}}(q) = e(Q_e/q)^{1/\varepsilon}$, the deadweight
loss evaluates to:
\begin{align*}
  \int_{Q_e}^{Q_e\rho^\varepsilon}
  \bigl(e-v_{\mathcal{D}}(q)\bigr)dq
  &= \int_{Q_e}^{Q_e\rho^\varepsilon}
    \left(e - e\left(\frac{Q_e}{q}\right)^{1/\varepsilon}\right)dq \\
  &= e\left[q -
    \frac{Q_e^{1/\varepsilon}\,q^{1-1/\varepsilon}}{1-1/\varepsilon}
    \right]_{Q_e}^{Q_e\rho^\varepsilon} \\
  &= eQ_e\left[\rho^\varepsilon - 1 -
    \frac{\varepsilon}{\varepsilon-1}
    \bigl(\rho^{\varepsilon-1}-1\bigr)\right].
\end{align*}
The welfare difference divided by $eQ_e$ therefore equals:
\begin{align*}
  &\phantom{{}={}}
  [\theta-(\alpha-1)]
  - \left[\rho^\varepsilon - 1 -
    \frac{\varepsilon}{\varepsilon-1}(\rho^{\varepsilon-1}-1)\right]
  - (\alpha-1)(\rho^\varepsilon-1) \\
  &= \theta - \alpha + 1
  - \rho^\varepsilon + 1
  + \frac{\varepsilon}{\varepsilon-1}(\rho^{\varepsilon-1}-1)
  - \alpha\rho^\varepsilon + \rho^\varepsilon + \alpha - 1 \\
  &= \theta + 1
  - \alpha\rho^\varepsilon
  + \frac{\varepsilon}{\varepsilon-1}(\rho^{\varepsilon-1}-1)
  \qquad[\pm\alpha,\ \pm 1,\ \pm\rho^\varepsilon \text{ cancel}] \\
  &= \alpha\rho - \alpha\rho^\varepsilon
  + \frac{\varepsilon}{\varepsilon-1}(\rho^{\varepsilon-1}-1)
  \qquad[1+\theta = \alpha\rho] \\
  &= \alpha\rho(1-\rho^{\varepsilon-1})
  + \frac{\varepsilon}{\varepsilon-1}(\rho^{\varepsilon-1}-1) \\
  &= (\rho^{\varepsilon-1}-1)
    \left(\frac{\varepsilon}{\varepsilon-1} - \alpha\rho\right) \\
  &= (\rho^{\varepsilon-1}-1)
    \left(\frac{1}{\varepsilon-1} - \theta\right),
\end{align*}
where the last step uses $\alpha\rho = 1+\theta$. 
Since $\rho > 1$ (as $\theta > \alpha-1 = \theta_{\rm low}$ by
hypothesis) and $\varepsilon > 1$, we have $\rho^{\varepsilon-1}
> 1$, so the first factor is always positive and the sign is
determined by $\frac{1}{\varepsilon-1} - \theta$.
\end{proof}

The threshold $\theta^* := \frac{1}{\varepsilon-1}$ has a natural
interpretation: more elastic demand
means a larger inference expansion, amplifying both the overhead and
the deadweight loss on expanded compute, so Duplexia dominates over
a smaller range of $\theta$. For $\varepsilon = 2$, $\theta^* = 1$ (i.e., Duplexia dominates whenever inference market revenue exceeds
block reward revenue).

Note also that the threshold $\theta^*$ is independent of $\alpha$:
the duplex technology parameter affects the magnitude of the welfare
difference but not the range of $\theta$ over which Duplexia
dominates Bitconia.

Figure~\ref{fig:welfare_decomp} illustrates the decomposition for
$(\alpha,\gamma,\varepsilon) = (1.3,1.3,2)$. The net saving on
baseline (blue) grows linearly in $\theta$, while the overhead on
expansion (orange) and deadweight loss (green) both start at zero at
$\theta_{\rm low} = \alpha-1 = 0.3$ and grow as the inference
expansion widens. The deadweight loss is visibly second order near
$\theta_{\rm low}$. The net gain over Bitconia (black) is positive
for $\theta \in (\theta_{\rm low}, 1)$ and vanishes at $\theta = 1$,
consistent with Corollary~\ref{cor:sv_isoelastic}.

\begin{figure}[h]
\centering
\includegraphics[width=\textwidth]{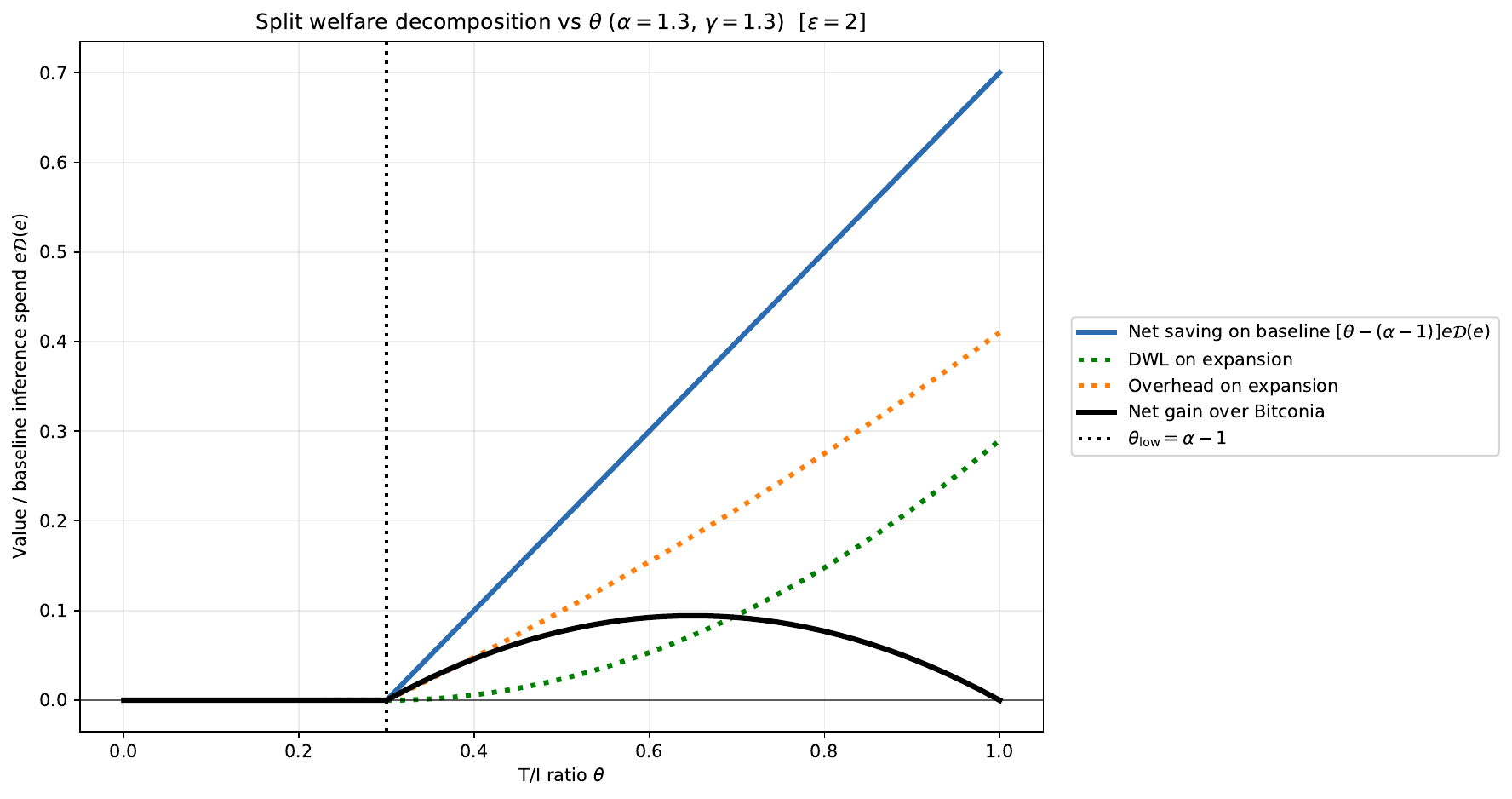}
\caption{Welfare decomposition of Corollary~\ref{cor:sv_comparison}
for $(\alpha,\gamma,\varepsilon) = (1.3,1.3,2)$, normalized by
$e\mathcal{D}(e)$. The net saving on baseline compute (blue solid)
grows linearly in $\theta$. The overhead on expansion (orange dotted)
and deadweight loss on expansion (green dotted) both start at zero
at $\theta_{\rm low} = \alpha-1 = 0.3$ (dotted vertical line) and
grow as the inference expansion widens. The deadweight loss is
visibly second order near $\theta_{\rm low}$. The net gain over
Bitconia (black solid) is positive for $\theta \in (\theta_{\rm
low}, 1)$ and vanishes at $\theta = 1$, consistent with
Corollary~\ref{cor:sv_isoelastic}.}
\label{fig:welfare_decomp}
\end{figure}

\section{Equilibria Regime Classifications without Fixed T/I ratio}
Theorem~\ref{thm:MLI_characterization} allows us to compare the effect of improved technologies while holding the T/I ratio fixed, enabling an ``apples-to-apples'' comparison between PoUW equilibria and the classical PoW benchmark (Bitconia).

We now use Theorem~\ref{thm:MLI_characterization} 
to study the structure of equilibria \emph{without 
fixing the T/I ratio}, restricting attention to 
robust equilibria. Without fixing $\theta$, the 
token price $P$ is a free variable: for any $P > 0$, 
there exists a competitive equilibrium consistent 
with that price, since nothing pins down the 
relationship between the security market and the 
inference market. As a result, absolute security 
levels $S$ cannot be meaningfully compared across 
different systems: a PoW system could exhibit 
higher $S$ than a PoUW system simply because its 
token trades at a higher price, not because of any 
fundamental efficiency advantage in converting 
compute into security. We therefore focus on 
\emph{security efficiency}, measured by comparing 
security $S$ to wasted compute $W$ (recall from 
Section~\ref{secefficiency.sec} that 
$W := M + D(1 - 1/\alpha)$ denotes the amount of 
computation that does not contribute to useful 
inference). By Proposition~\ref{prop:SI_duplex}, 
$S - W = D\Delta$, so whenever $\Delta > 0$ and 
duplex is active, security strictly exceeds wasted 
compute ($S > W$).

The key message of the following theorem is that 
the qualitative structure of the three 
regimes---Bitconia, Fortessia, and Duplexia---is 
not an artifact of fixing $\theta$. Even without 
fixing the T/I ratio, the same qualitative phenomena 
persist in any robust competitive equilibrium: in 
Bitconia the inference and security markets are 
fully decoupled; in Fortessia security \emph{efficiency} 
is strictly improved relative to Bitconia; and in 
Duplexia inference is subsidized ($p < e$), useful 
computation expands ($Q \geq \mathcal{D}(e)$), and 
security efficiency is again strictly improved. The 
T/I ratio determines the \emph{quantitative} 
magnitude of these effects, but not their 
\emph{qualitative} nature.

\begin{theorem}[Robust equilibrium regimes and security efficiency]\label{thm:robust_regimes_SI}
Fix primitives $(e,\alpha,\gamma,R)$ with
$e>0$, $\alpha\ge1$, $\gamma>1$, and $\mathcal D(e)>0$ such that  $\Delta\neq 0$.
In any non-degenerate robust competitive equilibrium $(M,D,L,S,Q,p,P)$, exactly one of the following holds.
\begin{enumerate}
\item \textbf{Case $\Delta<0$: Bitconia ($M+L$).} $M>0,\quad L>0,\quad D=0$ and
\[
p=e,\quad Q=\mathcal D(e),\quad S=M =W,\quad P=\frac{eS}{R}.
\]
\item \textbf{Case $\Delta>0$:} exactly one of the following holds.
\begin{enumerate}
\item[(a)] \textbf{Fortessia ($D{+}L$).} $M=0,\quad D>0,\quad L>0$ and
\[
p=e,\quad Q=\mathcal D(e),\quad S=\frac{D}{\gamma}>W,\quad
P=\frac{\gamma S}{R}\,e\Bigl(1-\frac1\alpha\Bigr).
\]

\item[(b)] \textbf{Duplexia ($L=0$).} $D>0,\quad L=0$, and exactly one of the following subcases holds.
\begin{enumerate}
\item[(b1)] \textbf{Pure Duplexia ($D$ only).} $M=0$ and
\[
p\in\Bigl(\alpha e\Bigl(1-\frac1\gamma\Bigr),\,e\Bigr),\qquad
P=\frac{\gamma S}{R}\Bigl(e-\frac{p}{\alpha}\Bigr)\;<\;\frac{eS}{R}.
\]

\item[(b2)] \textbf{Mixed Duplexia ($M{+}D$).} $M>0$ and
\[
p=\alpha e\Bigl(1-\frac1\gamma\Bigr),\qquad
P=\frac{eS}{R}.
\]
In both subcases,
\[
p<e,\qquad S > W,\qquad Q=\mathcal D(p) \geq D(e).
\]
\end{enumerate}
\end{enumerate}
\end{enumerate}
\end{theorem}
\begin{proof}
Let \((M,D,L,S,Q,p,P)\) be a non-degenerate robust competitive
equilibrium with \(\Delta\neq 0\).
First, let us argue that \(p>0\). If \(p=0\), then \(\pi_L=-e<0\), so \(L=0\). Since
\(Q>0\), market clearing gives \(D>0\). Zero profit for duplex then
gives \(PR/(\gamma S)=e\), hence \(PR/S=\gamma e>e\), contradicting
\(\pi_M\le 0\).

Next consider the case \(P=0\).\footnote{This is a degenerate special case that can only arise when
$\alpha = 1$; we include it for completeness.} Then \(PR=0\), so \(\pi_M=-e<0\) and
therefore \(M=0\). Since \(S>0\), we must have \(D>0\). Zero profit for
duplex gives \(p=\alpha e\). Free entry for solo inference gives
\(\pi_L=p-e\le0\), so \(\alpha=1\) and \(p=e\). Robustness then implies
\(L>0\), since if \(L=0\) \(\pi_L=p-e = 0\) contradicts robustness. Thus the equilibrium is
\((D+L)\), with
\[
        Q=\mathcal D(e),\qquad
        S=\frac{D}{\gamma},\qquad
        W=0<S,
\]
and
\[
        P=\frac{\gamma S}{R}e\left(1-\frac1\alpha\right)=0.
\]
Moreover \(\Delta=1/\gamma>0\). Hence this is exactly case~\((2a)\),
Fortessia.

It remains to consider \(P>0\). Since \(p>0\) and \(Q>0\), the T/I ratio
\[
        \theta:=\frac{PR}{pQ}
\]
is well-defined and strictly positive. Since the equilibrium is robust and
\(\Delta\neq0\), the parameters are away from the knife-edge cases of
Theorem~\ref{thm:MLI_characterization}. Thus
Theorem~\ref{thm:MLI_characterization} applies.

\emph{Case \(\Delta<0\): Bitconia.}
By Theorem~\ref{thm:MLI_characterization}, the unique robust regime is
\((M+L)\). Hence \(M,L>0\), \(D=0\), and \(p=e\). Market clearing gives
\(Q=\mathcal D(e)\). Since \(D=0\), we have \(S=M=W\), and since
\(M>0\), Proposition~\ref{prop:price_pinning} gives \(P=eS/R\). This is
case~\((1)\).

\emph{Case \(\Delta>0\): Fortessia.}
If the equilibrium is in the \((D+L)\) regime, then \(M=0\), \(D,L>0\),
and \(p=e\). Thus \(Q=\mathcal D(e)\) and \(S=D/\gamma\). Since \(D>0\),
Proposition~\ref{prop:price_pinning} gives
\[
        P=\frac{\gamma S}{R}e\left(1-\frac1\alpha\right).
\]
Moreover, since \(D>0\) and \(\Delta>0\), Proposition~\ref{prop:SI_duplex}
gives \(S-W=D\Delta>0\). This is case~\((2a)\).

\emph{Case \(\Delta>0\): pure Duplexia.}
If the equilibrium is in the \(D\)-only regime, then \(M=L=0\), \(D>0\),
and
\[
        p=\frac{\alpha e}{1+\theta}
        \in
        \left(\alpha e\left(1-\frac1\gamma\right),e\right).
\]
By Proposition~\ref{prop:price_pinning},
\[
        P=\frac{\gamma S}{R}\left(e-\frac p\alpha\right).
\]
The lower bound \(p>\alpha e(1-1/\gamma)\) implies
\[
        P<\frac{eS}{R}.
\]
Also \(p<e\), so \(Q=\mathcal D(p)\ge \mathcal D(e)\), with strict
inequality if demand is strictly decreasing on \([p,e]\). Finally,
\(S-W=D\Delta>0\). This is case~\((2b1)\).

\emph{Case \(\Delta>0\): mixed Duplexia.}
If the equilibrium is in the \((M+D)\) regime, then \(M,D>0\), \(L=0\),
and
\[
        p=\alpha e\left(1-\frac1\gamma\right)<e.
\]
Since \(M>0\), Proposition~\ref{prop:price_pinning} gives \(P=eS/R\).
Also \(Q=\mathcal D(p)\ge \mathcal D(e)\), with strict inequality if
demand is strictly decreasing on \([p,e]\). Finally, since \(D>0\) and
\(\Delta>0\), \(S-W=D\Delta>0\). This is case~\((2b2)\).
\end{proof}
\section{Heterogeneous Compute Costs}
\label{sec:heterogeneous_costs}

So far our model assumed a fixed per-operation compute cost $e$.
A more realistic
specification allows heterogeneous compute providers.\footnote{We are grateful to Noam Nisan for suggesting extending the model to consider heterogeneous supply.} We here observe that Theorems \ref{thm:MLI_characterization}  and \ref{thm:attack_cost} extend also to this setting.

\paragraph{Competitive Equilibria with Supply 
Curves for Compute.} 
We model heterogeneous compute providers via an 
\emph{upward-sloping supply curve}, captured by 
a weakly increasing marginal-cost function 
$c : [0,\infty) \to [0,\infty)$, where $c(T)$ 
denotes the marginal dollar cost of compute at 
total usage level $T := M + D + L$ (total 
GPU-operations per period); $c(\cdot)$ is the 
so-called \emph{inverse supply function} for 
compute, mapping total quantity to marginal cost.

\begin{definition}[Supply function]\label{def:supply}
A supply function is a map $c:[0,\infty)\to[0,\infty)$ that is continuous 
and weakly increasing.
\end{definition}

In the heterogeneous cost setting, the fixed 
cost $e$ appearing in the per-operation profit 
functions~\eqref{eq:profits} is replaced by an 
endogenous marginal cost $e^*$, determined in 
equilibrium by the supply curve: $e^* = c(T)$, 
where $T = M + D + L$ is the total compute 
usage. This is captured formally in 
Definition~\ref{def:competitive_equilibrium_supply}:

\begin{definition}[Competitive equilibrium with supply curve]
\label{def:competitive_equilibrium_supply}
Fix primitives $(\alpha, \gamma, R, \mathcal{D}, c)$ with $\alpha \geq 1$, 
$\gamma > 1$, $R > 0$, demand function $\mathcal{D}$, and supply function 
$c$. A \emph{competitive equilibrium with supply curve} is a tuple 
$(M, D, L, S, Q, p, P)$ satisfying all conditions of 
Definition~\ref{def:CE}, with the fixed cost $e$ replaced by the 
endogenous marginal cost
\[
e^* := c(M + D + L),
\]
i.e., the marginal cost of compute is determined by the supply curve 
evaluated at total compute usage $T = M + D + L$. A $\theta$-T/I ratio 
competitive equilibrium with supply curve is defined analogously, 
additionally requiring $\theta = PR/pQ$.
\end{definition}

We now establish an analogue of 
Theorem~\ref{thm:MLI_characterization} for the 
heterogeneous cost setting. The result shows that 
existence and uniqueness of equilibrium extend 
to this setting, and that the equilibrium is 
characterized by the same regime structure as in 
the fixed-cost case, with the fixed cost $e$ 
replaced by the endogenous marginal cost $e^*$.

\begin{theorem}[Equilibrium with heterogeneous compute costs]
\label{thm:heterogeneous}
Fix primitives $(\alpha, \gamma, R, \mathcal{D}, c)$ with $\alpha \geq 1$,
$\gamma > 1$, $R > 0$, demand function $\mathcal{D}:[0,\infty)\to[0,\infty)$ that is
continuous, weakly decreasing, with $\mathcal{D}(e) > 0$ for
all $e \geq 0$, and supply function $c$ such that $c(T) > 0$ for all $T > 0$.
Then for every $\theta > 0$, away from the knife-edge cases that
$\Delta = 0$, or $\Delta > 0$ but $\theta \in \{\theta_{\rm low},
\theta_{\rm high}\}$, there exists a unique non-degenerate $\theta$-T/I
ratio competitive equilibrium with supply curve.
\end{theorem}
\begin{proof}
We show that there exists a unique $e^* > 0$ 
satisfying the fixed point condition 
$e^* = c(T(e^*))$, where $T(e)$ denotes the 
total compute $M + D + L$ in the unique 
$\theta$-T/I ratio competitive equilibrium at 
marginal cost $e$ given by 
Theorem~\ref{thm:MLI_characterization}. Once 
$e^*$ is found, the unique non-degenerate 
$\theta$-T/I ratio competitive equilibrium with 
supply curve is given by 
Theorem~\ref{thm:MLI_characterization} applied 
at $e = e^*$, which by construction satisfies 
$e^* = c(M + D + L)$ as required by 
Definition~\ref{def:competitive_equilibrium_supply}.

\medskip
\noindent\textbf{Step 1: The total compute function $T(e)$.}
By Lemma~\ref{lem:zero_profit_identity},
\[
  e\cdot T(e) = pQ + PR.
\]
Since $\theta = PR/pQ$, we have $PR = \theta pQ$, and hence
\[
  e\cdot T(e) = pQ + \theta pQ = (1+\theta)pQ = (1+\theta)p\,\mathcal{D}(p),
\]
where we use market clearing $Q = \mathcal{D}(p)$. By
Theorem~\ref{thm:MLI_characterization}, in every regime the
equilibrium inference price satisfies $p = c_2 e$ for some constant
$c_2 > 0$ depending only on $(\alpha,\gamma,\theta)$. Substituting:
\[
  T(e) = \frac{(1+\theta)p}{e}\,\mathcal{D}(p)
       = (1+\theta)c_2\,\mathcal{D}(c_2 e)
       = c_1\,\mathcal{D}(c_2 e),
\]
where $c_1 := (1+\theta)c_2 > 0$. Since $\mathcal{D}$ is continuous
and weakly decreasing, $T(e)$ is continuous and weakly decreasing
in $e$. Strictly speaking, \(T(e)\) is defined by
Theorem~\ref{thm:MLI_characterization} only for \(e>0\). We slightly
abuse notation by using the formula
\[
T(e)=c_1\,\mathcal D(c_2e)
\]
to define its continuous extension at \(e=0\).

\medskip
\noindent\textbf{Step 2: Uniqueness of $e^*$.}
Define $G(e) := c(T(e)) - e$. Since $T(\cdot)$ 
and $c(\cdot)$ are both continuous, $G$ is 
continuous. Since $T(e)$ is weakly decreasing 
and $c(\cdot)$ is weakly increasing, the 
composition $c(T(e))$ is weakly decreasing, so 
$G$ is strictly decreasing. Hence $G(e) = 0$ 
admits at most one solution.

\medskip

\noindent\textbf{Step 3: Existence of $e^*$.}
Since $\mathcal{D}(\cdot)$ is weakly 
downward-sloping and takes values in 
$\mathbb{R}_+$, we have $\mathcal{D}(p) \leq 
\mathcal{D}(0)$ for all $p \geq 0$. Since in 
every regime $T(e) = c_1\mathcal{D}(c_2 e)$ 
for some constants $c_1, c_2 > 0$, it follows 
that
\[
0 \leq T(e) \leq c_1\mathcal{D}(0) \qquad 
\text{for all } e \geq 0.
\]
Hence $c(T(e)) \leq c(c_1\mathcal{D}(0))$ for 
all $e \geq 0$, and therefore
\[
\lim_{e \to \infty} G(e) 
= \lim_{e \to \infty} \bigl(c(T(e)) - e\bigr) 
\leq c\bigl(c_1\mathcal{D}(0)\bigr) - 
\lim_{e \to \infty} e = -\infty.
\]
Thus there exists $E > 0$ such that $G(E) < 0$. Since
$\mathcal{D}(0) > 0$, we have $T(0) = c_1\mathcal{D}(0) > 0$,
and hence $G(0) = c(T(0)) > 0$ by the assumption $c(T) > 0$
for all $T > 0$. The intermediate value theorem then implies
there exists $e^* \in (0, E)$ with $G(e^*) = 0$, i.e.,
$e^* = c(T(e^*))$, and $e^* > 0$ follows from $G(0) > 0$.
The resulting equilibrium is non-degenerate: since $e^* > 0$,
the equilibrium price satisfies $p^* = c_2 e^* > 0$, and hence
$Q = \mathcal{D}(p^*) > 0$ by the assumption $\mathcal{D}(e) > 0$
for all $e > 0$; and $S > 0$ since at least one of $M, D$ is
active in any equilibrium characterized by
Theorem~\ref{thm:MLI_characterization}.
\end{proof}

We now directly get the following corollary:
\begin{corollary}[Regime invariance]
\label{cor:regime_invariance}
Under the conditions of 
Theorem~\ref{thm:heterogeneous}, let $e^*$ 
denote the unique solution to the fixed point 
condition $e^* = c(T(e^*))$, where $T(e^*)$ 
is the total compute in the unique 
non-degenerate $\theta$-T/I ratio competitive 
equilibrium with homogeneous cost $e^*$. The 
unique non-degenerate $\theta$-T/I ratio 
competitive equilibrium with supply curve is 
then characterized by the same regime structure 
as in Theorem~\ref{thm:MLI_characterization}, 
with $e$ replaced by $e^*$ throughout. 
In particular:
\begin{enumerate}
\item The regime thresholds $\theta_{\rm low} 
:= \alpha - 1$ and $\theta_{\rm high} := 
\frac{1}{\gamma-1}$ are unchanged.
\item The equilibrium inference price and 
allocations are given by 
Theorem~\ref{thm:MLI_characterization} 
evaluated at $e = e^*$.
\end{enumerate}
\end{corollary}

\begin{proof}
By Theorem~\ref{thm:heterogeneous}, there exists 
a unique non-degenerate $\theta$-T/I ratio 
competitive equilibrium with supply curve, with 
equilibrium marginal cost $e^*$. Claims (1) 
and (2) follow immediately from 
Theorem~\ref{thm:MLI_characterization} applied 
at $e = e^*$. 
\end{proof}

\begin{remark}
For isoelastic demand $\mathcal{D}(p) = Kp^{-\varepsilon}$ and power 
supply $c(T) = e_0 \cdot T^{\beta}$, the equilibrium regime is first 
determined by $(\alpha, \gamma, \theta)$ via the thresholds 
$\theta_{\rm low} = \alpha - 1$ and $\theta_{\rm high} = 
\frac{1}{\gamma-1}$ of Corollary~\ref{cor:regime_invariance}, 
independently of $e^*$. Given the regime, with regime-specific 
constants $(c_1, c_2)$ from the proof of 
Theorem~\ref{thm:heterogeneous}, we have 
$T(e^*) = c_1 \mathcal{D}(c_2 e^*) = c_1 K (c_2 e^*)^{-\varepsilon} 
= c_1 K c_2^{-\varepsilon} (e^*)^{-\varepsilon}$, 
so the fixed point condition $e^* = e_0 \cdot T(e^*)^{\beta}$ becomes:
$$e^* = e_0 \cdot c_1^{\beta} K^{\beta} c_2^{-\varepsilon\beta} \cdot 
(e^*)^{-\varepsilon\beta},$$
and collecting terms gives $(e^*)^{1+\varepsilon\beta} = 
e_0 \cdot c_1^{\beta} K^{\beta} c_2^{-\varepsilon\beta}$, with unique 
solution:
\[
e^* = \left(e_0 \cdot c_1^{\beta} K^{\beta} 
c_2^{-\varepsilon\beta}\right)^{\frac{1}{1+\varepsilon\beta}}.
\]
All equilibrium prices and allocations then follow from 
Corollary~\ref{cor:regime_invariance}.
\end{remark}

\paragraph{Attack cost.}
In the heterogeneous cost setting, the opportunity cost of
redirecting one unit of compute from any active miner is $e^*$,
regardless of whether the miner is marginal or not. To
see this: a miner with physical cost $c(t) \leq e^*$ earns
equilibrium rent $e^* - c(t)$, so the attacker must compensate
them for both their physical cost $c(t)$ and their foregone rent
$e^* - c(t)$, totaling $e^*$. This justifies extending
Definition~\ref{def:cost50} to the heterogeneous setting by
replacing $e$ with $e^*$:

\begin{definition}[Cost of a $50\%$ attack with supply curve]
\label{def:attack_cost_supply}
Let $\Pi = (M, D, L, S, Q, p, P)$ be a competitive equilibrium
with supply curve $c(\cdot)$ with
equilibrium marginal cost $e^* := c(M+D+L)$. The \emph{economic
cost of a $50\%$ attack} is:
\[
\mathrm{Cost}_{50\%}(\Pi) := \min\!\left\{
\frac{e^* S}{2},\;
\frac{\gamma S}{2}\!\left(e^* - \frac{p}{\alpha}\right)
\right\}.
\]
\end{definition}

This definition models an attacker who obtains attack capacity by
redirecting compute already active in equilibrium. 

\begin{corollary}[Attack cost with supply curve]
\label{cor:attack_cost_supply}
Under the conditions of Corollary~\ref{cor:regime_invariance},
the economic cost of a $50\%$ attack
 equals $P^*R/2$.
\end{corollary}

\begin{proof}
By Corollary~\ref{cor:regime_invariance}, the unique competitive
equilibrium with supply curve is characterized by the same
conditions as Theorem~\ref{thm:MLI_characterization} evaluated
at $e = e^*$. In particular, the free-entry conditions
$\pi_Y \leq 0$ hold at $e = e^*$, and the zero-profit conditions
$\pi_Y = 0$ for active $Y$ hold at $e = e^*$. Since the
opportunity cost of redirecting any active miner's compute is
$e^*$ per unit (as argued above), Theorem~\ref{thm:attack_cost}
applies verbatim at $e = e^*$, giving
$\mathrm{Cost}_{50\%} = P^*R/2$.
\end{proof}

\section{Acknowledgments}

I am very grateful to Erez Badash, Ohad Klein, Jacob Leshno, Noam Nisan, Ilan Komargodski, Daniel Raskin and Omri Weinstein for helpful discussions and insightful
comments that significantly improved this work. Most notably, I am very grateful to Noam Nisan and Jacob Leshno for several comments that significantly improved the presentation of these results, and to Ohad Klein for providing detailed feedback on some of the proofs.

Large language models (primarily ChatGPT and Claude, occasionally Gemini) were used extensively during the development of this project,
including for discussing ideas, expanding proof sketches into initial
draft write-ups, implementing the pricing model in Python, and assisting
with editing.
All technical claims, proofs, errors, and interpretations remain my own.
\bibliographystyle{plain}
\bibliography{duplex}

\appendix
\section{A Simple Monetary Interpretation of the T/I Ratio}
\label{sec:monetary_foundation}
This appendix shows that the token/inference ratio $\theta = \frac{P R}{p Q}$ admits a simple monetary interpretation in a proof-of-work economy in which the native token serves as a settlement medium for inference. In particular, $\theta$ can be viewed as an index summarizing underlying settlement intensity, velocity, and hoarding behavior, and is independent of technological parameters governing the PoW blockchain (and in particular, duplex efficiency parameters  $\alpha,\gamma$). This perspective explains why it may be reasonable to treat $\theta$ as a behavioral constant when analyzing how changes in technology reshape equilibrium structure. 

We note that the derivation is not specific to proof-of-useful-work or inference markets: it applies to any PoW blockchain in which
the native token serves as a settlement medium for some transactional
flow. Our derivation applies the classical equation-of-exchange logic 
\cite{Fisher1911} to the PoW setting, augmented with a simple 
stock-flow model of circulating token supply in which newly issued 
tokens partially enter circulation and a fraction is periodically 
withdrawn into hoarding.

\subsection{Settlement Float and Monetary Environment}
Consider a PoW blockchain that issues $R>0$ tokens per block
at price $P>0$ (USD/token).
Let $\mathcal E_t$ denote the dollar value of token-denominated transaction in block period $t$.

Looking forward, we will be considering 
\[
\mathcal E_t = \theta_{\rm raw} p_t Q_t,
\]
where $p_t Q_t$ is total inference expenditure at block time period $t$ and
$\theta_{\rm raw}>0$ is a settlement multiple. That is, 
the total dollar value of transactions settled in the token at block time period $t$ is a constant multiplier of the inference expenditure at time period $t$.
(Note that the  multiple $\theta_{\rm raw}$ may exceed one, allowing for additional token-denominated demand beyond inference settlement.)
For now, however, we keep the token demand 
$\mathcal E_t$ abstract.

We adopt a very simple model for the evolution and usefulness of this settlement float:
\begin{enumerate}
\item \textbf{Turnover (velocity).}
The float generates transactional expenditure according to an exchange relation
\begin{equation}\label{eq:exchange_identity_MCP}
\mathcal E_t = P_t V X_t ,
\end{equation}
where $V>0$ is the average per-block turnover rate of the \emph{settlement float} (i.e., $V$ is velocity only for tokens actually in circulation).
This mirrors Fisher's ``equation of exchange'' logic \cite{Fisher1911} used in monetary models (transactional expenditure equals price times velocity times circulating balances)
 and in crypto-asset pricing analyses such as \cite{Athey2016BitcoinPricing}.

\item \textbf{Float inflow from issuance.}
A fraction $\nu\in(0,1]$ of newly issued tokens enters the settlement float each block,
capturing the idea that miners sell a portion of rewards into transactional circulation.

\item \textbf{Net float leakage.}
A fraction $\lambda\in(0,1]$ of the float exits circulation
each block (net of any re-entry from hoarders).
We assume $\lambda>0$ to ensure existence of a stationary float.
\end{enumerate}

The settlement float thus evolves according to
\begin{equation}
\label{eq:float_law_MCP}
X_{t+1} = X_t + \nu R - \lambda X_t.
\end{equation}

\paragraph{Stationary monetary equilibria}
We turn to defining a stationary monetary equilibrium in which token
prices may vary over time but the system converges to a stationary
monetary configuration.

\begin{definition}[Stationary monetary equilibrium (settlement-float subsystem)]
\label{def:stationary_monetary_eq}
Fix issuance $R>0$ and settlement-float primitives $(\nu,\lambda,V)$
with $\nu,\lambda\in(0,1]$ and $V>0$.
A \emph{stationary monetary equilibrium} is a tuple
$\bigl(\{P_t\}_{t\ge0},\{\mathcal E_t\}_{t\ge0},\{X_t\}_{t\ge0}\bigr)$
such that:

\begin{enumerate}
\item \textbf{Exchange (settlement capacity).}
For all $t\ge0$,
\[
\mathcal E_t \;=\; P_t\,V\,X_t .
\]

\item \textbf{Settlement-float law of motion.}
For all $t\ge0$,
\[
X_{t+1} \;=\; X_t + \nu R - \lambda X_t .
\]

\item \textbf{Stationarity.}
There exist constants $(P,\mathcal E,X)$ such that
\[
P_t = P,\qquad
\mathcal E_t = \mathcal E,\qquad
X_t = X
\qquad \text{for all sufficiently large } t .
\]
\end{enumerate}
We refer to $(P,\mathcal E,X)$ as the \emph{stationary triple}.
\end{definition}

\begin{proposition}[Stationary settlement float implies linear pricing]
\label{prop:settlement_linear}
Fix $(R,\nu,\lambda,V)$ and let
$(\{P_t\}_{t\ge0}$, $\{\mathcal E_t\}_{t\ge0}$, $\{X_t\}_{t\ge0})$
be a stationary monetary equilibrium
with stationary triple $(P,\mathcal E,X)$.
Then the issuance value satisfies
\begin{equation}\label{eq:linear_PR_general}
PR \;=\; \frac{\lambda}{\nu V}\,\mathcal E .
\end{equation}
\end{proposition}

\begin{proof}
By stationarity of the settlement-float law, for sufficiently large $t$,
$X_{t+1}=X_t=X$ implies
\[
0=\nu R-\lambda X,
\qquad\text{hence}\qquad
X=\frac{\nu}{\lambda}R .
\]
Under stationarity, for sufficiently large $t$, we also have $P_t= P$ and
$\mathcal E_t = \mathcal E$.
Substituting the stationary settlement float into the exchange relation
$\mathcal E = PVX$ gives
\[
\mathcal E
= PV\left(\frac{\nu}{\lambda}R\right).
\]
Rearranging yields
\[
PR=\frac{\lambda}{\nu V}\,\mathcal E ,
\]
which proves \eqref{eq:linear_PR_general}.
\end{proof}

\subsection{Implication for Equilibrium Analysis}

We now show how to apply Proposition \ref{prop:settlement_linear} to the setting of this paper.
Let $pQ$ denote the dollar value of useful-computation
(inference) spending. Suppose that a \emph{settlement multiple} $\theta_{\rm raw}$
of this spending is conducted in the native token. That is, the total
dollar value of token-denominated transactions per block period equals
\[
\mathcal E = \theta_{\rm raw} \, pQ,
\]
We now show that under the above stationary settlement float model, this pins the market coupling parameter $\theta$.

\begin{corollary} [Stationary implies fixed T/I ratio]
\label{thm:stationary_reduction}
Fix primitives $(e,\alpha,\gamma,R)$ and demand $\mathcal D$.
Let $(M,D,L,S,Q,p,P)$ be a competitive equilibrium.
Assume this equilibrium can be extended with a stationary monetary equilibrium with stationary triple 
$(P,{\mathcal E},X)$ for some settlement-float primitives $(\nu,\lambda,V)$ such that
$\mathcal E \;=\; \theta_{\mathrm{raw}}\,pQ$
for fixed $\theta_{\mathrm{raw}}>0$.
Define
\[
\theta \;:=\; \frac{\lambda}{\nu V}\,\theta_{\mathrm{raw}},
\]
Then $(M,D,L,S,Q,p,P)$ is a $\theta$-T/I ratio competitive equilibrium.
\end{corollary}
\begin{proof}
Let $(M,D,L,S,Q,p,P)$ be a competitive equilibrium, 
and suppose it admits an extension to a stationary monetary equilibrium
as in the statement.
By Proposition~\ref{prop:settlement_linear}, the stationary monetary equilibrium
satisfies:
\[PR \;=\; \frac{\lambda}{\nu V}\,\mathcal E = \theta\,pQ
\]
where the effective coefficients
\(
\theta=\frac{\lambda}{\nu V}\theta_{\mathrm{raw}}
\)
depend only on the settlement–float primitives
$(\nu,\lambda,V)$ and the raw demand parameters, and therefore are
independent of the realized allocation $(M,D,L,S$, $Q,p,P)$.
Since, by assumption $(M,D,L,S,Q,p,P)$ satisfies the notion of a competitive equilibrium, it follows that 
$(M,D,L,S,Q,p,P)$ satisfies the definition of a $\theta$-T/I ratio competitive equilibrium.
\end{proof}

\paragraph{Interpretation of $\theta$.}
The market coupling parameter $\theta$ can thus 
be thought of as a steady-state index aggregating:
\begin{itemize}
\item settlement intensity ($\theta_{\rm raw}$),
\item float inflow from issuance ($\nu$),
\item hoarding behavior ($\lambda$),
\item turnover ($V$).
\end{itemize}
Importantly, $\theta$ depends only on monetary primitives and is
independent of technology (i.e., parameters $\alpha,\gamma$).

\paragraph{A note on the stability of $\theta$.}
We have treated $V$, $\lambda$, and $\nu$ as 
fixed empirical parameters. In principle, these 
could all vary with economic conditions: velocity 
$V$ may shift with transaction patterns, and the 
hoarding parameters $\lambda$ and $\nu$ may 
respond to price expectations. We nonetheless 
maintain the stability assumption as a 
first-order approximation, following the 
equation-of-exchange framework of 
Fisher~\cite{Fisher1911} and the quantity theory 
tradition of Friedman~\cite{friedman1956}, and 
view the robustness of this assumption to richer 
behavioral models as an interesting direction for 
future work.
\section{Halving Schedules and Transaction Fees}
\label{app:extensions}

\subsection{Halving Schedules}
\label{app:halving}

Our main model considers a single epoch with a constant per-block
reward \(R\) and per-period cost \(e\). Bitcoin-style blockchains
instead unfold over multiple epochs, following a halving schedule
in which the block reward is periodically reduced. We show here
that, under the assumption that miners discount future rewards and
costs at rate \(\rho \in (0,1)\), commit to mining indefinitely, and
all equilibrium variables \((p,S,P,Q)\) are stationary across
periods, any issuance schedule \(\{R_t\}_{t\ge 0}\) can be
summarized by a single annuity-equivalent reward \(\bar R\).\footnote{All dollar-denominated quantities are interpreted in real terms; equivalently, \(\rho\) is the real discount factor, net of ordinary price inflation.} The
resulting infinite-horizon equilibrium problem is then equivalent
to the one-period model studied in the main text with block reward
\(\bar R\).

Suppose miners evaluate future rewards and costs using a per-period
discount factor \(\rho \in (0,1)\). A miner who commits to mining
indefinitely values the entire future reward stream at its net
present value:
\[
R_{\rm eff}
:=
\sum_{t=0}^{\infty}\rho^tR_t.
\]
We refer to \(R_{\rm eff}\) as the \emph{effective block reward}.

As a concrete example, consider Bitcoin's halving schedule
\[
R_t
=
R\cdot 2^{-\lfloor t/H\rfloor},
\]
where \(H\) is the halving interval in blocks. Grouping terms
within each halving period:
\begin{align*}
R_{\rm eff}
&=
\sum_{t=0}^{\infty}
\rho^t
\cdot
R\cdot
2^{-\lfloor t/H\rfloor}
\\
&=
R
\sum_{k=0}^{\infty}
2^{-k}
\sum_{t=kH}^{(k+1)H-1}
\rho^t
\\
&=
R
\sum_{k=0}^{\infty}
2^{-k}
\rho^{kH}
\sum_{j=0}^{H-1}
\rho^j
\\
&=
R
\left(
\sum_{j=0}^{H-1}
\rho^j
\right)
\sum_{k=0}^{\infty}
\left(
\frac{\rho^H}{2}
\right)^k
\\
&=
R\cdot
\frac{1-\rho^H}{1-\rho}
\cdot
\frac{1}{1-\rho^H/2},
\end{align*}
where the geometric series converges since
\(\rho<1\) implies \(\rho^H/2<1\). Thus \(R_{\rm eff}\) is finite
for any positive discount rate.
Define the annuity-equivalent reward
\[
\bar R := (1-\rho)R_{\rm eff}.
\]
\(\bar R\) is precisely the ``constant per-period" reward whose
discounted present value equals \(R_{\rm eff}\), since
\[
\sum_{t=0}^{\infty}\rho^t\bar R
=
\frac{\bar R}{1-\rho}
=
R_{\rm eff}.
\]

We now show that \(\bar R\) is the only quantity relevant for
equilibrium. Under stationarity, the discounted infinite-horizon
profits of solo mining, solo inference, and duplex compute are:
\[
\Pi_M
=
\sum_{t=0}^{\infty}
\rho^t
\left(
\frac{PR_t}{S}
-
e
\right)
=
\frac{PR_{\rm eff}}{S}
-
\frac{e}{1-\rho},
\]
\[
\Pi_L
=
\sum_{t=0}^{\infty}
\rho^t
(p-e)
=
\frac{p-e}{1-\rho},
\]
and
\[
\Pi_D
=
\sum_{t=0}^{\infty}
\rho^t
\left(
\frac{PR_t}{\gamma S}
+
\frac{p}{\alpha}
-
e
\right)
=
\frac{PR_{\rm eff}}{\gamma S}
+
\frac{p}{\alpha(1-\rho)}
-
\frac{e}{1-\rho}.
\]

Multiplying all profits by the positive constant \(1-\rho\) does
not affect their signs and therefore does not affect any
equilibrium condition. The zero-profit conditions are therefore
equivalent to
\[
(1-\rho)\Pi_M
=
\frac{P\bar R}{S}
-
e,
\]
\[
(1-\rho)\Pi_L
=
p-e,
\]
and
\[
(1-\rho)\Pi_D
=
\frac{P\bar R}{\gamma S}
+
\frac{p}{\alpha}
-
e.
\]

These are exactly the profit functions of the one-period model
(Equation~\ref{eq:profits}) with block reward \(\bar R\) in place
of \(R\). Consequently, the infinite-horizon model is equivalent
to the one-period model of the main text after replacing
\[
R
\longrightarrow
\bar R
=
(1-\rho)R_{\rm eff}.
\]

For Bitcoin's halving schedule, the annuity-equivalent reward is
therefore
\[
\bar R
=
(1-\rho)R_{\rm eff}
=
R\cdot
\frac{1-\rho^H}{1-\rho^H/2}.
\]

In particular, the equilibrium characterization of
Theorem~\ref{thm:MLI_characterization}, the regime thresholds
\(\theta_{\rm low}\) and \(\theta_{\rm high}\), and the attack-cost
result of Theorem~\ref{thm:attack_cost} remain unchanged under
general discounted issuance schedules after replacing \(R\) by
\(\bar R\). Note that the T/I ratio is defined with respect to the
annuity-equivalent reward: $\theta = P\bar{R}/(pQ)$. Under
a halving schedule, $\theta$ should therefore be interpreted
as $P\bar{R}/(pQ)$ rather than $PR_t/(pQ)$ for any specific
period $t$.

\subsection{Transaction Fees}
Our main model abstracts from transaction fees, which are 
payments made by blockchain users to miners in addition to 
the block reward. We show here that transaction fees are 
naturally accommodated within the monetary framework of 
Appendix~\ref{sec:monetary_foundation}, and that their effect reduces 
to a rescaling of the T/I ratio into an ``effective" T/I ratio $\theta^+$, 
leaving all equilibrium characterizations unchanged.

Following the monetary foundation in Appendix~\ref{sec:monetary_foundation},
let $\mathcal E$ denote the token-denominated settlement flow per period.
In the model, $\mathcal E=\theta_{\rm raw}pQ$. Suppose that transaction fees amount to a fraction $\phi\ge0$ of the
settlement flow.\footnote{In practice, not all settlement may occur
on-chain. However, if a constant fraction of settlement activity takes
place on-chain and transaction fees are proportional to on-chain
settlement, then fees are also proportional to total settlement flow.}
Miner revenue therefore becomes
\[
PR+\phi\mathcal E .
\]
Dividing by $pQ$ yields the \emph{effective} token/inference ratio once fees are included,
which we denote by $\theta^+$:
\[
\theta^+ := \frac{PR+\phi\mathcal E}{pQ}
           = \frac{PR}{pQ}+\phi\frac{\mathcal E}{pQ}.
\]
Using $\frac{PR}{pQ}=\theta$ and $\mathcal E=\theta_{\rm raw}pQ$ gives
\[
\theta^+ = \theta+\phi\theta_{\rm raw}.
\]
Since Appendix~\ref{sec:monetary_foundation} shows that
\[
\theta=\frac{\lambda}{\nu V}\theta_{\rm raw},
\]
equivalently $\theta_{\rm raw}=\frac{\nu V}{\lambda}\theta$, it follows that
\[
\theta^+=\left(1+\phi\frac{\nu V}{\lambda}\right)\theta.
\]
Thus incorporating transaction fees is equivalent to replacing $\theta$
with $\theta^+=\left(1+\phi\frac{\nu V}{\lambda}\right)\theta$. Since the
equilibrium characterization of Theorem~\ref{thm:MLI_characterization}
depends only on $\theta$, the characterizations of the main text apply directly
with $\theta^+$ in place of $\theta$.

\end{document}